\documentclass[a4paper,UKenglish,cleveref, autoref, thm-restate]{lipics-v2021}

\usepackage{amsmath, amssymb, amsfonts}
\usepackage{todonotes}
\usepackage{amsmath,amsfonts,amssymb,amsthm}
\usepackage{graphicx,color}
\usepackage{boxedminipage}
\usepackage{algpseudocode}
\usepackage{framed}
\usepackage{xspace}
\usepackage{todonotes}
\usepackage{footnote}
\usepackage{ifthen}
\newboolean{shortver}

\setboolean{shortver}{false}
\usepackage{mathtools}
\usepackage{mathrsfs}
\usepackage{todonotes}
\usepackage{footnote}
\usepackage{rotating}
\usepackage{adjustbox}

\newcommand{\gcon}{G_{\texttt{conn}}}
\newcommand{\gspar}{G_{\texttt{spar}}}
\newcommand{\gcov}{G_{\texttt{cov}}}
\newcommand{\gconf}{G_{\texttt{conf}}}

\newcommand{\hcon}{H_{\texttt{conn}}}
\newcommand{\hcov}{H_{\texttt{cov}}}

\usepackage{thmtools}
\usepackage{thm-restate}

\newcommand{\cR}{{\mathcal{R}}}

\usepackage{mdframed}

\newcommand{\omitted}{{\ifthenelse{\boolean{shortver}}{$\spadesuit$}{\!}}}

\usepackage{tcolorbox}
\tcbuselibrary{skins}

\usepackage{xspace}
\newcommand{\cO}{\mathcal{O}}
\newcommand{\no}{\texttt{No}\xspace} 
\newcommand{\yes}{\texttt{Yes}\xspace}

\makeatletter
\newcommand*{\rom}[1]{\expandafter\@slowromancap\romannumeral #1@}
\makeatother

\usepackage{mathtools}

\DeclareMathOperator*{\argmax}{arg\,max}

\newcommand{\fpt} {{\sf FPT}\xspace}
\newcommand{\nph} {{\sf NP}-hard}

\newcommand{\woh} {{\sf W}$[1]$-hard\xspace}

\newcommand{\npconp}{{\sf NP} $\subseteq$ {{\sf co-NP/poly}}\xspace}

\definecolor{anti-flashwhite}{rgb}{0.95, 0.95, 0.96}

\newcommand{\OO}{{\mathcal{O}}\xspace}

\newcommand{\cover}{\textbf{{\sf{cover}}}\xspace}

\newtheorem{reduction rule}{Reduction Rule}
\newtheorem*{reduction rule*}{Reduction Rule}

 \usepackage{xifthen}
 \usepackage{tabularx}
\usepackage{tikz} 
 \usetikzlibrary{calc}
 \usepackage{thm-autoref}
\usepackage{enumitem}

\newcommand{\OPT}{\texttt{OPT}}

\newcommand{\os}{\mathcal{O}^\star}
\newcommand{\Oh}{\mathcal{O}}

\newcommand{\pcds}{{\sc PCDS}}

\newlength{\RoundedBoxWidth}
\newsavebox{\GrayRoundedBox}
\newenvironment{GrayBox}[1]%
   {\setlength{\RoundedBoxWidth}{.93\textwidth}
    \def\boxheading{#1}
    \begin{lrbox}{\GrayRoundedBox}
       \begin{minipage}{\RoundedBoxWidth}}%
   {   \end{minipage}
    \end{lrbox}
    \begin{center}
    \begin{tikzpicture}%
       \node(Text)[draw=black!20,fill=white,rounded corners,%
             inner sep=2ex,text width=\RoundedBoxWidth]%
             {\usebox{\GrayRoundedBox}};
        \coordinate(x) at (current bounding box.north west);
        \node [draw=white,rectangle,inner sep=3pt,anchor=north west,fill=white] 
        at ($(x)+(6pt,.75em)$) {\boxheading};
    \end{tikzpicture}
    \end{center}}     

\newenvironment{defproblemx}[2][]{\noindent\ignorespaces%
                                \FrameSep=6pt%
                                \parindent=0pt%
                \vspace*{-1.5em}
                \ifthenelse{\isempty{#1}}{%
                  \begin{GrayBox}{\textsc{#2}}%
                }{%
                  \begin{GrayBox}{\textsc{#2}  parameterized by~{#1}}%
                }
                \begin{tabular*}{\textwidth}{@{\hspace{.1em}} >{\itshape} p{1.8cm} p{0.8\textwidth} @{}}%
            }{
                \end{tabular*}%
                \end{GrayBox}%
                \ignorespacesafterend
            }  

\newcommand{\lr}[1]{\left(#1\right)}
\newcommand{\lrc}[1]{\left\{#1\right\}}
\newcommand{\lrsq}[1]{\left[#1\right]}
\DeclarePairedDelimiter{\LRb}{\lbrace}{\rbrace}
\newcommand{\LR}{\LRb*}

\newcommand{\blue}[1]{{\color{blue}#1}}

\newcommand{\fptas}{{\sf EPAS}\xspace}
\newcommand{\cU}{\mathcal{U}}
\newcommand{\cF}{\mathcal{F}}

\newcommand{\cI}{\mathcal{I}}

\newcommand{\pcrbdsfull}{{\sc Partial Connected Red-Blue Dominating Set}\xspace}
\newcommand{\pcrbdsshort}{{\sc PartialConRBDS}\xspace}
\newcommand{\stpcrbdsfull}{{\sc Steiner Partial Connected RBDS}\xspace}
\newcommand{\stpcrbdsshort}{{\sc SteinerPartialConRBDS}\xspace}

\newcommand{\ball}{\textnormal{\textsf{Ball}}\xspace}

\newcommand{\algpas}{\texttt{PartialConRBDS-PAS}\xspace}

\newcommand{\extalgpas}{\texttt{InnerPartialConRBDS-PAS}\xspace}

\newcommand{\pcmcshort}{{\sc PCMC}\xspace}

\newcommand{\kdd}{\ensuremath{K_{d,d}}}

\renewcommand{\deg}{{\sf deg}\xspace}
\newcommand{\cov}{\textnormal{{\textsf{cov}}}\xspace}

\newcommand{\cC}{\mathcal{C}}

\newcommand{\kddfree}{{\textnormal{\textsf{$K_{d,d}$-free}}}\xspace}

\newcommand{\cut}{\textnormal{\textsf{cut}}\xspace}

\usepackage{relsize}

\newtoggle{includeFootnote}
\toggletrue{includeFootnote}

\newcommand{\ti}[1]{\textcolor{Red}{TI: #1}}

\newcommand{\runtimea}{2^{\Oh(kd(k^2+\log d))}}

\newcommand{\dist}{\textnormal{\textsf{dist}}}
\renewcommand{\cover}{\texttt{cover}}
\newcommand{\cH}{\mathcal{H}}
\newcommand{\cD}{\mathcal{D}}

\usepackage{soul}

\newcommand{\p}{\mathsf{P}}
\newcommand{\propone}{\hyperref[prop:conn]{($\p1$)}\xspace}
\newcommand{\proptwo}{\hyperref[prop:dom]{($\p2$)}\xspace}

\definecolor{newgray}{gray}{0.25} 

\newcommand{\Vred}{V_{\mathrm{purple}}}
\newcommand{\Vblue}{V_{\mathrm{green}}}

\makeatletter
\renewcommand{\paragraph}{%
  \@startsection{paragraph}{4}%
  {\z@}{1ex \@plus 1ex \@minus .2ex}{-0.5em}%
  {\normalfont\normalsize\bfseries}%
}
\makeatother 
\newcommand{\bcdshort}{{\sc BudgetedCDS}\xspace}

\usepackage{amsmath,amsfonts,amssymb,amsthm}
\usepackage{pifont}
\usepackage{amsthm}
\usepackage{graphicx}
\usepackage{colortbl}
\usepackage{boxedminipage}
\usepackage{multirow,multicol}
\usepackage{framed}
\usepackage{mathtools}
\usepackage{float}
\usepackage[boxed,lined,linesnumbered, ruled]{algorithm2e}
\SetKwInput{KwInput}{Input}
\SetKwInput{KwOutput}{Output}
\SetKwInput{KwInvariant}{Invariant}
\DontPrintSemicolon
\SetKwFor{ForEach}{for each}{do}{end for}
\usepackage{thmtools}
\usepackage{thm-restate}
\usepackage{xspace}
\usepackage{todonotes}
\usepackage{footnote}
\usepackage{tcolorbox}

\definecolor{anti-flashwhite}{rgb}{0.95, 0.95, 0.96}

\theoremstyle{plain}

\usepackage{todonotes}
\usepackage{footnote}

\colorlet{mix}{red!50!black}
\definecolor{BrickRed}{rgb}{0.8, 0.25, 0.33}
\definecolor{PineGreen}{rgb}{0.0, 0.47, 0.44}
\definecolor{Brown}{rgb}{0.59, 0.29, 0.0}
\definecolor{bblue}{rgb}{0.1,0.4,0.6}
\definecolor{ForestGreen}{rgb}{0.1333,0.5451,0.1333}

\usepackage{hyperref}
\hypersetup{
	pdfnewwindow=true,      
	colorlinks=true,       
	linkcolor=mix,          
	citecolor=ForestGreen,        
	urlcolor=blue          
}
\hideLIPIcs

\usepackage{xifthen}
\usepackage{tabularx}
\newtheorem{property}[theorem]{Property}
\Crefname{property}{Property}{property}
\Crefname{claim}{Claim}{Claims}
\Crefname{observation}{Observation}{Observations}
\usepackage{cite}

\title{FPT Approximations for Connected Maximum Coverage}
\titlerunning{FPT Approximations for Connected Maximum Coverage}

\author{Tanmay Inamdar}
{Indian Institute of Technology Jodhpur, Jodhpur, India}{taninamdar@gmail.com}{https://orcid.org/0000-0002-0184-5932}{The author is supported by Anusandhan National Research Foundation (ANRF), grant number
ANRF/ECRG/2024/004144/PMS, and IIT Jodhpur Research Initiation Grant, grant number I/RIG/TNI/20240072. }

\author{Satyabrata Jana}{University of Warwick, UK}{satyamtma@gmail.com}{https://orcid.org/0000-0002-7046-0091}{Supported by the  Engineering and Physical Sciences Research Council (EPSRC) via the project MULTIPROCESS (grant no. EP/V044621/1)}

\author{Madhumita Kundu}{University of Bergen, Norway}{kundumadhumita.134@gmail.com}{https://orcid.org/0000-0002-8562-946X}{}

\author{Daniel Lokshtanov}{Department of Computer Science, University of California Santa Barbara, Santa Barbara, USA}{daniello@ucsb.edu}{https://orcid.org/0000-0002-3166-9212}{}

\author{Saket Saurabh}{The Institute of Mathematical Sciences, HBNI, Chennai, India  \and University of Bergen, Norway }{saket@imsc.res.in}{https://orcid.org/0000-0001-7847-6402}{The author is supported by the European Research Council (ERC) under the European Union's Horizon 2020 research and innovation programme (grant agreement No. 819416); and he also acknowledges the support of Swarnajayanti Fellowship grant DST/SJF/MSA-01/2017-18.}

\author{Meirav Zehavi}{Ben-Gurion University of the Negev, Beersheba, Israel}{zehavimeirav@gmail.com}{https://orcid.org/0000-0002-3636-5322}{The author is supported by Israel Science Foundation (ISF), grant no. 1470/24, and by European Research Council (ERC) grant no. 101039913 (PARAPATH).}

\Copyright{Tanmay Inamdar, Satyabrata Jana, Madhumita Kundu, Daniel Lokshtanov, Saket Saurabh, Meirav Zehavi}

\authorrunning{T.~Inamdar, S.~Jana, M.~Kundu, D.~Lokshtanov, S.~Saurabh, M~.Zehavi}



\nolinenumbers 

\EventLongTitle{17th  Innovations in Theoretical Computer Science (ITCS 2026)}
\EventShortTitle{ITCS 2026}
\EventAcronym{ITCS}
\EventYear{2026}
\EventDate{January 27 -- January 30, 2026}
\EventLocation{}
\EventLogo{}
\SeriesVolume{}
\ArticleNo{}

\begin{document}

\maketitle


\begin{abstract}

We revisit connectivity-constrained coverage through a unifying model, \pcrbdsfull\ (\pcrbdsshort). 
Given a bipartite graph $G=(R\cup B,E)$ with red vertices $R$ and blue vertices $B$, an auxiliary connectivity graph $G_{\mathrm{conn}}$ on $R$, and integers $k,t$, the task is to find a set $S\subseteq R$ with $|S|\le k$ such that $G_{\mathrm{conn}}[S]$ is connected and $S$ dominates at least $t$ blue vertices.
This formulation captures connected variants of {\sc Maximum Coverage}~[Hochbaum--Rao, Inf.\ Proc.\ Lett., 2020; D'Angelo--Delfaraz, AAMAS 2025], {\sc Partial Vertex Cover}, and {\sc Partial Dominating Set}~[Khuller et al., SODA 2014; Lamprou et al., TCS 2021] via standard encodings.

\paragraph*{Limits to parameterized tractability.}
\pcrbdsshort\ is $\mathsf{W[1]}$-hard parameterized by $k$ even under strong restrictions: it remains hard when $G_{\mathrm{conn}}$ is a clique or a star and the incidence graph $G$ is $3$-degenerate, or when $G$ is $K_{2,2}$-free.

\paragraph*{Inapproximability.} For every $\varepsilon>0$, there is no polynomial-time $(1,\,1-\frac{1}{e}+\varepsilon)$-approximation unless $\mathsf{P}=\mathsf{NP}$. Moreover, under \textsf{ETH}, no algorithm running in $f(k)\cdot n^{o(k)}$ time achieves an $g(k)$-approximation for $k$ for any computable function $g(\cdot)$, or for any $\varepsilon > 0$, a $(1-\frac{1}{e}+\varepsilon)$-approximation for $t$.

\paragraph*{Graphical special cases.} {\sc Partial Connected Dominating Set} is $\mathsf{W[2]}$-hard parameterized by $k$ and inherits the same \textsf{ETH}-based $f(k)\cdot n^{o(k)}$ inapproximability bound as above; {\sc Partial Connected Vertex Cover} is $\mathsf{W[1]}$-hard parameterized by $k$.

These hardness boundaries delineate a natural ``sweet spot'' for study: within appropriate structural restrictions on the incidence graph, one can still aim for fine-grained (FPT) approximations.

\begin{sloppypar}
\paragraph*{Our algorithms.}
We solve \pcrbdsshort exactly by reducing it to {\sc Relaxed Directed Steiner Out-Tree} in time $(2e)^t \cdot n^{\Oh(1)}$. 
For biclique-free incidences (i.e., when $G$ excludes $K_{d,d}$ as an induced subgraph), we obtain two complementary parameterized schemes:
\begin{itemize}
  \item An Efficient Parameterized Approximation Scheme (EPAS) running in time $2^{\Oh(k^2 d/\varepsilon)}\cdot n^{\Oh(1)}$ that either returns a connected solution of size at most $k$ covering at least $(1-\varepsilon)t$ blue vertices, or correctly reports that no connected size-$k$ solution covers $t$; and
  \item A Parameterized Approximation Scheme (PAS) running in time $2^{\Oh(kd(k^2+\log d))}\cdot n^{\Oh(1/\varepsilon)}$ that either returns a connected solution of size at most $(1+\varepsilon)k$ covering at least $t$ blue vertices, or correctly reports that no connected size-$k$ solution covers $t$.
\end{itemize}
Together, these results chart the boundary between hardness and FPT-approximability for connectivity-constrained coverage.
\end{sloppypar}

\end{abstract}

\newpage 

\thispagestyle{empty}
\setcounter{tocdepth}{2} 
\tableofcontents
\pagenumbering{arabic}
\newpage
\setcounter{page}{1}
\section{Introduction} \label{sec:intro}

\begin{sloppypar}
In the last decade, \fpt approximation has enjoyed a tremendous amount of success in obtaining near-optimal approximations for {\sf NP}-hard optimization problems that resist polynomial-time approximations, as well as exact \fpt algorithms. Indeed, this paradigm has led to the best-known approximation guarantees for many fundamental problems, including {\sc Min $k$-Cut} \cite{DBLP:journals/siamcomp/LokshtanovSS24}, {\sc $k$-Median/$k$-Means} clustering~\cite{DBLP:conf/icalp/Cohen-AddadG0LL19}, {\sc $k$-Edge Separator} \cite{DBLP:conf/soda/GuptaLLM019}, {\sc Balanced Separator}~\cite{DBLP:conf/stoc/FeigeM06}, to mention a few. The field’s recent momentum owes much to the emergence of \emph{accompanying hardness-of-approximation} frameworks in the FPT setting, which delineate the achievable frontier and indicate when to stop pursuing better ratios or faster schemes~\cite{DBLP:journals/jacm/SLM19,ChalermsookCKLM17,Wlodarczyk20,ChenL19,Lin18,Lin19,DBLP:conf/stoc/GuruswamiLRS025,DBLP:conf/stoc/GuruswamiLRS024,DBLP:conf/stoc/BafnaSM25}

Possibly second only to clustering and treewidth, the strongest testament to the success of the \fpt approximation paradigm lies in its applications to {\sc Max Coverage} and its variants. In the classical {\sc Max Coverage} problem, we are given a set system $(\cU, \cF)$ and a positive integer $k$, where $\cU$ is a universe of $n$ elements and $\cF$ is a family of $m$ subsets of $\cU$. The goal is to select a subfamily $\cF' \subseteq \cF$ of size $k$ that covers as many elements of $\cU$ as possible. A well-known greedy algorithm achieves a tight $(1 - 1/e)$-approximation~\cite{DBLP:conf/stoc/DinurS14,hochbaum1998analysis}, and the problem is known to be {\sf W}[2]-hard when parameterized by $k$~\cite{DBLP:journals/siamcomp/DowneyF95}.  

Recent results have established that this approximation factor remains tight even if one allows \fpt running time. Specifically, Manurangsi~\cite{DBLP:conf/soda/Manurangsi20} showed that assuming {\sf Gap-ETH}, no algorithm can achieve an approximation ratio better than $(1-1/e)$ in time $f(k) \cdot (|\cU|+|\cF|)^{o(k)}$, even with the promise that there exist $k$ sets that cover the entire universe. Very recently, this result was strengthened by
Guruswami et al.~\cite{DBLP:conf/stoc/GuruswamiLRS025} by weakening the assumption to {\sf ETH}. Despite this barrier, significantly better results can be obtained for {\sc Max Coverage} and its variants when the underlying set system has additional structural properties. This line of research was initiated by Marx~\cite{DBLP:journals/cj/Marx08}, who gave an \fpt $(1-\varepsilon)$-approximation algorithm running in time $\os(2^{\Oh(k^2/\varepsilon)})$\footnote{We use $\os(\cdot)$ to suppress polynomial factors in the input size, i.e., $\os(T) = T \cdot |I|^{\Oh(1)}$, where $|I|$ denotes the input size.} for the {\sc Partial Vertex Cover} (PVC) problem. Here, the task is to select $k$ vertices in a graph that cover as many edges as possible—a special case of {\sc Max Coverage}. Subsequent works extended this to more general settings. In particular, for set systems where each element appears in at most $d$ sets ($d=2$ for PVC), the running time was improved to $\os((d/\varepsilon)^{\Oh(k)})$~\cite{DBLP:journals/jair/SkowronF17,DBLP:conf/soda/Manurangsi19}. More recently, Jain et al.~\cite{DBLP:conf/soda/0001KPSS0U23} further generalized these results to \kddfree set systems, where no $d$ sets share $d$ common elements.  


Building on the success of \fpt approximation for the \emph{vanilla} {\sc Max Coverage}, attention has shifted to \emph{coverage with additional constraints}, a class of problems previously studied in the classical polynomial-time approximation regime. Many of these classical results have been extended to obtain stronger \fpt approximations for coverage problems with fairness~\cite{DBLP:conf/icalp/00020LS0U24}, matroid~\cite{DBLP:conf/esa/Sellier23,DBLP:conf/icalp/00020LS0U24}, and capacity~\cite{DBLP:conf/soda/LokshtanovS0S025} constraints, when the underlying set systems are structurally well-behaved.  


\paragraph*{Connectivity Constraints.} 
Among the constrained variants, \emph{connectivity-constrained coverage} has received considerable attention in the classical setting. Motivated from applications in sensor and social networks, Khuller et al.~\cite{DBLP:journals/siamdm/KhullerPS20} introduced {\sc Connected Partial Dominating Set}, where, given a graph $G$, the goal is to find a minimum size connected subset $S \subseteq V(G)$ that dominates at least $t$ vertices. For this problem, they designed a polynomial-time $\Oh(\log \Delta)$-approximation, where $\Delta$ is the maximum degree in $G$. For the complementary problem, called {\sc Connected Budgeted Dominating Set}---where the goal is to find a connected vertex-subset of size at most $k$ that dominates the maximum number of vertices---they designed an $\lr{\frac{1}{12}(1-\frac{1}{e})}$-approximation~(this was also independently obtained in~\cite{DBLP:journals/tcs/LamprouSZ21}). Subsequently, Hochbaum et al.~\cite{DBLP:journals/ipl/HochbaumR20} introduced a more general {\sc Connected Max Coverage} problem, which is a variant of {\sc Max Coverage}, where the selected sub-family needs to induce a connected subgraph in another \emph{connectivity graph} additionally provided in the input. For this problem, they gave an $\lr{(1-\frac{1}{e})\cdot \max\LR{\frac{1}{R}- \frac{1}{k}, \frac{1}{k}}}$-approximation, where $R$ is the radius of the connectivity graph. Very recently, D'Angelo and Delfaraz~\cite{DBLP:conf/ifaamas/DAngeloD25} designed bicriteria approximations for the same problem; specifically, polylogarithmic approximations for the number of elements covered with a solution of size $(1+\varepsilon)k$. 

However, to the best of our knowledge, these problems have not yet been systematically explored within the parameterized approximation framework. The goal of our work is twofold:

\begin{tcolorbox}[colback=blue!4!white,colframe=blue!60!black]
  \begin{enumerate}
  \setlength{\itemsep}{-2pt}
      \item To formulate a general model of connectivity-constrained coverage that is expressive enough to capture a broad range of problems studied in the prior polynomial-time literature, and 
      \item To develop new techniques within the \fpt approximation paradigm for designing approximation schemes for the two fundamental minimization and maximization problems that naturally arise in this framework.  
  \end{enumerate}
 
\end{tcolorbox}

Our model separates the \emph{constraint layer}—a companion graph $G_{\mathrm{conn}}$ that can enforce connectivity, independence/packing, or fault tolerance—from the \emph{coverage layer}—the incidence bipartite graph $\gcov=(R\cup B,E)$ that represents the set-element incidences from a hypergraph. This separation lets us handle a wide range of set-system (hypergraph) families, including bounded-rank and biclique-free incidences, bounded VC-dimension and geometric ranges, as well as multi-coverage and weighted variants. In the next subsection, we formally define our model and demonstrate its expressivity by showing how previously studied problems can be naturally captured within this framework.


\subsection{Our Model of Connectivity-Constrained Coverage} \label{subsec:conncovmodel}

We introduce the following problem that is central to this work.

\begin{tcolorbox}[enhanced,title={\color{black} {\pcrbdsfull(\pcrbdsshort)}}, colback=white, boxrule=0.4pt,
	attach boxed title to top center={xshift=-.5cm, yshift*=-2.5mm},
	boxed title style={size=small,frame hidden,colback=white}]
	
	\textbf{Input:} An instance $\cI = (\gcon, \gcov,  k, t)$, where
	\begin{itemize}
         \item $\gcon = (R, E)$ is an arbitrary graph, called the \emph{\color{bblue}connectivity graph},
	\item $ \gcov= (R \uplus B, E')$ is a bipartite            graph, called the \emph{\color{bblue}coverage graph}, and
	\item $k$ and $t$ are non-negative integers. 
	\end{itemize}
	\textbf{Question:} \hspace*{.06cm} Does there exist a vertex subset $S \subseteq R$ such that,
	\begin{enumerate}[label=(\arabic*)]
            \item \label{oneprobdef} $\big|S\big| \le k$,
		\item \label{twoprobdef} $\gcon[S]$ is connected, and
		\item \label{threeprobdef} $\big|N_{\gcov}(S)\big| \ge t$ ?
	\end{enumerate}
\end{tcolorbox}

Let us unpack the definition. The input consists of two graphs $\gcon$ and $\gcov$, and two non-negative integers $k$ and $t$. Here, $\gcov$ is a bipartite graph with vertex set $R \uplus B$, that is to be thought of as the \emph{incidence graph} of a set system -- the ``red side'' ($R$) corresponds to the set family $\cF$, and the ``blue side'' ($B$) corresponds to the universe $\cU$, with an edge representing set-element containments. Notice that under this interpretation, {\sc Max Coverage} corresponds to selecting a $k$-sized subset of $R$ that dominates the maximum number of vertices in $B$. Next, we have a connectivity graph $\gcon$, whose vertex set is also $R$. This graph is used to model the connectivity of the solution, i.e., a solution $S \subseteq R$ is required to induce a connected subgraph of $\gcon$. 

While the decision question asks whether there exists such a connected solution of size at most $k$ that dominates at least $t$ vertices in $B$, two natural optimization variants stem from it: find a connected solution $S$ that 
\begin{enumerate}[label=(\roman*)]
\setlength{\itemsep}{-2pt}
    \item maximizes $|N_{\gcov}(S)|$ s.t.~$|S| \le k$, and 
    \item minimizes $|S|$ s.t.~$|N_{\gcov}(S)| \ge t$. 
\end{enumerate}
\noindent 
To handle both variants in a unified way, we define the notion of an $(\alpha, \beta)$-approximation algorithm: such an algorithm either finds a connected solution $S$ such that $|S| \le \alpha k$ and $|N_{\gcov}(S)| \ge \beta t$; or correctly outputs that there is no $S \subseteq R$ satisfying the original requirements. 

\medskip

\paragraph*{Modeling prior problems.} Here we describe how we can model the problems from the prior literature as instances of \pcrbdsshort. Consider {\sc Connected Partial/Budgeted Dominating Set}, studied in \cite{DBLP:journals/siamdm/KhullerPS20,DBLP:journals/tcs/LamprouSZ21}, where we are given a graph $G$, and we want to find a connected dominating set (with dual objectives). To model it as \pcrbdsshort, we let $\gcon \coloneqq G$, and $\gcov$ is a bipartite graph on the vertex set $R \uplus B$, where $R = V(G)$, and $B$ is another copy of $V(G)$. For each $u \in V(G)$, we add the edges between $u$ and the copies of all $w \in N_G[u]$ in the set $B$. Note that there is a close coupling between $\gcon$ and $\gcov$ while modeling {\sc Connected Partial Dominating Set}, which is not always the case. Indeed, consider the {\sc Connected Max Coverage} problem studied in~\cite{DBLP:journals/ipl/HochbaumR20,DBLP:conf/ifaamas/DAngeloD25}, where we are given a set system $(\cU, \cF)$, and a graph $G$ with $V(G)  = \cF$. To model it as \pcrbdsshort, we let $\gcov$ be the set-element incidence graph (as described above), and let $\gcon \coloneqq G$. At the other extreme, the vanilla {\sc Max Coverage} can be modeled by letting $\gcon$ be a clique, rendering the connectivity requirement redundant. 

{\sc Partial Hitting Set} is the ``dual variant'', where the roles of elements of the universe $\cU$ and the sets of $\cF$ are reversed~(\cite{DBLP:journals/jal/GandhiKS04}) --- we want to select a minimum number of elements from $\cU$ that hits at least $t$ elements. This again can be modeled as \pcrbdsshort by letting $\gcov$ be the set-element incidence graph as before; however with the roles of $R$ and $B$ flipped ($R$ corresponds to $\cU$, and $B$ to $\cF$), and $\gcon$ is a clique defined on the vertex set that corresponds to $\cU$.

\subsection{Our Algorithmic Results}
\pcrbdsshort\ is $\mathsf{W[1]}$-hard (\Cref{thm:hardkpds}) parameterized by $k$  even under strong restrictions: it remains hard when $G_{\mathrm{conn}}$ is a clique or a star and the incidence (coverage) graph $\gcov$ is $3$-degenerate, and even when $\gcov$ is $K_{2,2}$-free. For every $\varepsilon>0$, there is no polynomial-time $(1,\,1-1/e+\varepsilon)$-approximation unless $\mathsf{P}=\mathsf{NP}$; moreover, under \textsf{ETH}, no algorithm running in time $f(k)\cdot n^{o(k)}$ achieves either a $(g(k),1)$-approximation (\Cref{thm:hardkapproxpds}) for any computable function $g(\cdot)$ (with respect to the size bound $k$) or a $(1,\,1-1/e+\varepsilon)$-approximation (with respect to the coverage threshold $t$) (\Cref{thm:hardtapproxpds}). The graphical special cases reflect the same picture: \textsc{Partial Connected Dominating Set} is $\mathsf{W[2]}$-hard parameterized by $k$ and inherits the \textsf{ETH}-based $f(k)\cdot n^{o(k)}$ inapproximability bounds above, while \textsc{Partial Connected Vertex Cover} is $\mathsf{W[1]}$-hard parameterized by $k$ (\Cref{theo:whard}). For more details see \Cref{secintro:dichotomy}. These hardness boundaries delineate a natural ``sweet spot'': under suitable structural restrictions on the incidence graph one can still aim for fine-grained (FPT) approximation schemes.

A natural question is: \emph{which} restrictions on the incidence graph suffice?
Recall that \textsc{Partial Vertex Cover} admits a $(1-\varepsilon)$-approximation (for the coverage
target $t$)~\cite{DBLP:journals/cj/Marx08,DBLP:conf/soda/Manurangsi19,DBLP:journals/jair/SkowronF17}
and a $+1$ \emph{additive} approximation (for the size bound $k$)~\cite{DBLP:conf/soda/0001KPSS0U23}.
In fact, these guarantees hold when \(\mathcal{D}\) is the class of incidence graphs of
\(\kddfree\) set systems (i.e., no \(d\) sets share \(d\) common elements).
A natural next step is to ask whether the same techniques extend to \emph{connectivity-constrained}
coverage—namely, when the chosen set \(S\subseteq R\) must also satisfy graph-side constraints
imposed by an arbitrary companion graph \(\gcon\).
Answering this question (almost) affirmatively is the main technical contribution of our work.


Our main algorithmic contribution is parameterized approximation schemes for the two natural optimization variants of \pcrbdsshort when the coverage graph $\gcov$ is \kddfree. Specifically, we prove the following two theorems. The first of the following two theorems is when we insist on a connected solution of size at most $k$, and wish to approximate the number of blue vertices dominated by the solution. 

\begin{restatable}{theorem}{EPASfortcov} \label{thm:introEPASfortcov}
    For any $\varepsilon > 0$, there exists an $(1, 1-\varepsilon)$-approximation algorithm with running time $2^{\Oh(k^2d/\varepsilon)} \cdot |\cI|^{\Oh(1)}$ for an instance $\cI$ of \pcrbdsshort where the coverage graph is \kddfree, and the connectivity graph is arbitrary. That is, this algorithm takes an instance $\cI = (\gcon, \gcov, k, t)$ such that $\gcov$ is \kddfree, runs in time $2^{\Oh(k^2d/\varepsilon)} \cdot |\cI|^{\Oh(1)}$ and either (i) outputs $S \subseteq R$ of size at most $k$ such that $\gcon[S]$ is connected and $|N_{\gcov}(S)| \ge (1-\varepsilon)t$; or (ii) correctly reports that there exists no $S \subseteq R$ of size at most $k$ such that $\gcon[S]$ is connected, and $|N_{\gcov}(S)| \ge t$.
\end{restatable}
We prove a complementary result, where we approximate the solution size, but insist that it dominates at least $t$ blue vertices. More formally, we prove the following theorem.

\begin{restatable}{theorem}{PASforkcov} \label{thm:introPASforkcov}
    For any $\varepsilon > 0$, there exists an $(1+\varepsilon, 1)$-approximation algorithm with running time $2^{\Oh(kd(k^2+\log d))} \cdot |\cI|^{\Oh(1/\varepsilon)}$ for an instance $\cI$ of \pcrbdsshort, where the coverage graph is \kddfree, and the connectivity graph is arbitrary. That is, this algorithm takes an instance $\cI = (\gcon, \gcov, k, t)$ such that $\gcov$ is \kddfree, runs in time $2^{\Oh(kd(k^2+\log d))} \cdot |\cI|^{\Oh(1/\varepsilon)}$, and either (i) outputs $S \subseteq R$ of size at most $(1+\varepsilon)k$ such that $\gcon[S]$ is connected and $|N_{\gcov}(S)| \ge t$; or (ii) correctly reports that there exists no $S \subseteq R$ of size at most $k$ such that $\gcon[S]$ is connected, and $|N_{\gcov}(S)| \ge t$.
\end{restatable}
We reiterate that both of the following results are in the setting where $\cC$ is arbitrary and $\cD$ is a family of bipartite incidence graphs that are \kddfree for some $d \ge 1$. Recall that Khuller et al.~gave an $\Oh(\log \Delta)$-approximation for {\sc Partial Connected Dominating Set}. Our result above (\Cref{thm:introPASforkcov}) improves upon this result by returning an $(1+\varepsilon)$-approximation in $\os(2^{k\Delta (k^2 + \log \Delta)})$ time.

\paragraph*{Neighborhood Sparsifiers: A Conceptual Contribution.} 
A natural approach  to solve our problem is to assign a \emph{weight} to each red vertex proportional to its coverage degree in \(\gcov\)—say \(w(v)\coloneqq |N_{\gcov}(v)|\)—and then pick a connected \(k\)-set in \(\gcon\) with maximum total weight, i.e., a maximum-weight \(k\)-vertex tree in \(\gcon\) (solvable in \(2^{\cO(k)}n^{\cO(1)}\) time; see \Cref{sec:fptasfort}). 
The obstacle is \emph{additivity}: overlaps in neighborhoods mean that a single weight function \(w:R\to\mathbb{Z}_{\ge0}\) satisfying
\[
  w(S)\;-\;\varepsilon t \;\le\; |N_{\gcov}(S)| \;\le\; w(S)\;+\;\varepsilon t
  \qquad\text{for all } S\subseteq R,\ |S|\le k
\]
(and computable in FPT time) appears out of reach.

Instead, when $\gcov$ is \kddfree, we construct a \emph{small family} of surrogate weight functions, one of which has the desired property. This reduces the task to, for each weight function, solving a maximum-weight \(k\)-tree in \(\gcon\) and taking the best outcome. Technically, the family arises from a combinatorial \emph{sparsification} lemma for \(\gcov\): via a greedy step coupled with random separation (later derandomized), we obtain a sparsified incidence graph that \emph{preserves} the neighborhoods of all  \(k\)-sets. In the sparsified instance, the simple per-vertex degree \(w(v)=\deg_{\text{sparse}}(v)\) serves as the required weight, yielding the desired approximation once lifted back to the original graph.



\begin{lemma}[Neighborhood Sparsifiers for \kddfree\ coverage]
\label{lem:intro-weight-family}
Let \(\cI=(\gcon,\gcov,k,t)\) be an instance of \pcrbdsshort\ with \(\gcov\) \kddfree.
There exist \(\ell=f(k,d,\varepsilon) \cdot \log|B|\) weight functions
\[
\{w_1,\ldots,w_\ell\},\qquad w_i: R \to \mathbb{Z}_{\ge 0},
\]
computable in time \(g(k,d,\varepsilon)\,\cdot |B|^{\cO(1)}\), such that the following holds.
For every \(S\subseteq R\) with \(|S|=k\), there is some
\(j\in[\ell]\) with
\[
  w_j(S) - \varepsilon t \;\le\; |N_{\gcov}(S)| \;\le\; w_j(S),
\]
where \(w_j(S)\coloneqq \sum_{v\in S} w_j(v)\).
\end{lemma}

We believe that \Cref{lem:intro-weight-family} will find applications beyond this work, in the spirit of seminal cut and spectral sparsifiers~\cite{DBLP:conf/stoc/BenczurK96,DBLP:journals/siamcomp/SpielmanT11}.



\subsection{Potential for a Dichotomy Program} 
\label{secintro:dichotomy}

Our results motivate a clean classification program separating regimes that admit FPT-approximation from those that are provably hard. We study the (parameterized) complexity of \pcrbdsshort\ over “\((\mathcal{C}\times\mathcal{D})\)” families of instances, where the constraint graph \(\gcon\) ranges over a class \(\mathcal{C}\) and the coverage/incidence graph \(\gcov\) ranges over a class \(\mathcal{D}\). We instantiate this framework by varying \(\mathcal{C}\) (e.g., bounded treewidth, clique/star) and \(\mathcal{D}\) (e.g., biclique-free, bounded rank/degeneracy), thereby charting the frontier between FPT-approximability and hardness. We exemplify this by varying $\cC$ and $\cD$ in a few different ways.

\begin{description}[leftmargin=*,format=\textcolor{newgray}]
\setlength{\itemsep}{-2pt}
    \item[$\cC$ = cliques, $\cD$ = arbitrary.] As mentioned earlier, this models an arbitrary instance of {\sc Max Coverage}, and therefore \pcrbdsshort inherits all its positive and negative results. In \Cref{sec:lower}, we explicitly spell out these approximation and parameterized lower bounds in order to contrast them with our positive algorithmic results.
    \item[$\cC$ = max degree $\Delta$, $\cD$ = arbitrary.] In this case, we can enumerate $n \cdot \Delta^{\Oh(k)}$ connected vertex-subsets of $\gcon$ of size at most $k$, and check whether any subset dominates at least $t$ blue vertices in $\gcov$. Therefore, in this case \pcrbdsshort is \fpt parameterized by $k+\Delta$. Note, in particular, that this captures the special case when $\gcon$ is a path (in this case, in fact, the problem is polynomial-time solvable). It may be tempting to conjecture that \pcrbdsshort is (fixed-parameter) tractable even when $\cC$ is the class of trees. This, however, turns out not to be the case.
    \item[$\cC$ = stars, $\cD$ = arbitrary.] In this case, we can again model an arbitrary instance of {\sc Max Coverage}, by adding a dummy set that maps to the central vertex of the star, and all the original sets mapping to the leaves of the star. Thus, all the intractable results carry over from {\sc Max Coverage}. 
    \item[$\cC$ = arbitrary, $\cD$ = arbitrary.] In this most general case, we design an \fpt algorithm for \pcrbdsshort parameterized by $t$, which runs in time $\os((2e)^t)$ in \Cref{sec:fptpds}. This shows that, parameterized by the desired coverage, the problem is fixed-parameter tractable \emph{in spite of} the connectivity requirements. 
    \item[$\cC$ = arbitrary, $\cD$ = each vertex in $R$ has degree at most $\Delta_R$.] In this case, note that any $k$-sized set $S \subseteq R$ has $|N_{\gcov}(S)| \le k \Delta$, implying that if $t > k \Delta_R$ then we can conclude that we have a \no-instance. Otherwise, the \fpt algorithm from the previous paragraph implies that \pcrbdsshort is \fpt parameterized by $k + \Delta_R$.
    \item[$\cC$ = arbitrary, $\cD$ = each vertex in $B$ has degree at most $\Delta_B$.] Note that when $\Delta_B = 2$, this models the coverage of edges by vertices in a graph. Therefore, when $\gcon$ is a clique or a star, one can model {\sc Partial Vertex Cover}, which is known to be W[1]-hard parameterized by $k$; and \pcrbdsshort inherits this hardness.
\end{description}



\section{Technical Overview of Our Main Results} \label{sec:overview}
Next, we describe some of the technical ideas that lead to our algorithmic  results \Cref{thm:introEPASfortcov} and \Cref{thm:introPASforkcov}. Due to the inherent nature of the problem, our algorithms and the corresponding arguments need to juggle between the two graphs $\gcon$ and $\gcov$, and their interplay with a hypothetical solution that we seek. These arguments go via a couple of auxiliary graphs that make finding a solution easier, albeit at the expense of a small loss in the quality of the solution (either in terms of the solution size, or in terms of the number of dominated blue vertices).

\paragraph*{Overview of the proof of \Cref{thm:introEPASfortcov}.} Suppose that the input instance $\cI = (\gcon, \gcov, k, t)$ is a \yes-instance of \pcrbdsshort, where $\gcon$ is arbitrary and $\gcov$ is \kddfree. Let $S^\star \subseteq R = V(\gcon) = V(\gcov)$ be an (unknown) solution of size $k$ that is connected in $\gcon$ and dominates at least $t$ blue vertices. 

\begin{description}[leftmargin=*,format=\textcolor{newgray}]
    \item[1.~Defining and working with a conflict graph.] We construct an auxiliary graph, called the \emph{conflict graph}, denoted as $\gconf$, where $V(\gconf) = R$, and there is an edge between two vertices iff the number of common blue neighbors of the two vertices is a ``significant fraction'' of $t$ (we will use $\frac{\varepsilon t}{k^2}$). Using the fact that $\gcov$ is \kddfree, one can show that the maximum degree of $\gconf$ is bounded by some $f(k, d, \varepsilon)$, notably it is independent of $t$ (cf.~\Cref{lem:maxdegkddfreepds}). That is, for each red vertex $u$, the number of other red vertices that have a significant overlap with $u$ in terms of blue domination is bounded by $f(k, d, \varepsilon)$. Notice that some pairs of vertices from $S^\star$ (our unknown solution) can have significant overlap, but maybe not all. Therefore, $\gconf[S^\star]$ may be disconnected with at most $k$ connected components, denoted by $\mathcal{C}^\star$. Since $\Delta(\gconf) \le f(k, d, \varepsilon)$, we can use the technique of \emph{random separation} to obtain a collection $\cC$ of induced connected subgraphs of $\gconf$ such that, $\cC^\star \subseteq \cC$ (\emph{even though} $S^\star$ is unknown to us). The success probability of this step is inverse \fpt in $k, d$, and $\varepsilon$; and this step can be derandomized using known combinatorial tools. Let us assume henceforth that we have such a collection $\cC$ in our hand. The algorithmic task is to figure out which set of connected components from $\cC$ together form the desired solution. 

    \item[2.~Using $\cC$ to sparsify $\gcov$.] Recall that our unknown solution $S^\star$ is partitioned across $\cC$ in such a way that, either all or none of the vertices of a component in $\cC$ are a part of $S^\star$. We process each blue vertex $u$ as follows: whenever $u$ has edges to two or more red neighbors within the same component of $\cC$, we arbitrarily delete all but one of these edges. We repeat this until every blue vertex has at most one neighbor in each component. The resulting graph $\gspar$ is a subgraph of $\gcov$. This way, each blue vertex $u$ has at most one neighbor in $\gspar$ within any component of $\cC$, and that neighbor will be responsible for counting the domination of $u$. For each red vertex $r \in R$, we define its \emph{weight} as the number of its blue neighbors in the sparsified graph $\gspar$. Note that, while a solution may span across multiple components, and hence some domination overlap across different components may still remain, the construction of the conflict graph guarantees that such inter-component overlap is a negligible fraction of $t$.  Recall that the different weight functions mentioned in \Cref{lem:intro-weight-family} correspond to different weight functions obtained  in this manner, corresponding to different colorings from the random separation step. 
    
    \item[3.~Using sparsified weights to find an approximate solution.] We now view $\gcon$ as a weighted graph, where the weight of each vertex is defined from the sparsification step. Then, we find a connected subgraph of $\gcon$ on at most $k$ vertices with maximum total weight. This can be done using an application of the classical color-coding approach~\cite{AlonYZ95}. If there exists a feasible solution $S^\star$ of size $k$, it corresponds to a connected subgraph with total weight at least $t$. Conversely, if we identify a connected subgraph $S$ with total weight at least $t$, then even after discounting the small amount of over-counting across different components, we are guaranteed that $|N_{\gcov}(S)| \geq (1-\varepsilon)t$. This completes the description of \Cref{thm:introEPASfortcov}.
\end{description}

\paragraph*{Overview of the proof of \Cref{thm:introPASforkcov}.}  
As before, suppose that the given instance $\cI$ is a \yes-instance and let $S^\star$ be an unknown solution. In this setting, the goal of $(1+\varepsilon, 1)$-approximation is flipped: we no longer insist that the solution has size exactly $k$ (we allow solutions of size up to $(1+\varepsilon)k$); but instead require that it dominates at least $t$ blue vertices.

\begin{description}[leftmargin=*, format=\textcolor{newgray}]
    \item[1.~A first attempt, and why it fails.]  We begin by running the algorithm from the previous subsection with a suitable choice of $\delta$, obtaining an $(1,1-\delta)$-approximate solution $S$ of size $k$ that covers at least $(1-\delta)t$ blue vertices. If $|N_{\gcov}(S)| \geq t$, we are done. Otherwise, the solution is only slightly short; but remember that we are allowed to add a few more vertices to reach the required coverage.

    Let $H$ be the set of the $g(k,d,\delta)$ highest-degree red vertices. A combinatorial lemma (cf.~\Cref{lem:addsol} originally from \cite{DBLP:conf/soda/0001KPSS0U23}) guarantees that, when $\gcov$ is \kddfree, at least one of the following holds.  
    \begin{enumerate}
        \item[(i)] the unknown solution intersects $H$, i.e., $H \cap S^\star \neq \emptyset$, or
        \item[(ii)] there exists a vertex $h \in H$ such that $S \cup \{h\}$ dominates at least $t$ blue vertices
    \end{enumerate}
    This suggests a natural strategy: if case (i) holds, then we add $h$ to the solution, and recurse. Otherwise, in case (ii), we simply add $h$ to $S$, yielding a solution of size $k+1$, which is much smaller than the allowed $(1+\varepsilon)k$. But there is a catch: we also require that the solution be connected in $\gcon$, and there is no guarantee that $\gcon[S \cup \{h\}]$ is connected. If $h$ happened to be adjacent to, or close to, a vertex of $S$, this would be easy to fix, but this does not hold in general.

    \item[2.~A structural lemma to the rescue.] To overcome this, we first show the following structural lemma (cf.~\Cref{lem:solutioncovering}) that may be of an independent interest: since $\gcon[S^\star]$ is connected, there exists a subset $C \subseteq S^\star$ of size at most $1/\varepsilon$ such that every vertex of $S^\star$ is at distance at most $\varepsilon k$ from $C$. We ``guess'' $C$ by iterating over all subsets of $R$ of size at most $1/\varepsilon$, and in each iteration, we require that the guessed set $C$ must be a part of our solution that we will build. During recursion, $C$ may grow as we discover additional vertices that must belong to $S^\star$; but the original $C$ serves the purpose of ``cheaply connecting'' a vertex to the solution.

    \item[3.~Amended recursive strategy.] With this additional seed $C$, we run the previous algorithm again to find a set $S$ of size $k$ containing $C$ that dominates at least $(1-\delta)t$ blue vertices (this requires generalizing the previous algorithm so that we can provide to it an additional \emph{terminal} set---here $C$---that must be contained in the solution, cf.~\Cref{subsec:steinerext}). At this point, we focus only on high-degree vertices that are near $C$ in $\gcon$, since $S^\star$ must be contained in this neighborhood. We then apply the lemma from Step 1:
    \begin{itemize}
        \item In case (i), we guess $h \in H$ that is guaranteed to belong to $S^\star$. Then we branch on all possibilities for $h$, extend $C$ by including it, and recurse.
        \item In case (ii), for some $h \in H$, the set $S \cup \{h\}$ dominates at least $t$ blue vertices. Because $h$ is close to $C$, we can connect it to $S$ by adding a path $P$ in $\gcon$ of length at most $\varepsilon k$. The resulting set $S \cup P$ is connected, has size at most $(1+\varepsilon)k$, and dominates at least $t$ blue vertices. Iterating over all $h \in H$, we return a solution whenever this case applies.
    \end{itemize}
\end{description}

This concludes the overview of the proof of \Cref{thm:introPASforkcov}. We note that the running time of this $(1+\varepsilon, 1)$-approximation algorithm is of the form $h(k,d) \cdot |\cI|^{\Oh(1/\varepsilon)}$, in contrast to the $(1,1-\varepsilon)$-approximation of \Cref{thm:introEPASfortcov}, which runs in time $f(k,d,\varepsilon) \cdot |\cI|^{\Oh(1)}$. A natural question is whether the dependence on $\varepsilon$ in \Cref{thm:introPASforkcov} can be improved, i.e., whether one can achieve a running time of the form $h(k,d,\varepsilon) \cdot |\cI|^{\Oh(1)}$. In analogy with PTAS versus EPTAS, this corresponds to improving from a PAS to an EPAS~\footnote{EPAS stands for \emph{Efficient Parameterized Approximation Scheme}, i.e., a $(1\pm \varepsilon)$-approximation in time $f(\kappa, \varepsilon) \cdot |I|^{\Oh(1)}$ for some parameter $\kappa$. PAS stands for \emph{Parameterized Approximation Scheme}, where the running time can be $f(\kappa, \varepsilon) \cdot |I|^{g(\varepsilon)}$.}.

However, obtaining such an improvement would be unlikely due to the following simple reason. Indeed, suppose for any $\varepsilon > 0$ we could compute a connected solution of size at most $(1+\varepsilon)k$ that covers at least $t$ blue vertices. Setting $\varepsilon \gets \frac{1}{2k}$ would then yield a solution of size at most $k+\frac{1}{2}$, which must in fact be of size at most $k$. This would give an exact \fpt\ algorithm for \pcrbdsshort parameterized by $k$ and $d$; however the problem is \textsf{W}[1]-hard when parameterized by $k$ even when $d = 2$ (cf.~{\sc Partial Vertex Cover}). Thus, \Cref{thm:introPASforkcov} achieves, in a precise sense, the best possible form of approximation scheme, apart from potential improvements to the $\fpt$ factor of the running time.


\end{sloppypar}

\section{Preliminaries}\label{sec:preli}
Let $[n]$ be the set of integers $\{1,\ldots, n\}$. For a  graph $G$, we denote the set of vertices of $G$ by $V(G)$ and the set of edges by $E(G)$. For a  directed graph $D$, we denote the set of vertices of $D$ by $V(D)$ and the set of edges by $A(D)$. We denote an edge of an undirected graph by $uv$ and an arc of a directed graph by $(u,v)$.  For  a subset $X$ of vertices $X \subseteq V(G)$, we use the notation $G-X$ to mean the graph $G[V(G) \setminus X]$. For a vertex subset $S \subseteq V(G)$, $G[S]$ denotes the subgraph of $G$ induced on $S$, i.e., $V(G[S])= S,~E(G[S])= \{uv: u,v \in S, uv \in E(G)\}$. For a graph $G$ and a vertex $v \in V(G)$, we define \textit{open neighborhood} of a vertex as $N_G(v) \coloneqq \lrc{u \mid uv \in E(G)}$ and \textit{closed neighborhood} of a vertex as $N_G[v] \coloneqq N_G(v)\cup \{v\}$. Further for a set $X \subseteq V(G)$, open neighborhood of $X$ is defined as 
$$N_G(X) \coloneqq \lr{\bigcup\limits_{x \in X}N_G(x)} \setminus X$$

and closed neighborhood of $X$ is defined as

$$N_G[X] \coloneqq N_G(X) \cup X$$

We use standard terminology from the book of Diestel~\cite{diestel2025graph} for
those graph-related terms that are not explicitly defined here. We refer to \cite{cygan2015parameterized} for an introduction to the area of parameterized complexity and related terminology.



\newcommand{\rdsot}{{\sc RDSOT}\xspace}
\section{\fpt for \pcrbdsshort parameterized by $t$}\label{sec:fptpds}
In this section we design an algorithm for 	 \pcrbdsshort, parameterized by $t$. The result is obtained by giving an \fpt reduction to  {\sc Relaxed Directed  Steiner Out-Tree} problem. A \emph{directed out-tree} (or \emph{arborescence}) is a digraph whose underlying undirected graph is a tree and that has a unique root $r$ such that every arc is directed away from $r$. Equivalently, the root $r$ has in-degree $0$ and every other vertex has in-degree exactly $1$. 
The problem {\sc Relaxed Directed  Steiner Out-Tree} is defined as follows:

\begin{tcolorbox}[enhanced,title={\color{black} {{\sc Relaxed Directed  Steiner Out-Tree} (\rdsot)}}, colback=white, boxrule=0.4pt,
    	attach boxed title to top center={xshift=-2.5cm, yshift*=-2.5mm},
    	boxed title style={size=small,frame hidden,colback=white}]
    	\textbf{Input:} A  directed graph $D = (V, A)$, a set of terminals $T \subseteq V$, an   integer $p$.\\
        \textbf{Question:} Is there  a directed out-tree $D' \subseteq D$ with   $|V(D')| \leq p$ and $T \subseteq V(D')$?
\end{tcolorbox}

If in addition, we are given an additional distinguished vertex $r \in V$ and ask to check whether  $D$ contain a directed out-tree on at most $p$ vertices that is
rooted at $r$ and that contains all the vertices of $T$, then that problem is referred as  {\sc Directed  Steiner Out-Tree}. Misra et al.~\cite{DBLP:journals/jco/MisraPRSS12} gave an algorithm for {\sc Directed  Steiner Out-Tree} that runs in time $2^{|T|} \cdot n^{\mathcal{O}(1)}$. Note that we can also solve {\sc Relaxed Directed  Steiner Out-Tree} in time $2^{|T|} \cdot n^{\mathcal{O}(1)}$ by ``guessing'' (i.e., enumerating $n$ choices for) the root vertex $r$, and solving the resulting instance of {\sc Directed Steiner Out-Tree}. Now we use this result to design a randomized algorithm for \pcrbdsshort.

\paragraph*{Randomized reduction from \pcrbdsshort to \rdsot.} 
Let $\cI = \lr{\gcon, \gcov,  k, t}$ be the given instance of \pcrbdsshort. Let $f:B \to [t]$ denote an arbitrary function (referred to as a \emph{coloring}), where each vertex of $B$ is assigned an integer from $[t]$. For each $i$, let $B_i$ denote the vertices with color $i$ according to $f$, i.e., $B_i= \{v \mid v\in B, f(v)=i\}$ (we omit the dependence of $f$ in the notation for the sake of brevity).

Now, we construct a directed graph $D_f$ that will serve as the input for the eventual instance of {\sc Relaxed Directed Steiner Out-Tree}. We define  $$
V(D) := V(\gcov) \;\uplus\; T, 
\qquad\text{ where } T = \{\tau_i : i \in [t]\}.
$$  

The arc set $A(D_f)$ is constructed by adding the following three types of arcs.
\begin{itemize}
\item Type 1: For every edge $uv \in E(\gcon)$, add both arcs $(u,v)$ and $(v,u)$ to $A(D_f)$.  
\item  Type 2: For every edge $uv \in E(\gcov)$ with $u \in R, v \in B$, add the arc $(u,v)$ to $A(D_f)$. 
\item Type 3: For each $i \in [t]$, add the arcs $\{(v, \tau_i) : v \in B_i\}$ 
to $A(D_f)$.
\end{itemize}
Let $p = k+2t$, and let $\cI'_f = (D_f, T, p)$ be the resulting instance of {\sc Relaxed Directed Steiner Out-Tree}. 

Our randomized algorithm works on the given instance $\cI = (\gcon, \gcov, k, t)$ as follows. First, it obtains a \emph{random} $f: B \to [t]$, i.e., for each $u \in B$, it independently and uniformly assigns a color from $[t]$. Then, it constructs the instance $\cI'_f = (D_f, T, p)$ of \rdsot as described above. Then, it runs the $\os(2^{|T|}) = \os(2^t)$ time algorithm of \cite{DBLP:journals/jco/MisraPRSS12}, which returns \yes iff $\cI'_f$ is a \yes-instance of \rdsot. It repeats this procedure (of selecting a random coloring $f$ and solving the resulting $\cI'_f$) independently up to $e^t$ times. If in any of the iterations, the algorithm of \cite{DBLP:journals/jco/MisraPRSS12} outputs \yes, then our algorithm reports that the original instance $\cI$ of \pcrbdsshort is a \yes-instance. Otherwise, if in all iterations the algorithm returns \no, then the algorithm outputs that $\cI$ is a \no-instance of \pcrbdsshort.


\begin{restatable}{theorem}{fptbytpds} \label{thm:fptbytpds}
	   \pcrbdsshort admits a deterministic (resp.~randomized) algorithm that runs in time $\Oh^\star((2e)^{t + o(t)})$ (resp.~$\Oh^\star((2e)^{t})$ and outputs the correct answer with probability at least $1-1/e$).
\end{restatable}
\begin{proof}
Let $\cI = (\gcon, \gcov, k, t)$ be an instance of \pcrbdsshort. Suppose $\cI$ is a \yes-instance, and let $S \subseteq R$ be a solution of size $k$ such that $\gcon[S]$ is connected and $|N_{\gcov}(S)|\ge t$. Let $Q \subseteq N_{\gcov}(S)$ be an arbitrary subset of size exactly $t$. We say that a coloring $f: B \to [t]$ is a \emph{good coloring} iff $f$ assigns distinct colors to distinct vertices of $Z$. It is easy to see that if $\cI$ is a \yes-instance, then with probability $e^{-t}$, a random coloring $f$ is a good coloring. We prove the following two statements.

\begin{enumerate}
    \item If $\cI$ is a \yes-instance, and $f$ is a good coloring, then the resulting instance $\cI'_f$ is a \yes-instance of \rdsot.
    \item If $\cI$ is a \no-instance of \pcrbdsshort, then for any coloring $f$, $\cI'_f$ is a \no-instance of \rdsot.
\end{enumerate}


    \paragraph*{Proof of statement 1.} Consider the vertex set $ S \cup Q \cup T$. Since $\gcon[S]$ is connected, the Type 1 arcs ensure that $D_f[S]$ is strongly connected. Each $\tau_i \in T$ has an incoming arc from some vertex of $Q \cap B_i$, guaranteed by the good coloring. Thus one can orient a spanning out-tree rooted in $S$ that reaches all of $Q$ and subsequently all of $T$. Hence $D$ contains a directed out-tree on at most $|S| + |Q| + |T| = k+2t$ vertices covering all terminals $T$. This implies that $\cI'_f$ is a \yes-instance of \rdsot.

    \paragraph*{Proof of statement 2.} We prove the contrapositive instead. That is, we show that, for the given instance $\cI$ of \pcrbdsshort, if for some coloring $f$, the resulting instance $\cI'_f$ of \rdsot is a \yes-instance, then $\cI$ is a \yes-instance of \pcrbdsshort. Suppose $D_f$ contains a directed out-tree $T^\star$ on at most $k+2t$ vertices containing all terminals $T$. Note that the terminal vertices do not have any outgoing edges, and therefore must have out-degree $0$ and in-degree $1$ in $T^\star$. For each $\tau_i \in T$, let $b_i \in B$ denote the unique in-neighbor of $\tau_i$ in $T^\star$.
    Let $T_1$ be the out-tree obtained by deleting the terminals $T$ from $T^\star$. Note that the vertices in $B$ had out-degree exactly $1$ in $T^\star$, with each vertex having a unique neighbor in $T$. Therefore, the vertices of $B$ now have out-degree $0$ and in-degree exactly in $T_1$. Furthermore $|V(T_1)| \le k+t$. Now, let $T_2$ be the out-tree obtained by deleting all vertices of $V(T_1) \cap B$ from $T_1$. Since $\LR{b_i: i \in [t]} \subseteq V(T_1) \cap B$, it follows that $|V(T_1)| \le k$.  We claim that $S \coloneqq V(T_2)$ is a solution to $\cI$. As argued earlier, $|S| \le k$. Further, for any $i \in [t]$, the vertex $b_i$ had some in-neighbor $v \in S$ in $T_1$, which comes from type 2 edges directed from $v$ to $b_i$. Therefore, $vb_i \in E(\gcov)$. It follows that $|N_{\gcov}(S)| \ge |\LR{b_i: i \in [t]}| = t$. Finally, we claim that $\gcon[S]$ is connected. Indeed, since we deleted vertices of out-degree $0$ and in-degree exactly $1$, it follows that $T_2$ is still a directed out-tree which solely consists of type 1 edges, each of which corresponds to an edge of $\gcon$. This shows that $\cI$ is a \yes-instance.


    Given the proof of the two statements, we first prove the statement about the randomized algorithm. From statement 1, it follows that if $\cI$ is a \yes-instance, then with probability at least $e^{-t}$, the resulting instance $\cI_f$ in a specific iteration is a \yes-instance. By repeating the algorithm independently $e^t$ times, the probability that at least one of the instances is a \yes-instance is at least $1-1/e$. Then, by the correctness of \rdsot algorithm, our algorithm outputs \yes with probability at least $1-1/e$. On the other hand, if $\cI$ is a \no-instance, then statement 2 implies that all the instances of \rdsot constructed in all iterations will be \no-instances, and hence our algorithm will correctly output \no. It is easy to verify that that the running time of the algorithm is $\os((2e)^t)$.

    \paragraph*{Derandomization.} For designing a deterministic variant, we use the tool called \emph{$(n, t)$-perfect hash family}~\cite{DBLP:books/sp/CyganFKLMPPS15}. Such a family $\cF$ is a family of functions $f: B \to [t]$ such that for each $P \subseteq B$, there exists some $f$ such that $f$ is injective on $P$. Such a family has size at most $e^{t} \cdot t^{\Oh(\log t)} \cdot \log |B|$, and can be constructed in time $e^{t} \cdot t^{\Oh(\log t)} \cdot |B| \log |B| = e^{t+o(t)} \cdot |B|^{\Oh(1)}$. The deterministic algorithm iterates over each $f \in \mathcal{F}$, constructs the instance $\cI'_{f}$ of \rdsot, and runs the algorithm of \cite{DBLP:journals/jco/MisraPRSS12}. It outputs \yes if and only if in any of the iterations, the algorithm for \rdsot outputs \yes. It follows that the deterministic variant has the running time $\os((2e)^{t+o(t)})$ and solves \pcrbdsshort. 
\end{proof}

\section{$(1, 1-\varepsilon)$-approximation for \pcrbdsshort} \label{sec:fptasfort}

In this section we design an \fpt\ algorithm (parameterized by $k+d+1/\varepsilon$) that achieves a $(1,1-\varepsilon)$-bicriteria approximation for \pcrbdsshort when the coverage graph $\gcov$ is \kddfree. 
If $t<\tfrac{k^{2}d}{\varepsilon}$, we invoke the exact algorithm of \Cref{thm:fptbytpds}, which (when a solution exists) returns a set of size at most $k$ dominating at least $t$ vertices. 
When $t\ge \tfrac{k^{2}d}{\varepsilon}$, assuming there exists a solution of size at most $k$ dominating at least $t$ vertices, we compute a set $S\subseteq R$ with $|S|\le k$, $\gcon[S]$ connected, and $S$ dominating at least $(1-\varepsilon)t$ vertices of $B$ in $\gcov$. 
Combining these two cases yields \Cref{thm:kddfptapproxpds}. 
Since the small-$t$ case is handled by \Cref{thm:fptbytpds}, the remainder of this section focuses on the regime $t\ge \tfrac{k^{2}d}{\varepsilon}$. 
We next present the algorithmic components and then assemble them into the full procedure, referring back to the corresponding subsections as needed.


Suppose $\cI=(\gcon,\gcov,k,t)$ is a \yes-instance of \pcrbdsshort; i.e., there exists
$S^\star\subseteq R$ with $|S^\star|\le k$ such that
\begin{enumerate}[label=(P\arabic*),ref=(P\arabic*),itemsep=0pt,leftmargin=*]
  \item\label{prop:conn} $\gcon[S^\star]$ is connected, and
  \item\label{prop:dom} $|N_{\gcov}(S^\star)|\ge t$.
\end{enumerate}
We will refer to these properties as \propone\ and \proptwo\ throughout this section.

\subsection{Step 1\label{stepone}: Construction of the Conflict Graph.}\label{subsec:conflictgraph}

In this subsection, we define the notion of a \emph{conflict graph} as follows.

\begin{definition}[conflict graph] \label{def:conflictgraphpds}
	Let $\gcon$ and $\gcov$ denote the connectivity and coverage graphs, respectively, as given in the input. We define \emph{conflict graph}, denoted by $\gconf$ as follows: the vertex set $V\left(\gconf \right) \coloneqq V \bigl( \gcon  \bigr) = R$, and $E\big(\gconf \big) \coloneqq \bigl\{uv : u, v
    \in R ~\text{\em and}~ \big|N_{\gcov}(u) \cap N_{\gcov}(v)\big| \geq \frac{\varepsilon t}{k^2}\bigr\}$. 
\end{definition}
	
	Note that the conflict graph can be constructed in time quadratic in the size of $\gcov$. In the following lemma, we bound the maximum degree of the conflict graph.

	\begin{lemma}[Restatement of Lemma 4.1 from \cite{DBLP:conf/soda/0001KPSS0U23}] \label{lem:maxdegkddfreepds}
Suppose the coverage  graph $\gcov	$ is \kddfree, and $t \ge \frac{k^2d}{\varepsilon}$. Then, the maximum degree $\Delta \coloneqq \Delta(\gconf)$ of the conflict graph,  $\gconf$ as in \Cref{def:conflictgraphpds} is at most $(d-1) \cdot \left(\tfrac{ek^2}{\varepsilon}\right)^{d}$, where $e$ denotes the base of natural logarithm.
\end{lemma}

\begin{proof}
The proof is essentially identical to that of Lemma 4.1 of \cite{DBLP:conf/soda/0001KPSS0U23}, but we describe it here for the sake of completeness, using the terminology used in this paper. Recall that, $\gcon, \gcov, \gconf$ denote the connectivity, coverage and conflict graphs respectively.  We show that the degree of every vertex $v \in V\lr{\gconf}$, that is $\deg_{\gconf}(v)$ is bounded. To this end, fix an arbitrary vertex $v \in V\lr{\gconf}$. WLOG, we only focus on the vertices with $\deg_{\gconf}(v) \geq 1$; 
vertices with degree $0$ trivially have bounded degree in $\gconf$.
Let us consider the sets $P \coloneqq N_{\gconf}(v) \subseteq R$ and $Q \coloneqq N_{\gcov}(v) \subseteq B$. Now consider the induced subgraph $\gcov[P \uplus Q]$ of the original coverage graph $\gcov$. Let this induced subgraph be $\gcov^v$.

Let $\mathcal{K}$ denote the set of all $K_{1, d}$'s in the graph $\gcov^v$, where one vertex (the center vertex of the star) belongs to $N_{\gconf}(v) \subseteq R$ and its $d$ neighbors belong to $ N_{\gcov}(v) \subseteq B$. Since every $u \in P$ has at least $\frac{\varepsilon t}{k^2} \geq d$ neighbors in $Q$, i.e., 
$\big|N_{\gcov^v}(u)\big| \geq \frac{\varepsilon t}{k^2} \geq d$, it follows that each $u \in P$ 
participates in at least $\displaystyle\binom{\frac{\varepsilon t}{k^2}}{d}$ distinct copies of $K_{1,d}$. Therefore, $|\mathcal{K}| \ge |P| \cdot \dbinom{\frac{\varepsilon t}{k^2}}{d}$. 

Further, we claim that $|\mathcal{K}| \le (d-1) \dbinom{|Q|}{d}$. Assume towards contradiction that $|\mathcal{K}| > (d-1) \dbinom{|Q|}{d}$. For each set $Z\subseteq Q$ of size exactly $d$, let $\kappa(Z)$ denote the number of $K_{1,d}$'s in $\mathcal{K}$ that all the vertices of $Z$ together participate in. Then, it follows that $(d-1) \cdot \dbinom{|Q|}{d} < |\mathcal{K}| = \sum\limits_{Z \subseteq Q: |Z| = d} \kappa(Z) \le \dbinom{|Q|}{d}\kappa(Z_{\max})$, where $Z_{\max} \coloneqq \argmax\limits_{Z \subseteq Q: |Z| = d} \kappa(Z)$. This implies that $\kappa(Z_{\max}) > d-1$, i.e., $\kappa(Z_{\max}) \ge d$. However, this implies the existence of a \kdd~in $\gcov^v$, which is a subgraph of $\gcov$, a contradiction. Therefore, we obtain that,
\begin{align*}
	|P| \cdot \dbinom{\frac{\varepsilon t}{k^2}}{d} \le |\mathcal{K}| &\le (d-1) \cdot \binom{|Q|}{d} < (d-1)\cdot\binom{t}{d} \tag{Since $|Q| < t$, else singleton solution}
	\\
    \implies |P| &\le (d-1) \cdot \frac{\binom{t}{d}}{\dbinom{\frac{\varepsilon t}{k^2}}{d}}
	\\&\le (d-1) \cdot \frac{\lr{\frac{et}{d}}^d}{\left(\frac{\varepsilon t}{k^2d}\right)^d} \tag{Since $\lr{\frac{n}{k}}^k \le \binom{n}{k} \le \lr{\frac{en}{k}}^k$}
	\\&\le (d-1) \cdot \left(\tfrac{ek^2}{\varepsilon}\right)^{d}
\end{align*}
This concludes the proof. 
\end{proof}

\subsection{Step 2\label{steptwo}: Isolating Connected Components of $S^\star$ in $\gconf$} \label{step:isolatecc}
In this section we apply the classical \emph{random separation} technique~\cite{DBLP:conf/iwpec/CaiCC06}
to probabilistically isolate the solution vertices from their $\gconf$-neighbors in any \yes-instance.

Let $f: V(\gconf) \to \{\text{\textcolor{purple}{purple}}, \text{\textcolor{PineGreen}{green}}\}$ be a random coloring in which
each vertex of $\gconf$ is colored independently 
\[
\text{\textcolor{purple}{purple}} \text{~with probability~} p \coloneqq \frac{1}{1+\Delta}\quad\text{and}\quad \text{\textcolor{PineGreen}{green}}\ \text{with probability}\ 1-p=\frac{\Delta}{1+\Delta},
\]
where $\Delta \coloneqq \Delta(\gconf)$. For a coloring $f$, write
$\Vred \coloneqq f^{-1}(\text{\textcolor{purple}{purple}})$ and
$\Vblue \coloneqq f^{-1}(\text{\textcolor{PineGreen}{green}})$. 

Our objective for the coloring is the following “good separation” when the instance is a \yes-instance with witness $S^\star$:
\begin{enumerate}[label=(\roman*),itemsep=0pt,leftmargin=*]
  \item $S^\star \subseteq \Vred$, and
  \item $N_{\gconf}(S^\star) \subseteq \Vblue$.
\end{enumerate}
Since $|S^\star|\le k$ and $|N_{\gconf}(S^\star)| \le \Delta\,k$, we have
\[
  \Pr[\text{good separation}] \;\ge\; p^{|S^\star|}\,(1-p)^{|N_{\gconf}(S^\star)|}
  \;\ge\; \Bigl(\tfrac{1}{1+\Delta}\cdot\bigl(\tfrac{\Delta}{1+\Delta}\bigr)^{\Delta}\Bigr)^{k} \geq \Delta^{-k} \cdot 2^{-\Oh(k)}.
\]



\noindent\textbf{Derandomization.}
The random-coloring step can be derandomized using standard tools such as $(n,p,p^2)$-splitters (see, e.g., Exercises~5.15 and~5.21 in~\cite{DBLP:books/sp/CyganFKLMPPS15}) without affecting the final running time\footnote{Note that $\Delta=\Omega(k^2 d/\varepsilon)$ implies $\log \Delta=\Omega(\log k)$.}. 
Rather than presenting a randomized algorithm, we directly give a deterministic one via an $(n,p,q)$ \emph{lopsided-universal} family of functions.

Let $U$ be the ground set of vertices (so $U=V(\gconf)$ and $n\coloneqq |U|$). 
An $(n,p,q)$ \emph{lopsided-universal} family is a collection $\mathcal{H}=\{f_1,\ldots,f_\ell\}$ of functions $f_i:U\to\{\text{\textcolor{purple}{purple}},\text{\textcolor{PineGreen}{green}}\}$ such that for every $A\in \binom{U}{p}$ and $B\in \binom{U\setminus A}{q}$ there exists $f\in \mathcal{H}$ with
$A \subseteq f^{-1}(\text{\textcolor{purple}{purple}})$ and $B \subseteq f^{-1}(\text{\textcolor{PineGreen}{green}})$.

\begin{proposition}[\cite{AlonYZ95,DBLP:books/sp/CyganFKLMPPS15}]
\label{prop:lopsidedUniversal}
There is an algorithm that, given $n$, $p$, and $q$, constructs an $(n,p,q)$ lopsided-universal family $\mathcal{H}$ of size 
\[
\binom{(p+q)^2}{p}\cdot (p+q)^{\cO(1)}\cdot \log n
\]
in time
\[
\binom{(p+q)^2}{p}\cdot (p+q)^{\cO(1)}\cdot n\log n.
\]
In particular, for $p=k$ and $q=k(d-1)\cdot\!\left(\tfrac{e k^2}{\varepsilon}\right)^{d}$, the size of $\mathcal{H}$ is at most
\[
\left(\tfrac{1}{\varepsilon}\right)^{\cO(dk\log k)}\log n.
\]
\end{proposition}

For our purposes we set
\[
p \coloneqq k
\quad\text{and}\quad
q \coloneqq k(d-1)\!\left(\tfrac{e k^{2}}{\varepsilon}\right)^{d}.
\]
Equivalently, let \(\Delta_{\star} \coloneqq (d-1)\!\left(\tfrac{e k^{2}}{\varepsilon}\right)^{d}\); then \(q = k\,\Delta_{\star}\).
By \Cref{prop:lopsidedUniversal}, this yields an \((n,p,q)\) lopsided-universal family \(\mathcal{H}\) of size
\[
\left(\tfrac{1}{\varepsilon}\right)^{\cO(dk\log k)} \log n.
\]
If in addition \(\Delta(\gconf) \le \Delta_{\star}\), then \(k\,\Delta(\gconf) \le q\), so the same family suffices for enforcing
\(S^\star \subseteq f^{-1}(\text{\textcolor{purple}{purple}})\) and \(N_{\gconf}(S^\star) \subseteq f^{-1}(\text{\textcolor{PineGreen}{green}})\).



\paragraph*{Filtering of Connected Components.} For every coloring $f\in \mathcal{H}$ we proceed as follows. Delete all \textcolor{PineGreen}{green}
vertices, i.e., set $\Vblue \coloneqq f^{-1}(\text{\textcolor{PineGreen}{green}})$ and form
$\gconf^{f} \coloneqq \gconf - \Vblue$. The remaining graph $\gconf^{f}$ is a disjoint union of
connected components, each an induced subgraph of $\gconf$ on the \text{\textcolor{purple}{purple}}
vertices $\Vred \coloneqq f^{-1}(\text{\textcolor{purple}{purple}})$. Finally, discard every
\text{\textcolor{purple}{purple}} component $C$ with $|C|>k$; let $\mathcal{C}_{\le k}(f)$ denote the
family of surviving components (those with $|C|\le k$).

\smallskip
\noindent
\noindent\textbf{Motivation for the filtering step.}
Assume the instance is a \yes-instance and let $S^\star\subseteq R$, $|S^\star|\le k$, satisfy \propone\ and \proptwo. 
While $\gcon[S^\star]$ is connected, the conflict graph $\gconf$ need not be a subgraph of $\gcon$: there may exist $u,v\in R$ that are nonadjacent in $\gcon$ but adjacent in $\gconf$ due to a large overlap of their neighborhoods in $\gcov$ (so $uv\in E(\gconf)$ but $uv\notin E(\gcon)$). 
Conversely, $\gconf[S^\star]$ need not be connected even though $\gcon[S^\star]$ is: some pairs $u,v\in S^\star$ that are adjacent in $\gcon$ can become nonadjacent in $\gconf$ by the very definition of the conflict graph (so $uv\in E(\gcon)$ but $uv\notin E(\gconf)$).

Let $C_1^\star,\dots,C_\ell^\star$ be the connected components of $\gconf[S^\star]$, where
$1\le \ell\le |S^\star|\le k$. Clearly $S^\star=\bigcup_{i=1}^\ell V(C_i^\star)$. 
(When convenient, we do not distinguish between a component $C$ and its vertex set $V(C)$.)

By the choice of the lopsided-universal family $\mathcal{H}$, there exists $f\in\mathcal{H}$ with $S^\star\subseteq \Vred$ and $N_{\gconf}(S^\star)\subseteq \Vblue$. 
For this \emph{good} $f$, deleting $\Vblue$ removes no vertex of $S^\star$, and the purple subgraph
$\gconf^{f}=\gconf-\Vblue$ still contains $\gconf[S^\star]$ unchanged; in particular,
each $C_i^\star$ remains a connected purple component. 
Since $|V(C_i^\star)|\le |S^\star|\le k$, all these components survive the size filter. 
Thus, the filtering step preserves every piece of $S^\star$ we care about while discarding large purple components that cannot be part of any size-$k$ solution.




\medskip
 \noindent 
 \underline{\textbf{Constructing $\gconf^{\mathrm{purple}}:$\label{graphwithredvertices}}} 
Recall that $\mathcal{C}_{\le k}(f)$ denotes the family of surviving (purple) components, i.e.,
the connected components of $\gconf - \Vblue$ whose sizes are at most $k$.
Fix a coloring $f\in\mathcal{H}$ and enumerate
\[
\cC \;\coloneqq\; \mathcal{C}_{\le k}(f) \;=\; \{C_1,C_2,\dots,C_s\},
\quad\text{so that}\quad 1\le |C_i|\le k\ \text{for all }i\in[s].
\]
We refer to the filtered purple subgraph as
\[
\gconf^{\mathrm{purple}} \;\coloneqq\; \gconf[\Vred]\;-\!\!\!\!\bigcup_{\substack{C\subseteq \Vred\\ C\text{ is a purple component}\\ |C|>k}} \!\!\!\! C,
\]
so each $C_i\in\cC$ is a connected component of $\gconf^{\mathrm{purple}}$ with $1\le |C_i|\le k$.
(When clear from context, we drop the explicit dependence on $f$.) 
The next observation follows from the construction of $\gconf^{\mathrm{purple}}$.




\begin{observation} \label{obs:subsetvertices}
    $V\lr{\gconf^{\mathrm{purple}}} \subseteq V\lr{\gconf} = V\lr{\gcon} =R$.
\end{observation}

\subsection{Step 3\label{stepthree}: Construction of Sparsified Graph} \label{subsubsec:spargraph}

The goal here is to compute a bipartite graph which is called as \textit{sparsified graph} (will be defined shortly) which will eventually help us find an approximate solution. To this end, we first consider the induced bipartite subgraph of $\gcov$ defined on the vertices $V\lr{\gconf^{\mathrm{purple}}}$ and $B$, namely $\gcov\bigl[V\left(\gconf^{\mathrm{purple}}\right)\uplus B\bigr]$. Based on $\gcov\bigl[V\left(\gconf^{\mathrm{purple}}\right) \uplus B\bigr]$ and $\gconf^{\mathrm{purple}}$, we construct another bipartite graph, called \textit{sparsified graph}, denoted by $\gspar$ as follows:

\medskip
\noindent
 \begin{tcolorbox}[colback=green!5!white,colframe=blue!70!white, title=\textbf{\large{Construction of $\gspar$}}]
\begin{itemize}
    \item $V\lr{\gspar} \coloneqq V\lr{\gconf^{\mathrm{purple}}} \uplus B$ (two parts of the bipartite graph).
     
    \item Let $\cC = \lrc{C_1, \dots C_s}$ be the set of components in $\gconf^{\mathrm{purple}}$. Now we have the following Sparsification process which helps to construct the edges in the sparsified graph.
    
    \textbf{{\color{blue}Sparsification Process:}\label{process}} For every $b \in B$ and every $C_i \in \cC$, let $E_b$ denote the set of edges of $\gcov$ with one endpoint at $b$ and the other endpoint at a neighbor of $b$ in the component $C_i$ in  graph $\gcov$.
    More formally, $E_b \coloneqq \lrc{xb \in E(\gcov) \mid x \in V(C_i) \cap N_{\gcov}(b)}$. 
    
    \hspace{2em}Now we construct $E(\gspar)$ as follows: for every $b \in B$, if $E_b \neq \emptyset$, then add an arbitrary edge from $E_b$ to $E(\gspar)$.
\end{itemize}
\end{tcolorbox}

We now summarize few key properties of $\gspar$.

\begin{tcolorbox}[colback=yellow!5!white,colframe=Brown!75!white,title=\textbf{\large{Properties of $\gspar$}}]

Let $\cC = \big\{C_1, \dots C_s\big\}$ be the set of connected components in $\gconf^{\mathrm{purple}}$ and $\gspar$ be the graph defined above. 

 \begin{property}\label{prop:singlecomp}
     For every connected component $C_i \in \cC$ of $\gconf^{\mathrm{purple}}$ and every subset of vertices $X \subseteq V(C_i)$, it holds that, 
    \begin{align}
    \text{(i)}\quad & \big| N_{\gspar}(X)\big| = \sum_{x \in X}\big|N_{\gspar}(x)\big| = \sum_{x \in X}\deg_{\gspar}(x) 
    \label{eqn:exact} \\
    \text{(ii)}\quad & N_{\gspar}(X) \subseteq N_{\gcov}(X) 
    && \text{$\lr{\text{if}~~ X \subset V(C_i)}$} 
    \label{eqn:propersubsetneighborhood} \\
    \text{(iii)}\quad & N_{\gspar}(V(C_i)) = N_{\gcov}(V(C_i)) 
    && \text{$\lr{\text{if}~~ X = V(C_i)}$} 
    \label{eqn:equalneighborhood}
\end{align}

 \end{property}
\tcblower
  \begin{property}\label{prop:doublecomp}
     Consider any pair of vertices $x \in V(C_i)$ and $y \in V(C_j)$ where $i, j \in [s]$. It holds that, 
     \begin{align}
         &\text{(i)}~~\big| N_{\gspar}(x) \cap N_{\gspar}(y)\big| < \frac{\varepsilon t}{k^2} \quad \text{ when $i \neq j$ and}\label{eqn:smallintersection}\\
         &\text{(ii)}~~N_{\gspar}(x) \cap N_{\gspar}(y) = \emptyset \quad \text{ when $i = j$}\label{eqn:emptyintersection}
     \end{align}
 \end{property}
\end{tcolorbox}


\begin{proof}[\textbf{Proof for  \Cref{prop:singlecomp}}]
(i) \Cref{eqn:exact} follows directly from the construction of $\gspar$, 
obtained as a result of \hyperref[process]{Sparsification Process}. 
More specifically, it follows that for every vertex $b \in B$ and for every component $C_i \in \cC$, 
the vertex $b$ is adjacent to at most one vertex of $C_i$ in $\gspar$. 
Thus, the vertex $b$ contributes either $0$ or $1$ to each of the three terms in \Cref{eqn:exact}, and in each case, its contribution is same. This establishes the desired property.

(ii) It is clear that $N_{\gspar}(X) \subseteq N_{\gcov}(X)$, because every vertex $v \in N_{\gspar}(X)$ must also belong to $N_{\gcov}(X)$, given that $\gspar$ is a subgraph of $\gcov$. 

(iii) Moreover, when $X = V(C_i)$ and $b \in B$ 
is a vertex that has a neighbor in $C_i$ in the graph $\gspar$, then by the definition of the \hyperref[process]{Sparsification Process}, 
we have that $b$ is in fact adjacent to exactly one vertex of $C_i$ in $\gspar$. 
Thus $b$ contributes $1$ to both sides of \Cref{eqn:equalneighborhood}.
\end{proof}


\begin{proof}[\textbf{Proof for \Cref{prop:doublecomp}}]
(i) \Cref{eqn:smallintersection} follows directly from the property of $\gconf$. Note that, $N_{\gspar}(x) \cap N_{\gspar}(y) \subseteq  N_{\gcov}(x) \cap N_{\gcov}(y)$ since $\gspar$ is a subgraph of $\gcov$. Further since $x \in V(C_i), y \in V(C_j)$ for $i \neq j$, we have that $x$ and $y$ are non-adjacent in  $\gconf$, this implies that $\big| N_{\gcov}(x) \cap N_{\gcov}(y) \big| < \frac{\varepsilon t}{k^2}$ (by definition) which further implies that $\big|N_{\gspar}(x) \cap N_{\gspar}(y)\big| < \frac{\varepsilon t}{k^2}$.

\smallskip
(ii) Suppose, for the sake of contradiction, there exists a pair of vertices $x, y \in V(C_i)$ for some $C_i$ satisfying $N_{\gspar}(x) \cap N_{\gspar}(y) \neq \emptyset$. Let $z \in N_{\gspar}(x) \cap N_{\gspar}(y)$. 
This contradicts the property of $\gspar$, 
since $\gspar$ is obtained as a result of \hyperref[process]{Sparsification Process}, 
which ensures that both edges $xz$ and $yz$ cannot simultaneously appear in $E(\gspar)$.
\end{proof}

In the next section, we assign weights to the vertices of $\gconf^{\mathrm{purple}}$, 
where each weight reflects the number of vertices dominated by that vertex in $\gspar$.

\subsection{Step 4\label{stepfour}: Introducing vertex weights to $V\bigl(\gconf^{\mathrm{purple}}\bigr)$}\label{subsec:introweight}

In this section, we assign weights to the vertices of the graph $\gconf^{\mathrm{purple}}$. Let $w : V\bigl(\gconf^{\mathrm{purple}}\bigr) \to \mathbb{Z}_{\geq 0}$ be a weight function defined on $V\bigl(\gconf^{\mathrm{purple}}\bigr)$ as follows: 

\begin{tcolorbox}[colback=green!5!white,colframe=blue!75!black]
\begin{equation}\label{eqn:weightfunction}
    w(v) = \big|N_{\gspar}(v)\big| = \deg_{\gspar}(v)\quad \quad \text{for every } v \in V\lr{\gconf^{\mathrm{purple}}}
\end{equation}
\end{tcolorbox}

By  \Cref{obs:subsetvertices}, we have 
$V(\gconf^{\mathrm{purple}}) \subseteq V(\gconf) = V(\gcon) = R$. 
Therefore, the assigned weights are restricted to a subset of the vertices in $R$. 
Before proceeding further, it is important to note the following assumptions, which will be maintained 
throughout the subsequent discussion.



We now state a crucial lemma that, via the weight function defined above, provides a lower bound on the $\gcov$-coverage of any set \(Y\subseteq V(\gconf^{\mathrm{purple}})\) with \(|Y|\le k\); equivalently, it lower-bounds \(|N_{\gcov}(Y)|\) for all such \(Y\).

\begin{lemma}\label{lem:mainlemma}
    For any subset $Y \subseteq V\bigl(\gconf^{\mathrm{purple}}\bigr)$ of size at most $k$, it holds that 
    \begin{align}
            \big|N_{\gcov}(Y) \big| \geq \big|N_{\gspar}(Y) \big| \geq \sum \limits_{y \in Y} w(y) - \varepsilon t \label{eqn:formainlemma}
    \end{align}
\end{lemma}
\begin{proof} 
By Inclusion-exclusion Principle, we have 
\begin{eqnarray}
\big| N_{\gspar}(Y) \big| 
  &=& \sum\limits_{y \in Y} \big|N_{\gspar}(y)\big|
   - \sum\limits_{x, y \in Y} \big| N_{\gspar}(x) \cap N_{\gcov}(y) \big|\nonumber\\
    &+& \sum\limits_{x, y,z \in Y}\big| N_{\gspar}(x) \cap N_{\gspar}(y) \cap N_{\gspar}(z) \big|- \dotsm \label{eqn:formulaincluexclu}\\
    &\stackrel{(\P)}{\geq}& \sum\limits_{y \in Y}\big|N_{\gspar}(y)\big| - \sum\limits_{x, y \in Y}\big| N_{\gspar}(x) \cap N_{\gspar}(y)\big|\label{line:incluexclu} \\
    &=&  \sum\limits_{y \in Y}w(y) - \sum\limits_{x, y \in Y}\big| N_{\gspar}(x) \cap N_{\gspar}(y)\big| \label{line:weightneighbor}
\end{eqnarray}

Here, inequality $(\P)$ holds because, in Line~\ref{line:incluexclu}, we truncate the Inclusion--Exclusion formula by discarding all terms from the $3^{\textnormal{rd}}$ term onward, thereby obtaining a lower bound on the original expression. Line~\ref{line:weightneighbor} follows from \Cref{eqn:weightfunction}. Now, it remains to provide an upper bound on $\mathlarger{\sum}\limits_{x, y \in Y} \big| N_{\gspar}(x) \cap N_{\gspar}(y) \big|$. 


Recall that $\cC = \lrc{C_1, \dots C_s}$ is the set of connected components in $\gconf^{\mathrm{purple}}$. 
For each $i\in[s]$, consider the restriction of the connected component $C_i$ to $Y$, namely $C_i^{Y}\coloneqq V(C_i)\cap Y$. Note that $\gconf^{\mathrm{purple}}[C_i^{Y}]$ need not be connected. Without loss of generality, we assume that for every $i 
\in [s]$, $C_i^Y \neq \emptyset$, i.e., it contains at least one vertex from $Y$; otherwise, we discard the empty components from the collection $\lrc{C_i^Y , \dots , C_s^Y}$. Thus we can write,

\begin{sloppypar}
\begin{align}\label{eqn:intersection}
    \sum\limits_{x, y \in Y}\big| N_{\gspar}(x) \cap N_{\gspar}(y)\big| = 
    \sum\limits_{\substack{i \in [s],\\ x,y \in C^Y_i}}\big| 
    N_{\gspar}(x) \cap N_{\gspar}(y)\big| 
    + \sum\limits_{\substack{i,j \in [s],i\neq j\\ x\in C^Y_i,\\y\in C^Y_j}}\big| N_{\gspar}(x) \cap N_{\gspar}(y)\big|
\end{align}
\end{sloppypar}

Notice that due to \Cref{eqn:emptyintersection} of \Cref{prop:doublecomp}, we have that for every $C_i \in \cC$ and for every $x,y~ \in V(C_i)$, it holds that $N_{\gspar}(x) \cap N_{\gspar}(y) = \emptyset$. Hence the first term of \Cref{eqn:intersection} becomes $0$. 

Moreover, by \Cref{eqn:smallintersection} in \Cref{prop:doublecomp}, for all distinct $i,j\in[s]$ and all vertices $x\in C_i^{Y}\subseteq V(C_i)$ and $y\in C_j^{Y}\subseteq V(C_j)$, we have
\[
\bigl|N_{\gspar}(x)\cap N_{\gspar}(y)\bigr| \;<\; \frac{\varepsilon t}{k^{2}},
\]
since $x$ and $y$ are nonadjacent in the conflict graph $\gconf$. Since $|Y| \leq k$, there are at most $\binom{k}{2}$ pairs in $Y$. Hence the second component of
\Cref{eqn:intersection} is upper bounded by $\binom{k}{2} \cdot \frac{\varepsilon t}{k^2} \leq k^2\cdot \frac{\varepsilon t}{k^2} = \varepsilon t$. Substituting these values in \Cref{eqn:intersection}, we obtain 

$$ \sum\limits_{x, y~ \in Y}\big| N_{\gspar}(x) \cap N_{\gspar}(y)\big| \leq\varepsilon t$$ 

Substituting all these values in \Cref{eqn:formainlemma}, we get that $\big|N_{\gspar}(Y) \big| \geq \sum \limits_{y \in Y} w(y) - \varepsilon t$. 
Furthermore, since $\gspar$ is a subgraph of $\gcov$, we have
\[
\bigl|N_{\gcov}(Y)\bigr| \;\ge\; \bigl|N_{\gspar}(Y)\bigr|.
\]
(Equivalently, this follows by repeatedly applying \Cref{eqn:propersubsetneighborhood} from \Cref{prop:singlecomp}.)
This proves the lemma.
\end{proof}

We now require a converse to \Cref{lem:mainlemma}: an upper bound on the $\gcov$-coverage of any set \(X\subseteq V(\gconf^{\mathrm{purple}})\) with \(|X|\le k\), expressed in terms of the weight
function defined above; equivalently, we seek to upper bound \(|N_{\gcov}(X)|\) for all such \(X\).
However, we cannot establish this bound for arbitrary \(X\). Instead, we prove it for a particular
class of sets, which is sufficient for our purposes. The next definition characterizes this class.

\begin{definition}[Component-respecting set]
{\em Let $\cC=\{C_1,\dots,C_s\}$ be the collection of all connected components of $\gconf^{\mathrm{purple}}$.
A set $Z \subseteq V(\gconf^{\mathrm{purple}})$ is \emph{component-respecting} $\lr{\textnormal{w.r.t.}~\gconf^{\mathrm{purple}}}$ if for every $i\in[s]$, either $C_i \subseteq Z$ or $Z \cap C_i = \emptyset$.
Equivalently, $Z$ is a union of some components  from $\cC$, in other words, there exists $I \subseteq [s]$ with
$Z=\bigcup_{i\in I} C_i$.}
\end{definition}

 \begin{lemma}\label{claim:weightneighbor}
        For any component-respecting set $X \subseteq V\lr{\gconf^{\mathrm{purple}}}$, it holds that
        \begin{align}
            \sum\limits_{x \in X}w(x) \geq \big| N_{\gcov}(X)\big| \label{eqn:weightneighbor}
        \end{align}
    \end{lemma}

    \begin{proof} We know that,
        \begin{eqnarray}
            && \sum\limits_{x \in X}w(x) = \sum\limits_{x \in X}\big| N_{\gspar}(x)\big| \quad \quad \quad \quad \textnormal{(due to the definition in \Cref{eqn:weightfunction})}
        \end{eqnarray}
        Thus, it suffices to show that 
        \begin{align} 
            \sum\limits_{x \in X}\big| N_{\gspar}(x)\big| \geq \big| N_{\gcov}(X)\big| \label{line:neighbor-neighbor}
        \end{align}

    Every vertex $b \in  N_{\gcov}(X) \subseteq B$ contributes $1$ to the right-hand side of ~\Cref{line:neighbor-neighbor}, since every neighbor of $X$ in $\gcov$ is counted once in $\big| N_{\gcov}(X) \big|$.  To establish \Cref{claim:weightneighbor}, it suffices to show that every vertex $b \in  N_{\gcov}(X) \subseteq B$ contributes at least $1$ to the left-hand side of \Cref{line:neighbor-neighbor}. Recall that $\cC = \lrc{C_1, \dots, C_s}$ denotes the set of connected components in $\gconf^{\mathrm{purple}}$. Consider an arbitrary vertex $b 
    \in N_{\gcov}(X)$ and an arbitrary component $C_i \in \cC$ of $\gconf^{\mathrm{purple}}$ such that $ b \in N_{\gcov}(V(C_i))$. Note that such a component exists since $b \in N_{\gcov}(X)$. Since $X$ is component-respecting with respect to $\gconf^{\mathrm{purple}}$, it holds that 
    $V(C_i) \cap X = V(C_i)$, i.e., $V(C_i) \subseteq X$. Thus, by \Cref{eqn:equalneighborhood}, we have 
    $N_{\gspar}\!\bigl(V(C_i)\bigr) = N_{\gcov}\!\bigl(V(C_i)\bigr)$. Hence, by the construction of $\gspar$, we have that $b$ is adjacent to some $x \in V(C_i)$ in $\gspar$, i.e., $b \in N_{\gspar}(x)$.  Hence \(b\) contributes \(1\) to \(|N_{\gspar}(x)|\) in the sum on the left-hand side of \Cref{line:neighbor-neighbor} which concludes the Lemma.  
    \end{proof}

\subsection{Step 5\label{stepfive}: Finding a Maximum Weighted Subtree in $\gcon\lrsq{V(\gconf^{\mathrm{purple}})}$}\label{sec:bisppvd}
In this section we work in the vertex-weighted graph \(\gcon\!\lrsq{V(\gconf^{\mathrm{purple}})}\), endowed with
the weight function \(w\) defined above (in \Cref{eqn:weightfunction}). Our goal is to find a subtree on at most  \(k\) vertices
maximizing total weight (the weight of a tree is the sum of the weights of its vertices).  We formalize this as the following subproblem.

\begin{description}
\item[\underline{\sc Maximum-Weight \(k\)-Tree}:]
Given a graph \(G\), a weight function \(\gamma:V(G)\to\mathbb{Z}_{\ge 0}\), and an integer \(k\),
find a subtree \(T_G\subseteq G\) with \(|V(T_G)|\leq k\) maximizing
\(\gamma(T_G)\coloneqq \sum_{v\in V(T_G)} \gamma(v)\).
\end{description}

We solve this via \textsc{Weighted Tree Isomorphism}.

\begin{description}
\item[\underline{\sc Weighted Tree Isomorphism}:]
Given a host graph \(G\), a tree \(T\), and a weight function \(\gamma:V(G)\to\mathbb{Z}_{\ge 0}\),
find (if it exists) a subgraph \(T_G\subseteq G\) isomorphic to \(T\) and of maximum total weight
\(\gamma(T_G)\).
\end{description}

\begin{proposition} 
\label{thm:weightedcsi}
{\sc Weighted Tree Isomorphism} can be solved in time \(2^{\OO(h)}\,n^{\cO(1)}\), where
\(h=|V(T)|\) and \(n=|V(G)|\).
\end{proposition}

\Cref{thm:weightedcsi} follows from standard techniques. One route is the classic color-coding framework of Alon, Yuster, and Zwick~\cite[Thm.~6.3]{AlonYZ95}, whose dynamic program extends to the weighted objective by maximizing (rather than merely detecting) over states; see the discussion following \cite[Thm.~6.3]{AlonYZ95} (cf.\ also \cite{DBLP:conf/wg/PlehnV90}). Alternatively, representative-family techniques yield the same running time with a weight-aware state space~\cite{DBLP:journals/jacm/FominLPS16}.

To solve \textsc{Maximum-Weight \(k\)-Tree}, we enumerate all non-isomorphic trees on at most \(k\) vertices
and, for each such tree \(T\), run \textsc{Weighted Tree Isomorphism}
(\Cref{thm:weightedcsi}). Otter~\cite{otter1948number} showed that the number of non-isomorphic
(unrooted) trees on \(h\) vertices is \(t_h = 2.956^{h}\). Moreover, all non-isomorphic rooted
trees on \(h\) vertices can be generated in time \(\cO(t_h\,h)\) by the algorithm of
Beyer and Hedetniemi~\cite{DBLP:journals/siamcomp/BeyerH80}.

\begin{proposition}[\cite{otter1948number,DBLP:journals/siamcomp/BeyerH80}]
\label{prop:enumtree}
The number of non-isomorphic trees on \(h\) vertices is \(t_h = 2.956^{h}\).
Furthermore, all non-isomorphic rooted trees on \(h\) vertices can be enumerated in time
\(\cO(t_h\,h)\).
\end{proposition}

Combining \Cref{prop:enumtree} with \Cref{thm:weightedcsi} yields:

\begin{lemma}\label{lem:maxweight}
{\sc Maximum-Weight \(k\)-Tree} can be solved in time \(2^{\cO(k)}\,n^{\cO(1)}\).
\end{lemma}

\begin{proof}
There are \(t_{k^\star} = 2.956^{k^\star}\), $k^\star \leq k$, non-isomorphic trees on \(k^\star\) vertices. For each such tree \(T\),
\textsc{Weighted Tree Isomorphism} runs in \(2^{\cO(k^\star)}\,n^{\cO(1)}\) time
(\Cref{thm:weightedcsi}). Hence the total running time is
\(t_k \cdot 2^{\cO(k)}\,n^{\cO(1)} = 2^{\cO(k)}\,n^{\cO(1)}\).
\end{proof}


\subsection{Putting all the Pieces Together}
By assembling all the components and present the complete algorithm as pseudocode in \Cref{alg:conprbdspart}, we now obtain the main theorem of this section, and prove it by combining the tools and ideas developed above.

\begin{algorithm}[ht!]
\caption{\texttt{AlgPartialConnRBDS} \label{alg:conprbdspart}}
\KwInput{Coverage graph $\gcov = (R \uplus B, E')$ which is \kddfree, connectivity graph $\gcon = (R, E)$, integers $k, t, \varepsilon \in (0,1)$}
\KwOutput{A subset $S \subseteq R$ of size at most $k$ such that 
\propone$~\gcon[S]$ is connected, 
\proptwo$~|N_{\gcov}(S)| \geq \lr{1 - \varepsilon}t$ or returns $\bot$ \blue{\tcp*{\footnotesize{$\bot$ indicates no such set of size $k$ can dominate $t$ vertices.}}}}
  \uIf{$t \leq \frac{k^2d}{\varepsilon}$}{Run the exact algorithm of \Cref{thm:fptbytpds} and \Return that solution.}
  \Else{
  Construct the conflict graph $\gconf$ (described in \hyperref[stepone]{Step~1} \Cref{subsec:conflictgraph}) \\
  Set $\Delta \gets (d-1) \cdot \left(\tfrac{ek^2}{\varepsilon}\right)^{d}$ \\
  Let $\cH$ be the family of two-coloring functions obtained using \Cref{prop:lopsidedUniversal}, (\Cref{step:isolatecc})\\
  \ForEach{$f \in \cH$}{
    Compute $\gconf^{\mathrm{purple}}$ after color filtration (described in \hyperref[steptwo]{Step~2}, \Cref{step:isolatecc}). \\
    Build the sparsified graph $\gspar$ (described in \hyperref[stepthree]{Step~3}, \Cref{subsubsec:spargraph}). \\
    Assign weights to vertices of $\gconf^{\mathrm{purple}}$ (described in \hyperref[stepfour]{Step~4}, \Cref{subsec:introweight}). \\
   Let $T_f$ be a   maximum-weight subtree in $\gcon[V(\gconf^{\mathrm{purple}})]$  obtained using \Cref{lem:maxweight} (see \hyperref[stepfive]{Step~5} in \Cref{sec:bisppvd}).
    }
    Let $\widetilde{T}$ be the maximum-weight subtree in $\gcon[V(\gconf^{\mathrm{purple}})]$ among all $T_f, f \in \cH$. \\
    \uIf{$w(\widetilde{T}) < t$}{
        \Return $\bot$
    }
    \Else{\Return $V(\widetilde{T})$}}
\end{algorithm}

\begin{restatable}{theorem}{kddfptapproxpds} \label{thm:kddfptapproxpds}
Let $\cI=(\gcov,\gcon,k,t)$ be an instance of \pcrbdsshort, where $\gcov$ and $\gcon$ are the
coverage and connectivity graphs, respectively. If $\gcov$ is \kddfree, then for every
$\varepsilon\in(0,1)$ there is an algorithm running in time
$\os\!\bigl(2^{\cO(k^{2}d/\varepsilon)}\bigr)$ that either
\begin{enumerate}[label=(\roman*),itemsep=0pt,leftmargin=*]
  \item outputs a set $S\subseteq R$ with $|S|\le k$, $\gcon[S]$ connected, and
        $|N_{\gcov}(S)|\ge (1-\varepsilon)t$, or
  \item correctly concludes that no size-$k$ set dominates at least $t$ vertices in $\gcov$.
\end{enumerate}
\end{restatable}

\begin{proof}
We first show that, if the instance is a \yes-instance, our algorithm returns a set \(S\subseteq R\) with \(|S|\le k\), \(\gcon[S]\) connected, and \(|N_{\gcov}(S)|\ge (1-\varepsilon)t\). Let \(S^\star\subseteq R\) be a size-\(\le k\) witness satisfying \propone\ (i.e., \(\gcon[S^\star]\) is connected) and \proptwo\ (i.e., \(|N_{\gcov}(S^\star)|\ge t\)). 

\medskip 
\noindent 
\emph{Small $t$.} If $t<\tfrac{k^{2}d}{\varepsilon}$, algorithm invokes the exact FPT algorithm of
\Cref{thm:fptbytpds} to decide and, if possible, return a size-$\le k$ connected solution covering
at least $t$ vertices.

\medskip 
\noindent
\emph{Large $t$.} Assume $t \ge \tfrac{k^{2}d}{\varepsilon}$. The algorithm implements the purple/green
separation by enumerating an $(n,p,q)$ lopsided-universal family $\mathcal{H}$
(\Cref{prop:lopsidedUniversal}) with $p=k$ and
$q = k(d-1)\!\left(\tfrac{e k^{2}}{\varepsilon}\right)^{d}$. By the defining property of such
families, there exists a \emph{good} $f\in\mathcal{H}$ with
$S^\star \subseteq \Vred$ and $N_{\gconf}(S^\star) \subseteq \Vblue$. Fix this $f$.

We delete all green vertices and retain only purple components of size at most $k$ (the filtering step);
this preserves each connected piece $C_i^\star$ of $\gconf[S^\star]$. We then define the
vertex-weight function $w$ on $V(\gconf^{\mathrm{purple}})$ (equivalently, on the host
$\gcon[V(\gconf^{\mathrm{purple}})]$), as in \hyperref[stepfour]{Step~4} (\Cref{subsec:introweight}).
Finally, let $T_f$ be a maximum-weight subtree on  $k$ vertices in
$\gcon[V(\gconf^{\mathrm{purple}})]$, computed via \Cref{lem:maxweight}.

We can {\em potentially} take \(S \coloneqq V(T_f)\). Since \(V(T_f)\subseteq V(\gconf^{\mathrm{purple}})\) and \(|V(T_f)|=k\), by \Cref{lem:mainlemma} we have
\begin{equation}
  \bigl|N_{\gcov}(V(T_f))\bigr| \;\ge\; \sum_{\mathfrak{f}\in V(T_f)} w(\mathfrak{f}) \;-\; \varepsilon t .
  \label{eqn:towardslemma}
\end{equation}

Recall that \(S^\star\) has \(|S^\star|\le k\) and \(\gcon[S^\star]\) connected; let \(T^\star\) be any spanning tree of \(\gcon[S^\star]\).
Since \(S^\star \subseteq \Vred\) and each connected piece \(C_i^\star\) has size at most \(k\), the filtering step preserves \(\gcon[S^\star]\), hence \(T^\star \subseteq \gcon[V(\gconf^{\mathrm{purple}})]\).  
Because \(T_f\) is a maximum-weight \(k\)-vertex tree in \(\gcon[V(\gconf^{\mathrm{purple}})]\),
\begin{equation}
  \sum_{\mathfrak{f}\in V(T_f)} w(\mathfrak{f})
  \;\ge\; \sum_{v\in V(T^\star)} w(v)
  \;\ge\; \sum_{\mathfrak{s}\in S^\star} w(\mathfrak{s}) .
  \label{line:greatereq}
\end{equation}

Moreover, by \Cref{claim:weightneighbor} applied to \(S^\star\) and using \(|N_{\gcov}(S^\star)|\ge t\),
\begin{equation}
  \sum_{\mathfrak{s}\in S^\star} w(\mathfrak{s}) \;\ge\; \bigl|N_{\gcov}(S^\star)\bigr| \;\ge\; t .
  \label{line:optgreater}
\end{equation}

Substituting \eqref{line:greatereq} and \eqref{line:optgreater} into \eqref{eqn:towardslemma} yields
\begin{align*}
  \bigl|N_{\gcov}(V(T_f))\bigr|
  &\ge \sum_{\mathfrak{f}\in V(T_f)} w(\mathfrak{f}) - \varepsilon t \\
  &\ge \sum_{\mathfrak{s}\in S^\star} w(\mathfrak{s}) - \varepsilon t \\
  &\ge t - \varepsilon t
   \;=\; (1-\varepsilon)t .
\end{align*}
Thus \(S=V(T_f)\) satisfies \(|S|\le k\), \(\gcon[S]\) is connected, and \(|N_{\gcov}(S)|\ge (1-\varepsilon)t\).

Our algorithm outputs either \(\widetilde{T}\), the maximum-weight \(k\)-vertex subtree in
\(\gcon[V(\gconf^{\mathrm{purple}})]\) over all colorings \(f\in\cH\), or \(\bot\) if
\(w(\widetilde{T})<t\). Since there exists a good \(f\) with
\(w(T_f)\ge \sum_{\mathfrak{s}\in S^\star} w(\mathfrak{s}) \ge t\), we have
\(w(\widetilde{T}) \ge w(T_f) \ge t\), and thus the algorithm never returns \(\bot\).
Set \(S \coloneqq V(\widetilde{T})\). Because \(V(\widetilde{T})\subseteq V(\gconf^{\mathrm{purple}})\) and
\(|V(\widetilde{T})|=k\), by \Cref{lem:mainlemma} we obtain
\begin{align*}
  \bigl|N_{\gcov}(V(\widetilde{T}))\bigr|
  &\ge \sum_{\mathfrak{f}\in V(\widetilde{T})} w(\mathfrak{f}) - \varepsilon t \\
  &\ge \sum_{\mathfrak{f}\in V(T_f)} w(\mathfrak{f}) - \varepsilon t \\
  &\ge \sum_{\mathfrak{s}\in S^\star} w(\mathfrak{s}) - \varepsilon t \\
  &\ge t - \varepsilon t
   \;=\; (1-\varepsilon)t .
\end{align*}
Hence \(S\) satisfies the coverage guarantee.  Since our guarantee is conditioned on \yes-instances and we do not require bounds for the \no-instance case, this completes the description of the algorithm’s correctness and approximation analysis.

\paragraph*{Running time.}
\emph{Small-$t$ case.} When $t<\tfrac{k^{2}d}{\varepsilon}$ we invoke the exact routine of
\Cref{thm:fptbytpds}, which runs in time
\[
T_{\text{small}}=\os\!\bigl(2^{\cO(t)}\bigr)
\;\le\;
\os\!\bigl(2^{\cO(k^{2}d/\varepsilon)}\bigr).
\]

\emph{Large-$t$ case.} For $t\ge \tfrac{k^{2}d}{\varepsilon}$ we enumerate an $(n,p,q)$
lopsided-universal family $\cH$ (\Cref{prop:lopsidedUniversal}) with
$p=k$ and $q=k(d-1)\!\left(\tfrac{e k^{2}}{\varepsilon}\right)^{d}$, whose size satisfies
\(|\cH|\le \left(\tfrac{1}{\varepsilon}\right)^{\cO(dk\log k)}\log n\).
For each \(f\in\cH\), all preprocessing—forming \(\gconf^{\mathrm{purple}}\),
filtering to components of size \(\le k\), and computing the needed sparsifiers—is polynomial in \(n\).
The only exponential step per \(f\) is solving \textsc{Maximum-Weight \(k\)-Tree} in the host
\(\gcon[V(\gconf^{\mathrm{purple}})]\), which takes \(2^{\cO(k)}n^{\cO(1)}\) time
(\Cref{lem:maxweight}). Hence
\[
T_{\text{large}}
=\os\!\Bigl(|\cH|\cdot 2^{\cO(k)}\Bigr)
=\os\!\Bigl(\left(\tfrac{1}{\varepsilon}\right)^{\cO(dk\log k)} \cdot 2^{\cO(k)}\Bigr).
\]

\emph{Combined bound.} Taking the worse of the two regimes and hiding polynomial factors,
the overall running time is
\[
\os\!\bigl(2^{\cO(k^{2}d/\varepsilon)}\bigr),
\]
which matches the bound stated in \Cref{thm:kddfptapproxpds}. This conludes the proof. 
\end{proof}

\subsection{Extension to a Steiner version} \label{subsec:steinerext}
In this section, we design an {\sf EPAS} for a generalization of \pcrbdsshort problem, which we term \stpcrbdsfull (\stpcrbdsshort). While the primary reason for this is to use it as a subroutine in our {\sf PAS} for \pcrbdsshort in \Cref{sec:fptapproxfork}, we also believe that the problem is interesting in and of itself. The formal definition of the problem follows.
\begin{tcolorbox}[enhanced,title={\color{black} {\stpcrbdsfull(\stpcrbdsshort)}}, colback=white, boxrule=0.4pt,
    	attach boxed title to top center={xshift=-4mm, yshift*=-2.5mm},
    	boxed title style={size=small,frame hidden,colback=white}]
    	
    	\textbf{Input:} An instance $\cI = (\gcon, \gcov, T, k, t)$, where
	\begin{itemize}
         \item $\gcon = (R, E)$ is an arbitrary graph, called the \emph{\color{bblue}connectivity graph},
	\item $ \gcov= (R \uplus B, E')$ is a bipartite            graph, called the \emph{\color{bblue}coverage graph},
        \item a {\em terminal set} $T \subseteq R$, and
	\item $k$ and $t$ are non-negative integers. 
	\end{itemize}
    	\textbf{Question:} \hspace*{.06cm} Does there exist a vertex subset $S \subseteq R$ such that,
	\begin{enumerate}[label=(\arabic*)]
            \item \label{one} $\big|S\big| \le k$,
            \item \label{two} $T \subseteq S$,
		\item \label{three} $\gcon[S]$ is connected, and
		\item \label{four} $\big|N_{\gcov}(S)\big| \ge t$ ?
	\end{enumerate}
\end{tcolorbox}
Here, we sketch how to modify the \fptas from the previous section  to get an \fptas for \stpcrbdsshort. An algorithm for \stpcrbdsshort either returns a set $S \subseteq R$ of size at most $k$ such that $(i)~ T \subseteq S, (ii)~ \gcon[S]$ is connected and $(iii)~\big|N_{\gcov}(S)\big| \ge (1-\varepsilon)t$ or concludes that no set of size $k$ exists that satisfies \hyperref[two]{(2)}, \hyperref[three]{(3)}, and \hyperref[four]{(4)} (in which case it returns $\bot$).  

 
We note that by a slight alteration of the color coding-based reduction to {\sc Rooted Directed Steiner Out Tree} problem (cf.~\Cref{thm:fptbytpds} in \Cref{sec:fptpds}), we can design a randomized \fpt algorithm for \stpcrbdsshort that runs in time $\os(2^{(|T|+t} \cdot e^t)$; or a deterministic one in time $\os(2^{(|T|+t)} \cdot e^{t+o(t)})$. We sketch the idea, assuming familiarity with the construction of the digraph $D_f$ as described in \Cref{sec:fptpds}, given a coloring $f: B \to [t]$. We construct the digraph $D_f$ as described therein and make further modifications, as follows. For convenience, let us call the terminal set $Z$ in order to distinguish it from the terminals $T \subseteq R$ given in the input of \stpcrbdsshort. To this terminal set $Z$, we add additional vertices $\upsilon_1, \upsilon_2, \ldots, \upsilon_{|T|}$, and for each $1 \le i \le |T|$, we add a directed edge from a distinct $v_i \in T$ to $\upsilon_i$. Let us refer to the resulting instance as $(D_f, Z, p)$, where $p \coloneqq k + 2t + 2|T|$. Then, as in the algorithm corresponding to \Cref{thm:fptbytpds}, either we choose $e^t$ colorings $f$ independently, and run the $\os(2^{|Z|})$ algorithm for \rdsot from \cite{DBLP:journals/jco/MisraPRSS12}; or use $(n, t)$-perfect hash family of size $\os(e^{t+o(t)})$ for the deterministic variant. The correctness of this algorithm follows by closely mirroring the arguments from the proof of \Cref{thm:fptbytpds}, which we omit. We summarize this in the following proposition.

\begin{proposition} \label{prop:steinerbyt}
    \stpcrbdsshort can be solved deterministically in $\os((2)^{(|T|+t)} \cdot e^{t+o(t)})$; or using a randomized algorithm in time $\os((2)^{(|T|+t} \cdot e^t)$ that outputs the correct answer with probability at least $1-1/e$.
\end{proposition}
Using the above proposition, we can deterministically solve an instance of \stpcrbdsshort with $t < \frac{k^2d}{\varepsilon}$ in time $\os(2^{\Oh(k^2d/\varepsilon)})$.

Now we turn to the \fpt approximation in the case when $t \geq \frac{k^2d}{\varepsilon}$. 
Let $\cI = (\gcon, \gcov, T, k, t)$ be a \yes instance for \stpcrbdsshort. 
Let $S^\star \subseteq R$ be a set of size at most $k$ such that  
(1) $T \subseteq S^\star$,  
(2) $\gcon[S^\star]$ is connected, and  
(3) $\big|N_{\gcov}(S^\star)\big| \geq t \geq \tfrac{k^2d}{\varepsilon}$.  

\hyperref[stepone]{Step~1}, \hyperref[stepthree]{Step~3}, and \hyperref[stepfour]{Step~4} remain unchanged, since every graph we construct is derived from $\gcov$, which helps in approximating coverage. In these steps, the vertices of $T$ are treated in exactly the same way as any other vertices during the construction of the graphs. In \hyperref[steptwo]{Step~2}, we find a \emph{good separation} of the hypothetical solution $S^\star$ of \stpcrbdsshort, whose details remain unchanged.

The only step where the structure of the connectivity graph matters is \hyperref[stepfive]{Step~5}. 
Similar as before, a slight modification of the classic color-coding framework of Alon, Yuster, and Zwick~\cite[Thm.~6.3]{AlonYZ95} works. The modification proceeds as follows. First, observe that a good separation ensures that the terminal vertices are placed in distinct color classes, with each class containing at most one terminal from $T$. Moreover, if a color class contains a terminal vertex, we simply discard the other vertices in that class. We then apply the same dynamic programming routine as in the non-Steiner version. We omit the details and directly state the problem we solve and present the corresponding theorem.

\begin{description}
\item[\underline{\sc Steiner Maximum-Weight \(k\)-Tree}:]
Given a graph \(G\), a set of terminals $T \subseteq V(G)$, a weight function \(\gamma:V(G)\to\mathbb{Z}_{\ge 0}\), and an integer \(k\),
find a subtree \(T_G\subseteq G\) such that (i) $T \subseteq V(T_G)$, and (ii) \(|V(T_G)|\leq k\) maximizes
\(\gamma(T_G)\coloneqq \sum_{v\in V(T_G)} \gamma(v)\).
\end{description}

\begin{lemma}\label{lem:maxweightsteiner}
{\sc Steiner Maximum-Weight \(k\)-Tree} can be solved in time \(2^{\cO(k)}\,n^{\cO(1)}\).
\end{lemma}

 For an \yes instance, recall that $S^\star$ is a set of size at most $k$  satisfying \hyperref[two]{(2)}, \hyperref[three]{(3)}, and \hyperref[four]{(4)}. Following the arguments verbatim in \Cref{thm:kddfptapproxpds}, one can see that a set $S$ of size at most $k$ such that $(i)~ T \subseteq S, (ii)~ \gcon[S]$ is connected and $(iii)~\big|N_{\gcov}(S)\big| \ge (1-\varepsilon)t$ or correctly returns $\bot$, implying no set of size $k$ exists that satisfies \hyperref[two]{(2)}, \hyperref[three]{(3)}, and \hyperref[four]{(4)}. Thus we obtain the following theorem.


\begin{theorem} \label{thm:steinerkddfptapproxpds}
Let $\cI=(\gcov,\gcon,T, k,t)$ be an instance of \stpcrbdsshort, where $\gcov$ and $\gcon$ are the
coverage and connectivity graphs, respectively. If $\gcov$ is \kddfree, then for every
$\varepsilon\in(0,1)$ there is an algorithm running in time
$\os\!\bigl(2^{\cO(k^{2}d/\varepsilon)}\bigr)$ that either
\begin{enumerate}[label=(\roman*),itemsep=0pt,leftmargin=*]
  \item outputs a set $S\subseteq R$ with $|S|\le k$, $T \subseteq S$, $\gcon[S]$ connected, and
        $|N_{\gcov}(S)|\ge (1-\varepsilon)t$, or
  \item correctly concludes that no size-$k$ set dominates at least $t$ vertices in $\gcov$.
\end{enumerate}
\end{theorem}

When we state that we apply \Cref{thm:steinerkddfptapproxpds} with the input $\lr{\gcon, \gcov, T, k, t, \varepsilon}$, we refer to the instance $\cI = \lr{\gcon, \gcov, T, k, t}$ of the algorithm for \stpcrbdsshort, augmented with a fixed parameter $\varepsilon \in (0,1)$.

\section{$(1+\varepsilon,1)$-approximation for \pcrbdsshort} \label{sec:fptapproxfork}

In this section, we present a {\sf PAS} for \pcrbdsshort. In particular, we obtain a $(1+\varepsilon,1)$-bicriteria approximation for the problem with running time $2^{\Oh(kd(k^2 + \log d))} \cdot|\cI|^{\cO(1/\varepsilon)}$ when the coverage graph $\gcov$ is \kddfree. Unlike the previous section, which approximated the coverage value, here we approximate the solution size. More specifically, for \pcrbdsshort\ we seek a solution of size at most $(1+\varepsilon)k$ that covers at least $t$ edges.


Our new algorithm uses a modified form of \Cref{alg:conprbdspart}, described in \Cref{subsec:steinerext} as a subroutine. We first introduce some definitions.

\begin{definition}[ball of radius $r$]
    Let $r$ be a positive integer. For a graph $G$ and a vertex $v \in V(G)$, we define the \emph{ball of radius $r$ around $v$}, denoted as $\ball^r_{G}[v]$, to be the set of all vertices that are at most distance $r$ from $v$ in $G$. Formally, $\ball^r_{G}[v] \coloneqq \lrc{ u : \dist_{G}(u, v) \leq r }$, where $\dist_{G}(u, v)$ represents the shortest path distance between $u$ and $v$ in $G$, measured in terms of the number of edges. 
\end{definition}
By convention, a vertex is at distance zero from itself, i.e., $\ball^0_{G}[v] = \{ v \}$. We generalize this definition for a subset of vertices $X \subseteq V(G)$, as follows: 
$$\ball^r_{G}[X] \coloneqq \bigcup\limits_{v \in X} \ball^r_{G}[v]$$ 

Finally, we define \textit{$\cover$ of a graph} as follows.

\begin{definition}[$\cover$ of a graph]
    For a graph $G$, we define $\cover_r(G)$ to be the minimum-sized set of vertices $X \subseteq V(G)$ such that $\ball^r_{G}[X] = V(G)$.
\end{definition}
 
 In the next section, we show that if the graph is a tree, say $T$ on $n$ vertices, then the size of $\cover_r(T)$ is approximately bounded above by $\frac{n}{r}$, that is, $\big|\cover_r(T)\big| \lesssim \frac{n}{r}$.

\subsection{Tree Cover and Solution Covering Lemma}

In this section, we first prove that for a tree $T$ on $n$ vertices and a fixed integer $r$, the size of $\cover_r(T)$
is approximately bounded above by $\frac{n}{r}$, that is, $\big|\cover_r(T)\big| \lesssim \frac{n}{r}$.
This result, in turn, provides an upper bound on the cover of the solution subgraph 
for \pcrbdsshort. 

Before proceeding with the proof, we record the following simple observation about trees. 
Since it follows directly from basic properties of trees, we omit the proof. 

\begin{observation}\label{obs:tree_prop}  
    Let $T$ be a tree rooted at $r$, and let $T'$ be a subtree of $T$ rooted at $r'$. 
For any vertex $v \in V(T) \cap V(T')$, let $P_{rv}$ denote the unique path from $r$ to $v$ in $T$, 
and let $P_{r'v}$ denote the unique path from $r'$ to $v$ in $T'$. Then $P_{r'v}$ is a subpath of $P_{rv}$.
\end{observation}

\begin{lemma} \label{lem:treecover}
    Let $T$ be a tree on $n$ vertices. For any $q \in \mathbb{N}$,  $\big|\cover_q(T)\big| \leq \big\lceil \frac{n}{q+1} \big\rceil$.
\end{lemma}

\begin{proof}
    We prove the statement by induction on $n$. First, notice that we can assume $n > q$, since otherwise any arbitrary vertex forms a valid $\cover_q(T)$. 
Root $T$ at an arbitrary vertex $r \in V(T)$, and let $\ell$ be a leaf farthest from the root, i.e., for all $v \in V(T)$, we have 
$\dist(r, \ell) \geq \dist(r,v)$.\footnote{Throughout this proof, all distances are measured in the tree $T$, and hence we omit the subscript from $\dist_T(\cdot, \cdot)$.} 

If $\dist(r,\ell) \leq q$, then $\big|\cover_q(T)\big| = 1$, since $r$ lies within distance $q$ of every vertex in $T$, which implies that $\{r\}$ is a valid $\cover_q(T)$. 
Now consider the case when $\dist(r, \ell) > q$. In the base case when $n = 2$, the tree $T$ is a single edge with two endpoints. In this case, $\big|\cover_q\big(T)| = 1$ (either endpoint of the edge).

    Now assume the statement holds for every $2 \leq n' < n$, and we prove it for $n$.
Let $P_{r\ell}$ denote the (unique) path in $T$ from the root $r$ to the farthest leaf $\ell$.
Let $x \in V(T)$ be the vertex on $P_{r\ell}$ at distance exactly $q$ from $\ell$; such a vertex exists because $\dist(r,\ell) > q$.
Let $T_x$ be the subtree of $T$ rooted at $x$, with $V(T_x) = \{\, v \in V(T) : v \text{ is a descendant of } x \,\}$.
We claim that every vertex of $T_x$ lies within distance $q$ of $x$.

    \begin{claim}
          $V(T_x) \subseteq \ball^q_{T}[x]$.
    \end{claim}
    \begin{proof}
        Assume towards contradiction that there exists a vertex $y \in V(T_x)$ such that $\dist(x,y) > q$. Recall that $\dist(x,\ell)= q$. Let $P_{r\ell}$ and $P_{ry}$ denote the paths from $r$ to $\ell$ and $r$ to $y$ in $T$ respectively. Similarly let $P_{x\ell}$ and $P_{xy}$ denote the paths from $x$ to $\ell$ and $x$ to $y$ in $T_x$ respectively. By Observation \ref{obs:tree_prop}, $P_{x\ell}$ is a subpath of $P_{r\ell}$ and likewise $P_{xy}$ is a subpath of $P_{ry}$. Moreover, the path $P_{rx}$ from $r$ to $x$ is a common subpath of both $P_{r\ell}$ and $P_{ry}$. Since $\dist(x,y) = \big|P_{xy}\big| > q$ while $\dist(x,\ell) = \big|P_{x\ell}\big| = q$, it follows that $\big|P_{ry}\big| > \big|P_{r\ell}\big|$, contradicting the assumption that $\ell$ is the leaf farthest from the root $r$ in $T$. Hence every vertex $v \in V(T_x)$ is at distance at most $q$ from $x$ in $T$, i.e. $V(T_x) \subseteq \ball^q_{T}[x]$.
    \end{proof}
Let $T' = T - T_x$ be the graph obtained from $T$ by deleting the vertices of the subtree $T_x$.  
This directly leads to the following observations:  
(i) $T'$ is a subtree of $T$, and  
(ii) $|V(T_x)| \geq (q+1)$.  

Since $|V(T')| \leq n - (q+1) < n$, by the induction hypothesis we have  
$$
|\cover_q(T')| \leq \left\lceil \frac{n - (q+1)}{q+1} \right\rceil 
= \left\lceil \frac{n}{q+1} - 1 \right\rceil 
= \left\lceil \frac{n}{q+1} \right\rceil - 1.
$$  
Since $\cover_q(T') \cup \{x\}$ is a feasible solution for $\cover_q(T)$, it follows that  
$$
|\cover_q(T)| \leq |\cover_q(T')| + 1 
= \left( \left\lceil \frac{n}{q+1} \right\rceil - 1 \right) + 1 
= \left\lceil \frac{n}{q+1} \right\rceil.
$$
\end{proof}

 Let $\cI = (\gcon, \gcov, k, t)$ be an \yes instance for \pcrbdsshort. That is, there exists a solution $S^\star \subseteq R$ such that $\big|S^\star\big| \leq k$, $\gcon[S^\star]$
is connected, and $\big|N_{\gcov}(S^\star)\big| \geq t$.  
Since the subgraph $\gcon[S^\star]$ is connected, it contains a spanning tree $T^\star$.
The following lemma then follows directly from \Cref{lem:treecover}.

\begin{lemma}[Solution Covering Lemma] \label{lem:solutioncovering}
Let $\cI = (\gcon, \gcov, k, t)$ be a \yes instance for \pcrbdsshort and let $r$ be a non-negative integer. Furthermore, let $S^\star \subseteq V\lr{\gcon} = R$ be a solution of size at most $k$ satisfying $(1)~\gcon[S^\star]$ is connected and $(2)~\big| N_{\gcov}(S^\star)\big| \geq t$. Let $T^\star$ denote a spanning tree induced by the vertices of $S^\star$. Then $$\cover_{r}\lr{\gcon[S^\star]} \leq \cover_{r}\lr{T^\star} \leq \left\lceil\frac{k}{r+1}\right\rceil \leq \left\lceil\frac{k}{r}\right\rceil.$$
\end{lemma}

\subsection{Algorithm}

In this section, we first present our algorithm and subsequently analyze its approximation guarantee and running time. For our algorithm, we require the following result from Jain et al.~\cite{DBLP:conf/soda/0001KPSS0U23}.

\begin{lemma}{\rm \cite[Lemma~6.1]{DBLP:conf/soda/0001KPSS0U23}} \label{lem:addsol}
Let $G$ be a \kddfree bipartite graph with $V(G) = A \uplus B$, and let $\ell \leq k$ and $t'$ be two positive integers. Let $S'\subseteq A$ be an $\ell$ sized set such that $~\big|N(S')\big|\geq t'(1-\frac{1}{4\ell})$, where $t'\geq 4\ell^2d$,  and $H$ be a set of $\ell (d-1)(\ell^2)^{d-1}+1$ highest degree vertices from $A$. For any $S\subseteq A$ of size $\ell$ with $|N(S)|\geq t'$, either $S\cap H\neq\emptyset$ or there exists a vertex $ x \in H$ such that $|N(\{x\}\cup S')| \geq t'$.  
\end{lemma}


\begin{sloppypar}
Let \( \cI= \lr{\gcon, \gcov, k, t}\) be an \yes instance for \pcrbdsshort where $\gcov$ is 
\kddfree. Let \( S^\star \) be a set of \( k \) vertices such that $ (1)~ \gcon[S^\star]$ is connected, $(2)~\big|N_{\gcov}(S)\big| \geq t$. Now we provide a detailed description of our algorithm for \pcrbdsshort followed by the pseudocode of our algorithm.

\medskip
\noindent
\textbf{Description of \Cref{alg:addconprbdspark} and \Cref{alg:extaddconprbdspark}:} \Cref{alg:addconprbdspark} is an outer algorithm that, after making certain 
guesses, calls the subroutine \Cref{alg:extaddconprbdspark}, which serves as the 
main workhorse of our approach. Specifically, \Cref{alg:addconprbdspark} begins 
by enumerating a subset \(C\) of size at most \(\lceil 1/\varepsilon \rceil\) such that  $S^\star \subseteq \ball^{\varepsilon k}_{\gcon}[C] \subseteq V(\gcon)$. Essentially $C$ is a cover of $S^\star$. The existence of such a set $C$ is guaranteed by \Cref{lem:solutioncovering}, by setting the value of the radius, $r \coloneqq  \varepsilon k$. Consider $\widehat{R} \coloneqq \ball^{\varepsilon k}_{\gcon}[C]$.


 For every subset $C \subseteq R$, we invoke \Cref{alg:extaddconprbdspark}~($\extalgpas$) with input
\[
\lr{\,\hcon=\gcon\lrsq{\widehat{R}},\ \hcov=\gcov\lrsq{\widehat{R}\uplus N_{\gcov}\!\lr{\widehat{R}}},\ C,\ k,\ t\,}
\]
at \Cref{line:recursivecall}. The call returns either a set $S$ of size at most $(1+\varepsilon)k$ satisfying
\begin{enumerate}
\setlength{\itemsep}{-2pt}
  \item\label{prop:connNew} $\gcon\lrsq{S}$ is connected,
  \item\label{prop:subset} $C \subseteq S$, and
  \item\label{prop:covera} $\bigl|N_{\gcov}(S)\bigr| \ge t$,
\end{enumerate}
or correctly reports that no size-$k$ solution exists satisfying \Cref{prop:connNew,prop:subset,prop:covera}.

\medskip

\noindent
\Cref{alg:extaddconprbdspark} proceeds as follows. It first tests whether there exists a size-$k$ set $S_0 \supseteq C$ with $\gcon\lrsq{S_0}$ connected and
\[
\bigl|N_{\gcov}(S_0)\bigr| \ge (1-\delta)\,t,
\]
for a carefully chosen constant $\delta=\delta(k)$. By the guarantees of the subroutine, either such an $S_0$ is found or we correctly conclude that no size-$k$ solution exists satisfying \cref{prop:connNew,prop:subset,prop:covera}; in the latter case we return $\bot$.

Assume now that a size-$k$ solution ${\sf O}$ satisfying \Cref{prop:connNew,prop:subset,prop:covera} exists. We apply a combinatorial  Lemma (\Cref{lem:addsol}): let $H$ be the set of the top $f(k,\delta)$ vertices of $\widehat{R}$ by degree in $\hcov$. Then there is $w\in H$ such that either
\[
\bigl|N_{\gcov}(S_0 \cup \{w\})\bigr| \ge t
\quad\text{or}\quad
w \in {\sf O}.
\]
In the first outcome, since every vertex of $\hcon$ is within distance at most $\varepsilon k$ of some vertex in $C$, adding a path in $\hcon$ from $w$ to $S_0$ yields a connected solution of size at most $(1+\varepsilon)k$ that covers at least $t$ vertices. In the second outcome, we branch on candidates in $H$: for each $v \in H$, we recurse with parameters
\[
\lr{\,\hcon=\gcon\lrsq{\widehat{R}},\ \hcov=\gcov\lrsq{\widehat{R} \uplus N_{\gcov}\!\lr{\widehat{R}}},\ C\cup\{v\},\ k,\ t\,}.
\]

To test whether there exists a size-$k$ set $S_0 \supseteq C$ with $\gcon\lrsq{S_0}$ connected and
\[
\bigl|N_{\gcov}(S_0)\bigr| \ge (1-\delta)\,t,
\]
we invoke a routine that enforces inclusion of $C$ in the returned solution. This matches the specification of \stpcrbdsshort, where the solution must contain a designated terminal set. We instantiate its terminal set with $C$. A feasible output $S_0$ then certifies the test, while infeasibility certifies that no size-$k$ superset of $C$ satisfies the above connectivity and coverage requirements.

We now formalize the above intuition as the steps of \Cref{alg:extaddconprbdspark}~($\extalgpas$). After the sanity checks and base cases, at \Cref{line:steiner} we invoke the
\stpcrbdsshort algorithm guaranteed by \Cref{thm:steinerkddfptapproxpds} on the instance
\[
\lr{\hcon,\ \hcov,\ X,\ k,\ \delta=\tfrac{1}{4k}\,}.
\]
If the call returns a feasible solution $S_0$, we analyze it using \Cref{lem:addsol};
otherwise it returns $\bot$.

 \end{sloppypar}




\begin{algorithm}[ht!]
\caption{\algpas \label{alg:addconprbdspark}}
\KwInput{Connectivity graph $\gcon = (R, E)$, coverage graph $\gcov = (R \uplus B, E')$ which is \kddfree,  integers $k, t, \varepsilon \in (0,1)$}
\KwOutput{A set $S \subseteq R$ of size at most $(1+\varepsilon)k$ such that 
$(i)~\gcon[S]$ is connected, 
$(ii)~\big|N_{\gcov}(S)\big| \geq t$ or returns $\bot$ \blue{\tcp*{\footnotesize{$\bot$ indicates no such set of size $k$ can dominate $t$ vertices.}}}}
  Let $\cR$ be the family of $\big \lceil\frac{1}{\varepsilon} \big\rceil$ sized subsets of $R$, that is, $\cR = \lrc{C : C \subseteq R, \big|C \big| = \big\lceil \frac{1}{\varepsilon} \big\rceil}$\label{line:enumerate}\\
  \ForEach{$C \in \cR$}{
     Let $\widehat{R} \coloneqq \ball^{\varepsilon k}_{\gcon}[C]$\\ 
     \uIf{$\gcon[\widehat{R}]$ is connected}{
        Let $S_C \gets \extalgpas\lr{\gcon[\widehat{R}], \gcov\lrsq{\widehat{R} \uplus N_{\gcov}(\widehat{R})}, C,  k ,t}$\label{line:recursivecall}\\
     }
     \Else{
        Let $S_C \coloneqq \bot$
     }
    }
    \uIf{for all $C \in \cR$, $S_C = \bot$}{
        \Return $\bot$}
    \Else{
        \Return the smallest $S_C$ among all the valid solutions \blue{\tcp*{\footnotesize{valid solution}}} 
    }
\end{algorithm}

 \medskip

\begin{algorithm}[ht!]
\caption{\extalgpas \label{alg:extaddconprbdspark}}
\KwInput{Connectivity graph $\hcon=(\widetilde{R},\widetilde{E})$, coverage graph $\hcov= (\widetilde{R}\uplus \widetilde{B}, \widetilde{E}')$ which is \kddfree,  a set $X \subseteq \widetilde{R}$, integers $k, t$. }

\KwInvariant{$\hcon$ is connected.}
\KwOutput{\label{outputofalgext}A subset $S \subseteq \widetilde{R} = V(\hcon)$ of size at most $(k+p)$ such that
$(i)~\hcon\lrsq{S}$ is connected,
$(ii)~X \subseteq S$, and
$(iii)~\big|N_{\hcov}(S)\big| \geq t$; or returns $\bot$. Here, $p=\max_{v\in \widetilde{R}} \dist_{\hcon}(v, X)$
\blue{\tcp*{\footnotesize{$\bot$ indicates that no subset of size $k$ containing $X$ exists that dominates $t$ vertices}}}}
  \If{$\big| X \big| = k$}{
        \uIf{$\hcon\lrsq{X}$ is connected and $\big|N_{\hcov}\lr{X}\big| \geq t$ \label{lin:baseifconditionpas}}{
            \Return $X$ \label{lin:basecasereturnpas}\textnormal{\blue{\tcp*{\footnotesize{base case, valid solution}}}}
            }\Else
            {\Return $\bot$ \textnormal{\blue{\tcp*{\footnotesize{base case, invalid solution}}}}} 
    }
        Apply \Cref{thm:steinerkddfptapproxpds} with the input $\lr{\hcon, \hcov, X, k ,t, \tfrac{1}{4k}}$, and let $\widetilde{S}$ be the solution returned by it \label{line:steiner} \blue{\tcp*{\footnotesize{call approximate solution  for $t$}}}
        \uIf{$\widetilde{S} = \bot$}{\Return $\bot$ \label{line:invalidfromepas}} 
        \Else{
            Let $H$ be the set of $k(d-1)(k^2)^{d-1} + 1$ highest degree vertices from $\widetilde{R}$ in $\hcov$\\
            \If{there exists an $h \in H$ such that $\big| N_{\hcov}\lr{\widetilde{S} \cup \{h\}} \big| \geq t$\label{line:branchingonh}}{ 
                Let $c$ be a closest vertex in $X$ from $h$ in terms of $\dist_{\hcon}(\cdot, \cdot)$\\
                Consider the shortest path $P_{ch}$ between $c$ and $h$ in $\hcon$ \label{lin:closestvertexpas}\\
                \Return $\widetilde{S} \cup V(P_{ch})$ \blue{\tcp*{\footnotesize{additive approximate solution for $k$}}}
            }
            \ForEach{ $h \in H$\label{line:recursivecalls}}{
                $S_h \gets \extalgpas\lr{\hcon, \hcov, X \cup \{h\}, k, t}$\label{line:residueinstance}\\
            }
            \uIf{for all $h \in H, S_h = \bot$}{
                \Return $\bot$
            }
            \Else{
                \Return the smallest among all the valid solutions \blue{\tcp*{\footnotesize{valid solution}}} 
            }
        }
\end{algorithm}

\subsection{Approximation Factor Analysis and Running time}

\begin{sloppypar}
We proceed as follows. In \Cref{lem:ext:alg}, we prove the correctness of \Cref{alg:extaddconprbdspark} $\lr{\extalgpas}$. Then in \Cref{thm:pasforchs}, we prove the correctness of \Cref{alg:addconprbdspark} $\lr{\algpas}$. And finally we prove the main theorem. Now we have the following lemma. 


\end{sloppypar}

\begin{lemma} \label{lem:ext:alg}
\begin{sloppypar}
Let \(\cI=(\hcon,\hcov,X,k,t)\) be an input to \(\extalgpas\) with \(\hcon\) connected. Then, in time
\(\runtimea\cdot |\cI|^{\cO(1)}\), \Cref{alg:extaddconprbdspark} either outputs a set
\(S\subseteq V(\hcon)\) with \(|S|\le k+p\) such that
\begin{enumerate}[label=(\roman*),itemsep=0pt,leftmargin=*]
  \item \(\hcon[S]\) is connected,
  \item \(X\subseteq S\), and
  \item \(\bigl|N_{\hcov}(S)\bigr|\ge t\),
\end{enumerate}
or correctly returns \(\bot\), certifying that no solution of size at most \( k\) satisfies \((i)\)–\((iii)\).
Here \(p \coloneqq \max_{v\in \widetilde{R}} \dist_{\hcon}(v,X)\).
\end{sloppypar}
\end{lemma}

\begin{proof}
\begin{sloppypar}
    By a slight abuse of notation, we refer to a particular input $(\hcon, \hcov, X, k, t)$ as an \emph{instance} of the algorithm \extalgpas.     Suppose an instance $\cI = \lr{\hcon, \hcov, X , k , t}$ admits a solution $S^\star$ of size at most $k$ satisfying properties $(i) - (iii)$, then we say that it is a \yes-instance of \extalgpas. Then we show that \Cref{alg:extaddconprbdspark} outputs a solution of size at most $k+p$ satisfying properties $(i) - (iii)$. 
    We prove the correctness of \Cref{alg:extaddconprbdspark} by induction on the measure $\mu(\cI) = k - |X|$.
\end{sloppypar}


    \medskip
    \noindent
    \textbf{Base case:} When $\mu(\cI) = 0$, it follows that $|X| = k$. Since we assume that the input instance $\cI$ is a \yes-instance, we need to show that the algorithm does not return $\bot$. In this case, note that any solution to $\cI$ that is of size $k$ and that contains $X$, must be precisely $X$ itself. Further, since $X$ is a solution, it holds that (i) $\hcon[X]$ is connected, and (ii) $|N_{\hcov}(X)| \ge t$. These are precisely the conditions checked by the algorithm in the base case in \Cref{lin:baseifconditionpas}. Therefore, the algorithm will return $X$ in \Cref{lin:basecasereturnpas}. Note that $X$ is of size $k \le k+p$, since $p \ge 0$. This finishes the base case. 
    
    
  \medskip

    \begin{sloppypar}
        \noindent
    \textbf{Induction Hypothesis:} Suppose for any instance $\cI' = \lr{\hcon, \hcov, X' , k , t}$ of \extalgpas with measure $\mu(\cI') = \mu^\star - 1 \geq 1$, the statement is true. That is, if $\cI'$ is a \yes-instance to \extalgpas, i.e., there is some $S \subseteq \widetilde{R}$ of size at most $k$ such that (I) $X' \subseteq S$, (II) $\hcon[S]$ is connected, and (III) $|N_{\hcov}(S)| \ge t$, then the algorithm outputs a solution $S'$ of size at most $k+p'$ satisfying (I)-(III), where $p' = \max_{v \in \widetilde{R}} \dist_{\hcon}(v, X')$.
    \end{sloppypar}

    \medskip
  \noindent 
    \textbf{Inductive Step:}
    Now we prove the statement for an arbitrary instance $\cI = (\hcon, \hcov, X, k, t)$ with $\mu(\cI) = \mu^\star$. Suppose $\cI$ is a \yes-instance of \extalgpas. To reiterate, this means that there exists a solution $S^\star \subseteq \widetilde{R}$ of size at most $k$ such that (i) $\hcon[S^\star]$ is connected, (ii) $X \subseteq S^\star$, and (iii) $|N_{\hcov}(S^\star)| \ge t$.


    It follows that, when we apply \Cref{thm:steinerkddfptapproxpds} in \Cref{line:steiner} with the input $(\hcon, \hcov, X, k, t, \frac{1}{4k})$, it must return a solution $\widetilde{S}$ (that is not equal to $\bot$) of size at most $k$ such that $X \subseteq \widetilde{S}$, $\hcon[\widetilde{S}]$ is connected, and $|N_{\hcov}(\widetilde{S})| \ge (1-\frac{1}{4k})t$. 

 Observe that \Cref{lem:addsol} applies to any arbitrary pair of sets $(S', S^\star)$, provided they satisfy the preconditions of \Cref{lem:addsol}.
In particular, it also holds for a specifically chosen pair $(\widetilde{S}, S^\star)$ where $\widetilde{S}$ (a solution from \Cref{line:steiner}) and $S^\star$ (hypothetical solution) induce connected subgraphs in $\gcon$. We apply \Cref{lem:addsol} with the following parameters: (1) $S' \gets$ solution $\widetilde{S}$ returned by \Cref{thm:steinerkddfptapproxpds}, (2)  $H \gets$ the set of $k(d-1)(k^2)^{d-1} + 1$ highest degree vertices in $\widetilde{R}$, (3) $\ell \gets k$, (4) $t' \gets t$, and (5) $S \gets S^\star$ (hypothetical solution as assumed above). Then, \Cref{lem:addsol} implies that there are following two possibilities.  
    \begin{enumerate}
        \item Either      
        for some $h \in H$, $\big| N_{\hcov}(\widetilde{S} \cup \{h\})\big| \ge t$. 

        \item  Or 
        $S^\star \cap H \neq \emptyset$. 
    \end{enumerate}
    Now we consider each case separately.

    In the first case, consider the iteration corresponds to $h \in H$ that satisfies the condition.
    This implies that $\widetilde{S} \cup \{h\}$ is a potential solution, provided that conditions $(i) - (iii)$ are satisfied. Since $X \subseteq \widetilde{S} \subseteq \widetilde{S} \cup \LR{h}$, condition $(ii)$ is satisfied. Condition $(iii)$ is satisfied due to the choice of $h$. However, condition $(i)$ may not be satisfied, i.e., $\hcon[\widetilde{S} \cup \LR{h}]$ may not be connected. To this end, consider the vertex $c \in X$ from $h$ that is chosen in \Cref{lin:closestvertexpas}. It follows that $\dist_{\hcon}(c,h) \le p$, which implies $|V(P_{ch})| \le p$, where $P_{ch}$ is the shortest path between $c$ and $h$ in $\hcon$. This implies that $\widetilde{S} \cup V(P_{ch})$ has size at most $k+p$, and it is easily verified that it satisfies properties $(i)-(iii)$. This finishes the first case. 

        

    Now let us consider the second case, i.e., when we have that $S^\star \cap H \neq \emptyset$. Let $h$ be an arbitrary vertex in $S^\star \cap H$. Then, in \Cref{line:residueinstance}, consider the recursive call made to the instance $\cI_h = \lr{\hcon, \hcov, C, X \cup \{h\}, k, t}$, which is $\extalgpas\lr{\cI_h}$. Notice first that the graph $\hcon$ is unchanged and the invariant is maintained. Now we show that, $\cI_h$ is a \yes-instance of \extalgpas. To this end, recall that $S^\star$ is a hypothetical solution of size at most $k$ for $\cI$ satisfying ~$(i)$--$(iii)$ w.r.t $\cI$. Furthermore, since $X \cup \{h\} \subseteq S^\star$, it follows that $S^\star$ is also a valid solution of size at most $k$ for $\cI_h$ satisfying $(i) - (iii)$ w.r.t.~$\cI_h$; thereby making $\cI_h$ a \yes-instance of \extalgpas. Further, note that $\mu(\cI_h) = k-|X \cup \LR{h}| = k - |X| - 1 = \mu(\cI) - 1 = \mu^\star - 1$. 
   
    Hence, by induction hypothesis, the recursive call corresponding to $h$ made in \Cref{line:residueinstance} returns a solution $S_h$ of size  at most $k+p_{\scriptscriptstyle h}$, such that $S_h$ satisfies conditions $(i) - (iii)$ w.r.t.~$\cI_h$. Here, $p_{\scriptscriptstyle h} = \max_{v \in \widetilde{R}} \dist_{\hcon}(v, X \cup \LR{h})$. Note that $p_{\scriptscriptstyle h} \le p$ since $X \subseteq X \cup \LR{h}$. Due to condition $(ii)$, since $X \cup \{h\}\subseteq S_h$, this implies that $X \subseteq S_h$. Hence $S_h$ is also a solution for $\cI$ of size at most $k+ p_{\scriptscriptstyle h} \leq k+p$ that satisfies conditions $(i) -(iii)$ w.r.t. $\cI$. This concludes the proof of the approximation factor of \Cref{lem:ext:alg}. Next we prove the running time.

\smallskip
\begin{sloppypar}
\noindent
\textbf{Running time:} 
We analyze the running time by induction on the measure $\mu(\cI)$. 
Let $T(\mu(\cI))$ denote the total time required for one execution of 
$\extalgpas(\cI)$ when the measure is $\mu(\cI) = k - |X|$. 
The contributions those dominate the running time come from two steps: 
\Cref{line:steiner}, where the algorithm in \Cref{thm:steinerkddfptapproxpds} is invoked, 
and \Cref{line:recursivecalls}, where $|H|$ recursive calls are made. 

In \Cref{line:steiner}, since the algorithm in \Cref{thm:steinerkddfptapproxpds} 
is applied to the instance 
$\cI = \lr{\hcon, \hcov, X, k, t, \delta \gets \tfrac{1}{4k}}$ to compute a solution 
$\widetilde{S}$, its running time is 
\[
2^{\Oh(k^2 d / \delta)} \cdot |\cI|^{\cO(1)} 
= 2^{\Oh(k^3 d)} \cdot |\cI|^{\cO(1)}.
\]

Furthermore, the set $H$ has size 
\[
|H| = k(d-1)(k^2)^{d-1} + 1 = d \cdot k^{\cO(d)},
\] 
and for each $h \in H$, the algorithm makes a recursive call to an instance whose measure is equal to $\mu(\cI) - 1$.


Therefore, the recurrence relation for the running time of 
$\extalgpas(\cI)$ is given by
\begin{eqnarray*}
T\lr{\mu(\cI)} &\leq &  |H| \cdot T(\mu(\cI) - 1)+ 2^{\Oh(k^3d)}\cdot|\cI|^{\cO(1)} + |\cI|^{\cO(1)} \\
&\leq &  |H| \cdot T(\mu(\cI) - 1)+ 2^{\Oh(k^3d)}\cdot|\cI|^{\cO(1)}\\
& \leq & |H|^{\mu(\cI)}\cdot 2^{\Oh(k^3d)}\cdot |\cI|^{\cO(1)}
\end{eqnarray*}

Due to space constraints, we omit the details of solving the recurrence and directly state its compact form. 
Since the upper bound on the measure $\mu(\cI)$ is $k$, solving the recurrence yields
\[
T(k) \;\leq\; 2^{\Oh(k^3 d + dk \log d)} \cdot |\cI|^{\cO(1)} = 2^{\Oh(kd (k^2 +\log d))} \cdot |\cI|^{\cO(1)}.
\]
This completes the running time analysis.
\end{sloppypar}
\end{proof}

\begin{lemma} \label{lem:alg}
\begin{sloppypar}
Let $\cI = (\gcon, \gcov, k, t)$ be an instnace of \pcrbdsshort. Then, when \extalgpas is called with $(\gcon, \gcov, k, t, \varepsilon)$ for some $\varepsilon \in (0, 1)$, the following holds.
The algorithm \extalgpas runs in time $2^{\Oh(kd(k^2 + \log d))} \cdot |\cI|^{\Oh(1/\varepsilon)}$, and if $\cI$ is a \yes-instance of \pcrbdsshort, then it outputs a set $S \subseteq V \lr{\gcon} = R$ of size at most $(1+ \varepsilon)k$ such that $(i)~\gcon\lrsq{S}$ is connected, and $(ii)~\big|N_{\gcov}(S)\big| \geq t$; or correctly returns $\bot$, $\cI$ is a \no-instance of \pcrbdsshort.
\end{sloppypar}
\end{lemma}

\begin{proof}

    Consider $\cI$ and $\varepsilon \in (0, 1)$ as in the statement of the lemma. Further suppose that $\cI$ is a \yes-instance. Then, it admits a solution $S^\star$ of size at most $k$ satisfying properties~$(i)$ and~$(ii)$ mentioned in the statement. Then we show that \Cref{alg:addconprbdspark} produces a solution $S$ of size at most $(1+\varepsilon)k$ that also satisfies $(i)$ and~$(ii)$.

     In \Cref{line:enumerate}, algorithm guesses (enumerates) a subset $C$ of size at most $\lceil 1/\varepsilon \rceil$ such that $S^\star \subseteq \ball^{\varepsilon k}_{\gcon}[C] \subseteq V(\gcon)$. The existence of such a set $C$ is guaranteed by \Cref{lem:solutioncovering}, by setting the value of the radius, $r \coloneqq  \varepsilon k$. Consider $\widehat{R} \coloneqq \ball^{\varepsilon k}_{\gcon}$. Note that since $\gcov[S^\star]$ is connected, and $C \subseteq S^\star$, it follows that $\gcon[\widehat{R}]$ is connected. Therefore, we make a call $\extalgpas\lr{\gcon[\widehat{R}], \gcov\lrsq{\widehat{R} \uplus N_{\gcov}(\widehat{R})}, C, k ,t}$ in \Cref{line:recursivecall}. Consider this call.

    By Lemma~\ref{lem:ext:alg}, the algorithm $\extalgpas$ must return a $S_C$ of size at most $k + p = k + \varepsilon k = (1+\varepsilon)k$ satisfying properties: $(1)~\gcon[S_C]$ is connected, $(2)~C \subseteq S_C$, and $(3)~\big| N_{\gcon}(S_C)\big|\geq t$ or it returns $\bot$. Finally \Cref{alg:addconprbdspark} outputs the set that has the minimum size among all the $|R|^{1/\varepsilon}$ guesses.
    
    The running time is bounded by the execution time of $\extalgpas$, multiplied by the number of guesses for $C$, that is $\big|V\lr{\gcon}\big|^{1/\varepsilon} = |R|^{1/\varepsilon} \le |\cI|^{1/\varepsilon}$. Thus, the bound on the running time follows.
\end{proof}

This directly leads to our main theorem of this section.

\begin{restatable}{theorem}{kddfptapproxfork}\label{thm:pasforchs}
Let \(\cI=(\gcov,\gcon,k,t)\) be an instance of \pcrbdsshort, where \(\gcov\) and \(\gcon\) are the
coverage and connectivity graphs, respectively. For any \(0<\varepsilon<1\), if \(\gcov\) is \kddfree,
there is an algorithm running in time
\[
  2^{\Oh(kd(k^2 + \log d))} \cdot|\cI|^{\cO(1/\varepsilon)}
\]
that either
\begin{enumerate}[label=(\roman*),itemsep=0pt,leftmargin=*]
  \item outputs a set \(S\subseteq R\) with \(|S|\le (1+\varepsilon)k\), \(\gcon[S]\) connected, and
        \(|N_{\gcov}(S)|\ge t\), or
  \item correctly concludes that no connected size-\(k\) set dominates \(t\) blue vertices in \(\gcov\).
\end{enumerate}
\end{restatable}

\section{Lower Bounds}\label{sec:lower}
In this section we establish several parameterized and approximation lower bounds for \pcrbdsshort. Many of these results are established by modeling already existing problems such as {\sc Max Coverage, Connected Partial/Budgeted Dominating Set} as special cases of \pcrbdsshort. 

\subsection{\pcrbdsshort parameterized by $k$}
In this subsection, we show a {\sf W[1]}-hardness and two \fpt-time hardness-of-approximation results for \pcrbdsshort. We start by our {\sf W[1]}-hardness result. 

\begin{restatable}{theorem}{hardkpds}  \label{thm:hardkpds}
	 \pcrbdsshort is \woh when parameterized by $k$, even if the coverage graph is $3$-degenerate and the connectivity graph is either a clique or a star. 
\end{restatable}

\begin{proof}

 We give a parameterized reduction from {\sc Partial Dominating Set (PDS)} which is known to be  \woh for 2-degenerate graphs when parameterized by $k$ (\cite{DBLP:conf/wg/GolovachV08}). In PDS we are given a graph $G$ with two integers $k$ and $t$ as input, where our goal is to check whether there is a vertex set of size at most $k$ that dominates  at least $t$ vertices in $G$.     

Given an instance \((G,k,t)\) of \textsc{PDS}, we build an instance
\(\cI=(\gcon,\gcov,k',t)\) of \pcrbdsshort as follows.

\medskip\noindent\textbf{Coverage graph \(\gcov\).}
Create two copies of \(V(G)\): let
\(R=\{v^{R}: v\in V(G)\}\) and \(B=\{v^{B}: v\in V(G)\}\).
For each \(u\in V(G)\), make \(u^{R}\) adjacent in \(\gcov\) to all \(v^{B}\) with \(v\in N_G[u]\)
(the closed neighborhood in \(G\)). Thus \(\gcov\) is a bipartite graph with sides \(R\) and \(B\).

\medskip\noindent\textbf{\(3\)-degeneracy of \(\gcov\).}
Let \(v_1,\dots,v_n\) be a \(2\)-degeneracy ordering of \(G\); that is, each \(v_i\) has at most
two neighbors among \(\{v_j : j>i\}\).
Consider the interleaved ordering
\[
v_1^{R},\, v_1^{B},\, v_2^{R},\, v_2^{B},\,\dots,\, v_n^{R},\, v_n^{B}.
\]
Then:
\begin{itemize}[itemsep=0pt,leftmargin=*]
  \item The later neighbors of \(v_i^{R}\) are exactly \(v_i^{B}\) and \(v_j^{B}\) for those
        \(j>i\) with \(v_j\in N_G(v_i)\); hence \(v_i^{R}\) has at most \(1+2=3\) later neighbors.
  \item The later neighbors of \(v_i^{B}\) are \(v_j^{R}\) for those \(j>i\) with
        \(v_j\in N_G(v_i)\); hence \(v_i^{B}\) has at most \(2\) later neighbors.
\end{itemize}
Therefore every vertex of \(\gcov\) has at most three later neighbors in this ordering, and so
\(\gcov\) is \(3\)-degenerate.

\paragraph*{Connectivity graph \(\gcon\).}
We will define \(\gcon\) differently in the two cases considered below.


 \begin{description}
     \item[(i) $\gcon$ is a clique:] 	Let $\gcon$ be the complete graph on $R$, and set $k' = k$. Now we show the correctness of the reduction. In the forward direction, if $S$ is a solution of PDS to $(G,k,t)$ then   all the copies of vertices of $S$ in $R$ must dominate all the copies of vertices of $N_G[S]$ in $B$. Additionally, $\gcon[S]$ is connected since it induces a clique. Hence $\cI$ is \yes-instance of \pcrbdsshort. In the backward direction, let $S' \subseteq R$ be a solution to $\cI$. Then the corresponding vertex set of $S'$ in $G'$ also dominates same number of vertices in $G$. Hence  $(G,k,t)$ is an \yes-instance of PDS.

     \item[(ii) $\gcon$ is a star:] Add a new vertex $z$ to $R$, and let $\gcon$ be the star centered at $z$, with all other vertices in $R$ as leaves.  Set $k' = k+1$.  Now we show the correctness of the reduction. In the forward direction, if $S$ is a solution of PDS to $(G,k,t)$ then   all the copies of vertices of $S$ in $R$ must dominate all the copies of vertices of $N[S]$ in $B$. And $\gcon[S \cup \{z\}]$ is connected. Hence $\cI$ is \yes-instance of \pcrbdsshort. In the backward direction, let $S' \subseteq R$ be a solution to $\cI$. Clearly $z \in S'$ as $\gcon[S']$ is connected. Now the vertex $z \in R$ has no neighbour in $B$. So  the corresponding vertex set of $S' \setminus \{z\}$ in $G$ also dominates same number of vertices in $G$. Hence  $(G,k,t)$ is an \yes-instance of PDS.
 \end{description}
 	We have shown that for both cases ($\gcon$ being a clique or a star), $(G,k,t)$ is a \yes-instance of PDS if and only if $\cI$ is a \yes-instance of \pcrbdsshort.  Since PDS is \textsf{W[1]}-hard on $2$-degenerate graphs parameterized by $k$, and our reduction creates  instances where $\gcov$ is $3$-degenerate, the proof follows. 
    \end{proof}

In the next sections, we present our hardness-of-approximation results.

\subsection{$(1,1-e^{-1}+\varepsilon)$-approximation of \pcrbdsshort.} 

Here we show that assuming the  Exponential Time Hypothesis ({\sf ETH}) \cite{DBLP:journals/jcss/ChenHKX06}, for any $\varepsilon>0$, any $(1, 1-e^{-1}+\varepsilon)$-approximation algorithm for \pcrbdsshort   must take   $\Omega_k(n^{k^{\Omega(1)}})$ running time.  
The reduction is from the {\sc Max $(\beta,k)$-Coverage} problem  \cite{{DBLP:conf/stoc/GuruswamiLRS025}}. For $0 < \beta \le t$, an instance of {\sc Max $(\beta,k)$-Coverage} consists of a family $\mathcal{F}$ of $m$ sets over a universe
$\mathcal{U}$ of size $n$, and the goal is to distinguish
between the following two cases.

\begin{itemize}
    \item Type I:  There exists $k$ subsets from $\mathcal{F}$ whose union equals $\mathcal{U}$.

    \item Type II: Any $k$ subsets from $\mathcal{F}$ has the union size at most $\beta|\mathcal{U}|$.
\end{itemize}

\noindent The following result is known for  {\sc Max $(\beta,k)$-Coverage} under {\sf ETH}.

\begin{proposition}{\rm \cite[Theorem 1.5]{DBLP:conf/stoc/GuruswamiLRS025}}\label{prop:cover}
Assumming {\sf ETH}, for any constant $\varepsilon \in (0,1)$, any algorithm deciding {\sc Max $(1-e^{-1}+\varepsilon,k)$-Coverage} must take runtime $\Omega_k(|\mathcal{F}|^{k^{\Omega(1)}})$.
\end{proposition}

We now show the analogous result  for \pcrbdsshort.

\begin{restatable}{theorem}{hardtapproxpds}  \label{thm:hardtapproxpds}
Assuming {\sf ETH}, there is no \fpt\ algorithm parameterized by \(k\) that achieves a
\((1,\,1-e^{-1}+\varepsilon)\)-approximation for \pcrbdsshort.
Moreover, under {\sf ETH}, for any fixed \(\varepsilon\in(0,1)\), every
\((1,\,1-e^{-1}+\varepsilon)\)-approximation algorithm for \pcrbdsshort must run in time
\(\Omega_k\!\bigl(|\mathcal{F}|^{\,k^{\Omega(1)}}\bigr)\), even when the connectivity graph \(\gcon\)
is a clique or a star.
\end{restatable}

\begin{proof}
	 We give a  reduction from {\sc Max $(\beta,k)$-Coverage} problem with $\beta = 1-e^{-1}+\varepsilon$. Given an instance $(\mathcal{U}, \mathcal{F})$ of the {\sc Max $(\beta,k)$-Coverage}  problem with $\beta= 1-e^{-1}+\varepsilon$, we construct an instance $\cI = (\gcon, \gcov,  k', n)$ of \pcrbdsshort as follows:  
	$$
	V(\gcov) \leftarrow R = \{ v_F : F \in \mathcal{F} \} \uplus B = \{ u : u \in \mathcal{U} \},
	$$
	$$
	E(\gcov) \leftarrow \{ v_Fu  : F \in \mathcal{F}, u \in \mathcal{U} \}
	$$
Observe that \(\gcov\) is the incidence bipartite graph of the set system \((\mathcal{U},\mathcal{F})\).
We now analyze the two cases for the connectivity graph separately.

\paragraph*{(i) $\gcon$ is a clique:} 	Let $\gcon$ be the complete graph on $R$, and set $k' = k$.  
	Then any subset $R' \subseteq R$ is connected in $\gcon$.    It is easy to observe that  there exists $k$ subsets from $\mathcal{F}$ whose union size is  $\ell$ if and only if there is a vertex set $R'\subseteq R \subseteq V(\gcov)$ such that $|R'|=k, |N_{\gcov}(R)|=\ell$ and $\gcon[R']$ is connected. 
    
   We now show that an $(1, 1-e^{-1}+\varepsilon)$-approximation algorithm for \pcrbdsshort to the instance $\cI = (\gcon, \gcov,  k, n)$   can solve    {\sc Max $(\beta,k)$-Coverage} for the instance $(\mathcal{U}, \mathcal{F})$. This would imply \Cref{thm:hardtapproxpds} (for the case when corresponding connectivity graph  is a  clique) by \Cref{prop:cover}.
   
   Our approach is as follows: Let $\mathcal{A}$ be  $(1, 1-e^{-1}+\varepsilon)$-approximation algorithm for \pcrbdsshort to the instance $\cI$. 
    
    \begin{itemize}
        \item If $\mathcal{A}$ returns no, i.e., there is no $k$ size vertex set in $ R$  that dominates  $B$ such that  $\gcon[S]$ is connected  then we return that $(\mathcal{U}, \mathcal{F})$ is of Type II. 
        \item If $\mathcal{A}$ returns  some vertex subset of $R$ then we return that $(\mathcal{U}, \mathcal{F})$ is of Type I.  
    \end{itemize}

    \paragraph*{Correctness (soundness).}
Suppose the algorithm \(\mathcal{A}\) outputs \emph{no}. We claim that the corresponding
\((\mathcal{U},\mathcal{F})\) instance of \textsc{Max\((\beta,k)\)-Coverage} is of Type~II.
Assume, toward a contradiction, that it is Type~I. Then there exist \(k\) sets
\(\mathcal{F}'\subseteq \mathcal{F}\) whose union equals \(\mathcal{U}\).
Let \(R' \coloneqq \{v_{F} : F\in \mathcal{F}'\}\subseteq R\). By construction of \(\gcov\),
we have \(|N_{\gcov}(R')| = |\mathcal{U}|=n\). Moreover, by our construction of \(\gcon\)
(clique), the induced subgraph \(\gcon[R']\) is connected. Hence \(\mathcal{A}\)
should have answered \emph{yes}, a contradiction. Therefore, if \(\mathcal{A}\) returns
\emph{no}, the instance \((\mathcal{U},\mathcal{F})\) must be Type~II.
\paragraph*{Correctness (completeness).}
Assume \(\mathcal{A}\) outputs a vertex set \(R^*\). By the guarantee of \(\mathcal{A}\),
\(|R^*|\le k\) and \(R^*\) dominates at least \((1-e^{-1}+\varepsilon)\,n\) vertices of \(B\).
We claim that the corresponding \((\mathcal{U},\mathcal{F})\) instance of
\textsc{Max\((\beta,k)\)-Coverage} is Type~I.

Suppose, for contradiction, that \((\mathcal{U},\mathcal{F})\) is Type~II; that is, for
\(\beta=1-e^{-1}+\varepsilon\), every \(k\)-subfamily \(\mathcal{F}'\subseteq\mathcal{F}\)
covers fewer than \(\beta|\mathcal{U}|\) elements. By the reduction, domination in \(\gcov\)
matches coverage in \((\mathcal{U},\mathcal{F})\): for any \(k\)-subset \(R'\subseteq R\),
if \(\mathcal{F}'=\{F:\ v_F\in R'\}\), then \(R'\) dominates exactly the \(B\)-vertices
corresponding to the elements covered by \(\mathcal{F}'\). Hence no \(k\)-subset of \(R\)
can dominate \((1-e^{-1}+\varepsilon)\,n\) vertices of \(B\), contradicting the existence
of \(R^*\). Therefore the instance must be Type~I.


Therefore, a \((1,\,1-e^{-1}+\varepsilon)\)-approximation for \pcrbdsshort would let us
distinguish Type~I from Type~II instances of \textsc{Max} \((\beta,k)\)\textsc{-Coverage}.
By \Cref{prop:cover}, this distinction requires time
\(\Omega_k\!\bigl(|\mathcal{F}|^{\,k^{\Omega(1)}}\bigr)\) under \textsf{ETH}; hence any
\((1,\,1-e^{-1}+\varepsilon)\)-approximation for \pcrbdsshort must also run in
\(\Omega_k\!\bigl(|\mathcal{F}|^{\,k^{\Omega(1)}}\bigr)\) time. \qed

\paragraph*{(ii) $\gcon$ is a star:} 	
Introduce a new vertex \(z\) into \(R\), and let \(\gcon\) be the star centered at \(z\) with
leaf set \(R\setminus\{z\}\). Set \(k' \coloneqq k+1\) and \(t \coloneqq n=|\mathcal{U}|\).
Then any vertex subset \(S\subseteq R\) that contains \(z\) induces a connected subgraph of \(\gcon\)
(while any subset of size \(\ge 2\) avoiding \(z\) is disconnected). The rest of the reduction
mirrors the clique case; we include it for completeness.

It is immediate that there exist \(k\) sets \(\mathcal{F}'\subseteq \mathcal{F}\) whose union has
size \(\ell\) if and only if there is a set
\(R'\subseteq R\) with \(|R'|=k+1\), \(z\in R'\), \(|N_{\gcov}(R')|=\ell\), and \(\gcon[R']\)
connected. Indeed, take \(R'=\{z\}\cup\{v_F : F\in\mathcal{F}'\}\).

\medskip
\noindent\textbf{Using \(\mathcal{A}\) to decide Type~I vs.~Type~II.}
Let \(\mathcal{A}\) be a \((1,\,1-e^{-1}+\varepsilon)\)-approximation algorithm for \pcrbdsshort
on the instance \(\cI=(\gcon,\gcov,k',n)\).

\begin{itemize}[leftmargin=*]
  \item If \(\mathcal{A}\) returns \emph{no}, i.e., there is no \(S\subseteq R\) with \(|S|=k'\)
        (here \(k'=k+1\)) such that \(\gcon[S]\) is connected and \(S\) dominates \(n\) vertices
        of \(B\), we output that \((\mathcal{U},\mathcal{F})\) is \emph{Type~II}.
  \item If \(\mathcal{A}\) returns some \(R^{*}\subseteq R\), we output that
        \((\mathcal{U},\mathcal{F})\) is \emph{Type~I}.
\end{itemize}

\paragraph*{Correctness.}
If \(\mathcal{A}\) returns \emph{no} but \((\mathcal{U},\mathcal{F})\) were Type~I, then there would
exist \(\mathcal{F}'\subseteq\mathcal{F}\) with \(|\mathcal{F}'|=k\) covering all \(n\) elements,
and hence \(R''\coloneqq \{z\}\cup\{v_F:F\in\mathcal{F}'\}\) would satisfy \(|R''|=k+1\),
\(|N_{\gcov}(R'')|\ge n\), and \(\gcon[R'']\) connected—a contradiction. Thus, if \(\mathcal{A}\)
returns \emph{no}, the instance is Type~II.

Conversely, if \(\mathcal{A}\) outputs \(R^{*}\), then by the approximation guarantee
\(|R^{*}|\le k+1\), \(R^{*}\) dominates at least \((1-e^{-1}+\varepsilon)n\) vertices of \(B\), and
\(\gcon[R^{*}]\) is connected; in a star this forces \(z\in R^{*}\).
If \((\mathcal{U},\mathcal{F})\) were Type~II (with \(\beta=1-e^{-1}+\varepsilon\)), then no
\(k\)-subfamily would cover \(\beta n\) elements, and therefore no \((k+1)\)-subset of \(R\) with
\(z\) included would dominate \(\beta n\) vertices—contradicting the existence of \(R^{*}\).
Hence the instance must be Type~I.

\medskip
Therefore, distinguishing Type~I from Type~II instances of
\textsc{Max}\((\beta,k)\)\textsc{-Coverage} reduces to obtaining a
\((1,\,1-e^{-1}+\varepsilon)\)-approximation for \pcrbdsshort. By \Cref{prop:cover},
\textsc{Max}\((\beta,k)\)\textsc{-Coverage} requires
\(\Omega_k\!\bigl(|\mathcal{F}|^{\,k^{\Omega(1)}}\bigr)\) time under \textsf{ETH}, and the claimed
hardness in \Cref{thm:hardtapproxpds} follows for the star connectivity case as well.

Introduce a new vertex \(z\) into \(R\), and let \(\gcon\) be the star centered at \(z\) with
leaf set \(R\setminus\{z\}\). Set \(k' \coloneqq k+1\) and \(t\coloneqq n=|\mathcal{U}|\).
Then any vertex subset \(S\subseteq R\) containing \(z\) induces a connected subgraph in \(\gcon\).

The rest of the reduction is analogous to the clique case; we include it for completeness.
For any subfamily \(\mathcal{F}'\subseteq \mathcal{F}\) with \(|\mathcal{F}'|=k\), define
\[
R' \;\coloneqq\; \{z\}\ \cup\ \{v_F : F\in \mathcal{F}'\}\ \subseteq\ R .
\]
By construction of \(\gcov\), we have
\(|N_{\gcov}(R')| \;=\; \bigl|\bigcup_{F\in \mathcal{F}'} F\bigr|\),
and since \(z\in R'\), the subgraph \(\gcon[R']\) is connected. Conversely, any
\(S\subseteq R\) with \(z\in S\) and \(|S|=k+1\) corresponds to a \(k\)-subfamily
\(\mathcal{F}'\) obtained from \(S\setminus\{z\}\), with the same coverage size in \(\gcov\).

Therefore, a \((1,\,1-e^{-1}+\varepsilon)\)-approximation for \pcrbdsshort\ on the instance
\(\cI=(\gcon,\gcov,k',t)\) (with \(\gcon\) a star) would solve
\textsc{Max\((\beta,k)\)-Coverage} on \((\mathcal{U},\mathcal{F})\) for
\(\beta=1-e^{-1}+\varepsilon\). By \Cref{prop:cover}, this implies the lower bound in
\Cref{thm:hardtapproxpds} also holds when the connectivity graph is a star.	This completes the proof.
\end{proof}

\subsection{Hardness of Approximating  Budgeted Connected Dominating Set}
In this section we observe that even a \emph{graphical} variant of \pcrbdsshort\ is hard to approximate.
In this variant the input is a single graph \(G\) and integers \(k,t\), and the goal is to decide
whether there exists a \(k\)-vertex set \(S\subseteq V(G)\) such that \(G[S]\) is connected and
\(|N_G[S]|\ge t\). In literature this problem is called {\sc Budgeted Connected Dominating Set} (\bcdshort)~\cite{DBLP:journals/siamdm/KhullerPS20}. 

This can be realized as \pcrbdsshort\ with tightly coupled coverage and connectivity graphs:
set \(\gcon := G\) and define \(\gcov\) as follows. Create two copies of \(V(G)\), denoted
\(R=\{v^R:v\in V(G)\}\) and \(B=\{v^B:v\in V(G)\}\). For each \(u\in V(G)\), make \(u^R\) adjacent
in \(\gcov\) to every \(v^B\) with \(v\in N_G[u]\) (the closed neighborhood in \(G\)).
Under this construction, for any \(S_0\subseteq V(G)\) and its copy \(S=S_0^R\subseteq R\),
we have \(N_{\gcov}(S)=\{v^B: v\in N_G[S_0]\}\), so \(|N_{\gcov}(S)|=|N_G[S_0]|\).
Thus the graphical variant is exactly the special case of \pcrbdsshort\ with \(\gcon=G\) and
\(\gcov\) defined as above.



We derive our hardness via the reduction underlying the polynomial-time
hardness-of-approximation for \bcdshort~\cite{DBLP:journals/siamdm/KhullerPS20}, which
reduces from \textsc{Max}\((\beta,k)\)\textsc{-Coverage}.  For completeness, we present the full proof.

\begin{theorem}\label{theo:relatecov}
    Assuming {\sf ETH}, for any $\varepsilon >0$ in \fpt time  it is
impossible to distinguish between the following instances $(G,k,t)$ of the \bcdshort 

   \begin{itemize}
   \setlength{\itemsep}{-2pt}
    \item Type 1:  There exists a set of $k$ vertices $S$ which dominates at least $t$ vertices  such that $G[S]$ is connected.

    \item Type 2: Any $k$ subsets from $V(G)$ dominates at most  $(1-e^{-1}+2\varepsilon)t$ vertices such that $G[S]$ is connected.
\end{itemize}
Moreover, assuming {\sf ETH}, any algorithm  to distinguish between the above instances  of the \bcdshort   must take runtime $\Omega_k(|\mathcal{F}|^{k^{\Omega(1)}})$.

\end{theorem}

\begin{proof}
Given an instance \((\mathcal{U},\mathcal{F})\) of \textsc{Max}\((\beta,k)\)\textsc{-Coverage} with
\(|\mathcal{U}|=n\) and \(\beta=1-e^{-1}+\varepsilon\), we construct an instance \((G,k',t)\) of
\bcdshort as follows. Let
\[
M \;\coloneqq\; \frac{|\mathcal{F}|}{\varepsilon n}\quad\text{(assume \(M\in\mathbb{N}\); otherwise take \(M=\lceil |\mathcal{F}|/(\varepsilon n)\rceil\)).}
\]
Define the vertex sets
\[
V_{\mathrm{comp}} \;=\; \{\,v_F : F\in\mathcal{F}\,\},
\qquad
V_{\mathrm{ind}} \;=\; \{\,u_i : u\in\mathcal{U},\ i\in[M]\,\},
\]
and set \(V(G) \coloneqq V_{\mathrm{comp}}\uplus V_{\mathrm{ind}}\).
The edge set is
\[
E(G) \;\coloneqq\; \{\, v_F u_i : u\in F,\ i\in[M]\,\} \;\cup\; \{\, v_F v_{F'} : F,F'\in\mathcal{F},\ F\neq F'\,\},
\]
i.e., \(G[V_{\mathrm{comp}}]\) is a clique and each \(v_F\) is adjacent to the \(M\) copies
of every element \(u\in F\).
Finally, set \(k'\coloneqq k\) and \(t\coloneqq Mn\).

We show that any procedure that distinguishes \emph{Type~1} from \emph{Type~2} instances of
\bcdshort (for the above thresholds) also distinguishes \emph{Type~I} from \emph{Type~II}
instances of \textsc{Max}\((\beta,k)\)\textsc{-Coverage}. By \Cref{prop:cover}, this yields
\Cref{theo:relatecov}.

\paragraph*{(\(\text{\bcdshort Type~2} \Rightarrow \text{Coverage Type~II}\)).}
Assume, for contradiction, that the constructed \textsc{BCDS} instance is Type~2 but
\((\mathcal{U},\mathcal{F})\) is Type~I. Then there exists \(\mathcal{F}'\subseteq\mathcal{F}\)
with \(|\mathcal{F}'|=k\) and \(\bigl|\bigcup_{F\in\mathcal{F}'} F\bigr|=n\).
Let \(S\coloneqq \{v_F : F\in\mathcal{F}'\}\subseteq V_{\mathrm{comp}}\).
Since \(G[V_{\mathrm{comp}}]\) is a clique, \(G[S]\) is connected, and every \(u\in\mathcal{U}\)
is adjacent to each \(v_F\) with \(u\in F\), the set \(S\) dominates all of \(V_{\mathrm{ind}}\),
so \(|N_G[S]|\ge |V_{\mathrm{ind}}|=Mn=t\). Hence the \bcdshort instance is Type~1,
a contradiction.

\paragraph*{(\(\text{Coverage Type~II} \Rightarrow \text{\bcdshort Type~2}\)).}
Assume \((\mathcal{U},\mathcal{F})\) is Type~II (so no \(k\)-subfamily covers
\((1-e^{-1}+\varepsilon)n\) elements), and suppose, for contradiction, that the constructed
\bcdshort instance is \emph{not} Type~2. Then there exists \(S\subseteq V(G)\) with \(|S|\le k\),
\(G[S]\) connected, and
\[
|N_G[S]| \;>\; (1-e^{-1}+2\varepsilon)\,Mn.
\]
We build a \(k\)-subfamily \(\mathcal{F}''\subseteq \mathcal{F}\) as follows:
(i) include every \(F\) with \(v_F\in S\cap V_{\mathrm{comp}}\);
(ii) for each \(u_i\in S\cap V_{\mathrm{ind}}\), if possible add some \(F\) with \(v_F\notin S\) and \(u\in F\).
Clearly \(|\mathcal{F}''|\le |S|\le k\).
Each dominated vertex in \(V_{\mathrm{ind}}\) corresponds to one of the \(M\) copies of some
covered element \(u\in\mathcal{U}\), while we may overcount by at most \(|\mathcal{F}|\) due to
neighbors inside \(V_{\mathrm{comp}}\). Therefore the total coverage of \(\mathcal{F}''\) exceeds
\[
\frac{(1-e^{-1}+2\varepsilon)Mn - |\mathcal{F}|}{M}
\;\;=\;\; (1-e^{-1}+\varepsilon)\,n,
\]
contradicting that \((\mathcal{U},\mathcal{F})\) is Type~II.

\noindent
This completes the reduction and the proof.
\end{proof}

As a standard corollary—translating the gap reduction into a hardness-of-approximation lower bound—we obtain the following from \Cref{theo:relatecov}.

\begin{corollary}
\label{cor:fpthardnessbcds}
Assuming {\sf ETH}, there is no \fpt\ algorithm parameterized by \(k\) that achieves a
\((1,\,1-e^{-1}+\varepsilon)\)-approximation for \bcdshort.
Moreover, under {\sf ETH}, for any fixed \(\varepsilon\in(0,1)\), every
\((1,\,1-e^{-1}+\varepsilon)\)-approximation algorithm for \bcdshort must run in time
\(\Omega_k\!\bigl(|\mathcal{F}|^{\,k^{\Omega(1)}}\bigr)\). 
\end{corollary}

 \subsection{$(g(k),1)$-approximation  for \pcrbdsshort}

In this section we rule out the possibility of a \((g(k),1)\)-approximation for \pcrbdsshort,
for any computable function \(g\), when parameterized by \(k\), under standard complexity
assumptions. Our hardness will follow from a reduction from the classical
\textsc{Dominating Set} problem.

\paragraph*{Dominating Set.}
Given a graph \(G\) and an integer \(k\), decide whether there exists a dominating set
\(D\subseteq V(G)\) with \(|D|\le k\), i.e., \(N_G[D]=V(G)\).
An algorithm is an \(\alpha(k)\)-\fpt\ approximation for \textsc{Dominating Set} if it runs in
\fpt\ time parameterized by \(k\) and, whenever \(G\) has a dominating set of size at most \(k\),
it outputs a dominating set of size at most \(\alpha(k)\cdot k\).
The following hardness results for \textsc{Dominating Set} are known.

 \begin{proposition}{\rm \cite[Theorem 1.3,1.4,1.5]{DBLP:journals/jacm/SLM19}}\label{prop:gcover}
 \begin{enumerate}
 \item Assuming $\mathsf{W[1]} \neq \mathsf{FPT}$,  no $\mathsf{FPT}$ time algorithm can approximate {\sc Dominating Set}  to within a
factor of $(\log n)^{1/{\mathsf{poly}}(k)}$.
     \item  Assuming {\sf ETH}, no $f(k) \cdot n^{o(k)}$
-time algorithm can approximate {\sc Dominating Set} to within a
factor of $(\log n)^{1/{\mathsf{poly}}(k)}$.
\item There is a function $h : \mathbb{R}^+ \to \mathbb{N}$ such that, assuming {\sf SETH}, for every integer $k \geq  2$
and for every $\varepsilon >0$, no $\mathcal{O}(n^{k - \varepsilon})$-time algorithm can approximate {\sc Dominating Set}  to within a
factor of $(\log n)^{1/{\mathsf{poly}}(k,h(\varepsilon))}$.
\end{enumerate}
 \end{proposition}

 We now show that analogous hardness result  for \pcrbdsshort.

  \begin{sloppypar}

\begin{restatable}{theorem}{hardkapproxpds}  \label{thm:hardkapproxpds}
	The following statements hold even when the connectivity graph $\gcon$ is restricted to be either a clique or a star.
\begin{enumerate}
 \item Assuming $\mathsf{W[1]} \neq \mathsf{FPT}$,  no $f(k) \cdot n^{\mathcal{O}(1)}$ time algorithm can give  $((\log n)^{1/{\mathsf{poly}}(k)},1)$-approximation for \pcrbdsshort.
     \item  Assuming {\sf ETH}, no $f(k) \cdot n^{o(k)}$
-time algorithm can give  $((\log n)^{1/{\mathsf{poly}}(k)},1)$-approximation for \pcrbdsshort.
\item There is a function $h : \mathbb{R}^+ \to \mathbb{N}$ such that, assuming {\sf SETH}, for every integer $k \geq  2$
and for every $\varepsilon >0$, no $\mathcal{O}(n^{k - \varepsilon})$-time algorithm can give  $(\log n)^{1/{\mathsf{poly}}(k,h(\varepsilon))},1)$-approximation for \pcrbdsshort. 
\end{enumerate}
\end{restatable}
\end{sloppypar}

\begin{proof}
	We reduce from \textsc{Dominating Set}. Let \((G,k)\) be an instance with \(n=|V(G)|\).
We construct an instance \(\cI=(\gcon,\gcov,k,n)\) of \pcrbdsshort\ as follows.

	$$
	V(\gcov) \leftarrow R = \{ v : v \in V(G) \} \uplus B = \{ u : u \in V(G) \},
	$$
	$$
	E(\gcov) \leftarrow \{ vu  : v \in R, u \in B, u \in N_G[v] \},
	$$

 \begin{description}
     \item[(i) $\gcon$ is a clique:] 	Let $\gcon$ be the complete graph induced on $R$. Then any subset of $R$ induces a connected subgraph in $\gcon$.   Suppose $G$ admits a dominating set $D \subseteq V(G)$ of size $\ell \leq k$. Then, the corresponding vertices $R' = \{ v : v \in D \} \subseteq R$ satisfy
	\[
	|R'| = \ell, 
	\qquad |N_{\gcov}(R')| = n,
	\qquad \text{and } \gcon[R'] \text{ is connected}.
	\]
	Conversely, any solution $R'$ of \pcrbdsshort{}  of size $\ell$ with $|N_{\gcov}(R')| = n$ yields a dominating set $D \subseteq V(G)$ of size $\ell$. 
    Hence  any $(\alpha, 1)$-approximation algorithm for \pcrbdsshort on  $\cI = (\gcon, \gcov,  k, n)$  would yield an $\alpha$-approximation for {\sc $k$-DomSet} in $G$. Thus, by \Cref{prop:gcover}, \Cref{thm:hardkapproxpds} for the case $\gcon$ is a clique  holds. 

   \item[(ii) $\gcon$ is a star:] Here we give a {\em turing reduction}. In essence we create $n$ many corresponding connectivity graph. 	Let $V(G) = \{ v_1, v_2, \ldots, v_n \}$. For each $i \in [n]$, construct a connectivity graph $\gcon^i$ as a star centered at $v_i \in R$, with all other vertices of $R \setminus \{v_i\}$ as leaves.  Observe that $G$ has a dominating set $D \subseteq V(G)$ of size $\ell$ if and only if there exists some $i \in [n]$ such that the set $	R'' = \{v_i\} \cup \{ v : v \in D \} \subseteq R $
	satisfies
	\[
	|R''| = \ell, 
	\qquad |N_{\gcov}(R'')| = n, 
	\qquad \gcon^i[R''] \text{ is connected}.
	\]

	Now, for each $i \in [n]$, run the $(\alpha,1)$-approximation algorithm for \pcrbdsshort{} on $\cI^i = (\gcon^i, \gcov, k, n)$, and obtain a solution $O_i$. Let $
	O = \arg\min_{i \in [n]} |O_i|.$
	Then, $O$ yields an $\alpha$-factor approximation for {\sc $k$-DomSet} in $G$.  Thus, by \Cref{prop:gcover}, \Cref{thm:hardkapproxpds} for the case $\gcon$ is a star  holds.
 \end{description}
 
 \end{proof}
Note that a factor of the form \(\bigl((\log n)^{1/\mathrm{poly}(k)},\,1\bigr)\) is strictly stronger than any \(\bigl(g(k),\,1\bigr)\) factor. Indeed, suppose we have a \(\bigl(g(k),1\bigr)\)-\fpt-approximation for some computable \(g\). We obtain a \(\bigl((\log n)^{1/\mathrm{poly}(k)},1\bigr)\)-approximation by using the smaller of the two ratios:
\begin{itemize}[itemsep=0pt,leftmargin=*]
  \item If \(g(k)\le (\log n)^{1/\mathrm{poly}(k)}\), run the \(\bigl(g(k),1\bigr)\)-\fpt algorithm.
  \item Otherwise \(g(k)>(\log n)^{1/\mathrm{poly}(k)}\), so \(\log n < g(k)^{\mathrm{poly}(k)}\) and hence
        \(n \le \exp\bigl(g(k)^{\mathrm{poly}(k)}\bigr)\). Then brute-force search (e.g., \(n^k\) enumeration)
        runs in time \(n^k = \exp\bigl(k\log n\bigr) \le \exp\bigl(k\,g(k)^{\mathrm{poly}(k)}\bigr)\),
        which is \fpt\ in \(k\).
\end{itemize}

\subsection{{\sc ConnPHS} parameterized by $k$} \label{sec:lbconpvc}
In {\sc Connected Partial Vertex Cover (ConnPVC)}, we are 
given a graph \(G\) and integers \(k,t\ge 0\), and the objective is to decide whether there exists a set
\(S\subseteq V(G)\) with \(|S|\le k\) such that \(G[S]\) is connected and \(S\) covers at least
\(t\) edges of \(G\).

Observe that \textsc{ConnPVC} is a special case of \textsc{ConnPHS}, obtained when every set has
size exactly two. Consequently, every lower bound for \textsc{ConnPVC} immediately holds for
\textsc{ConnPHS}.



\begin{theorem}\label{theo:whard}
    {\sc ConnPVC} is \woh when parameterized by $k$.
\end{theorem}
\begin{proof}
We give a parameterized reduction from \textsc{Partial Vertex Cover} (PVC) to \textsc{ConnPVC}.
In PVC, the input is a graph \(G\) with integers \(k,t\ge 0\); the task is to decide whether there
exists \(S\subseteq V(G)\) with \(|S|\le k\) that covers at least \(t\) edges of \(G\).
PVC is known to be \(\mathsf{W[1]}\)-hard~\cite{DBLP:journals/mst/GuoNW07}.

\medskip
\noindent\textbf{Reduction.}
Given an instance \((G=(V,E),k,t)\) of PVC with \(|V|=n\) and \(|E|=m\), construct
\((G',k',t')\) for \textsc{ConnPVC} as follows:
add a new vertex \(z\) adjacent to every vertex of \(G\), then add \(m\) new degree-1 vertices
and make each of them adjacent only to \(z\).
Set \(k'\coloneqq k+1\) and \(t'\coloneqq m+n+t\).

\medskip
\noindent\textbf{Forward direction.}
If \(S\subseteq V(G)\) is a PVC solution for \((G,k,t)\), then
\(S'\coloneqq S\cup\{z\}\) satisfies \(|S'|\le k'\) and \(G'[S']\) is connected (since \(z\) is
adjacent to every vertex in \(G\)).
Moreover, \(z\) alone covers all \(m+n\) new edges, and \(S\) covers at least \(t\) edges of \(G\),
so \(S'\) covers at least \(m+n+t=t'\) edges.
Hence \((G',k',t')\) is a \yes-instance of \textsc{ConnPVC}.

\medskip
\noindent\textbf{Backward direction.}
Let \(S'\) be a solution for \((G',k',t')\) with \(|S'|\le k'\), \(G'[S']\) connected, and
\(|N_{G'}[S']|\ge t'=m+n+t\).
If \(z\notin S'\), then with at most \(k'\) vertices one cannot cover all \(m+n\) edges incident to
\(z\) (each such edge requires selecting either \(z\) itself or its unique other endpoint), which
precludes reaching the threshold \(m+n+t\).
Therefore \(z\in S'\).
Since \(z\) covers exactly the \(m+n\) new edges, the remaining vertices
\(S\coloneqq S'\setminus\{z\}\subseteq V(G)\) must cover at least \(t\) edges of \(G\).
We also have \(|S|\le k\).
Thus \((G,k,t)\) is a \yes-instance of PVC.

\noindent This completes the reduction.
\end{proof}

\section{Conclusion} \label{sec:conclusion}

\begin{sloppypar}

We introduced \pcrbdsfull\ (\pcrbdsshort) as a unifying model for many connectivity-constrained coverage problems that have been studied in prior polynomial-time approximation literature. The expressivity of our model comes from having two separate graphs (i) $\gcon$ for imposing structural properties on the solution (in our case, connectivity), and (ii) $\gcov$ for modeling hypergraph incidences (which we use for requiring certain amount of coverage). We hope that this will serve as a model for understanding a variety of problems from this two-layer perspective.

As for our algorithmic contributions, we began with an exact \fpt algorithm parameterized by $t$. In the realm of \fpt approximation, we focused on biclique-free instances, and designed an EPAS that finds a size-$k$ connected set that dominates at least
$(1-\varepsilon)t$ blue vertices; and a complementary PAS that finds a connected set of size at most $(1+\varepsilon)k$  that covers $t$ blue vertices. Our key tool is a small family of \emph{surrogate weight functions} built via \emph{neighborhood
sparsification} procedure, which lets us reduce the search to maximum-weight $k$-trees in $G_{\mathrm{conn}}$. Among the lower bound results, the main takeaway is that the problem inherits strong $\mathsf{W}$-hardness and approximation lower bounds from {\sc Max Coverage}, when $\gcov$ is arbitrary. Together, these results mark a clear line between what is possible with \fpt approximation and what
is not, and point to a clean classification by the structure of the coverage and constraint graphs.

\paragraph*{Future directions.}
\begin{itemize}[leftmargin=*,itemsep=-2pt]
  \item \textbf{Faster algorithms.} Improve the running time (and parameter dependence) of our EPAS/PAS.
  
  \item \textbf{Broader tractable classes.} Identify larger families of coverage instances where the problem remains FPT (beyond biclique-free), e.g. bounded VC-dimension, or geometric incidences.
  \item \textbf{Other constraints.} Study variants where the constraint layer enforces independence/packing, degree bounds, or fault tolerance (e.g., $r$-connectivity) instead of (or in addition to) connectivity.
  \item \textbf{Lossy kernels.} It would be ideal to have lossy kernels as in the vanilla setting~\cite{DBLP:journals/tcs/Manurangsi25}; however, we can show {\sc Connected Partial Vertex Cover} admits no lossy kernels. 
\end{itemize}

\end{sloppypar}


\bibliographystyle{plain}
\bibliography{reference}

\begin{thebibliography}{10}

\bibitem{AlonYZ95}
Noga Alon, Raphael Yuster, and Uri Zwick.
\newblock Color-coding.
\newblock {\em J. {ACM}}, 42(4):844--856, 1995.

\bibitem{DBLP:conf/stoc/BafnaSM25}
Mitali Bafna, {Karthik {C. S.}}, and Dor Minzer.
\newblock Near optimal constant inapproximability under {ETH} for fundamental
  problems in parameterized complexity.
\newblock In Michal Kouck{\'{y}} and Nikhil Bansal, editors, {\em Proceedings
  of the 57th Annual {ACM} Symposium on Theory of Computing, {STOC} 2025,
  Prague, Czechia, June 23-27, 2025}, pages 2118--2129. {ACM}, 2025.

\bibitem{DBLP:conf/stoc/BenczurK96}
Andr{\'{a}}s~A. Bencz{\'{u}}r and David~R. Karger.
\newblock Approximating \emph{s-t} minimum cuts in
  \emph{{\~{O}}}(\emph{n}\({}^{2}\)) time.
\newblock In Gary~L. Miller, editor, {\em Proceedings of the Twenty-Eighth
  Annual {ACM} Symposium on the Theory of Computing, Philadelphia,
  Pennsylvania, USA, May 22-24, 1996}, pages 47--55. {ACM}, 1996.

\bibitem{DBLP:journals/siamcomp/BeyerH80}
Terry Beyer and Sandra~Mitchell Hedetniemi.
\newblock Constant time generation of rooted trees.
\newblock {\em {SIAM} J. Comput.}, 9(4):706--712, 1980.

\bibitem{DBLP:conf/iwpec/CaiCC06}
Leizhen Cai, Siu~Man Chan, and Siu~On Chan.
\newblock Random separation: {A} new method for solving fixed-cardinality
  optimization problems.
\newblock In Hans~L. Bodlaender and Michael~A. Langston, editors, {\em
  Parameterized and Exact Computation, Second International Workshop, {IWPEC}
  2006, Z{\"{u}}rich, Switzerland, September 13-15, 2006, Proceedings}, volume
  4169 of {\em Lecture Notes in Computer Science}, pages 239--250. Springer,
  2006.

\bibitem{ChalermsookCKLM17}
Parinya Chalermsook, Marek Cygan, Guy Kortsarz, Bundit Laekhanukit, Pasin
  Manurangsi, Danupon Nanongkai, and Luca Trevisan.
\newblock From gap-exponential time hypothesis to fixed parameter tractable
  inapproximability: Clique, dominating set, and more.
\newblock {\em {SIAM} J. Comput.}, 49(4):772--810, 2020.

\bibitem{DBLP:journals/jcss/ChenHKX06}
Jianer Chen, Xiuzhen Huang, Iyad~A. Kanj, and Ge~Xia.
\newblock Strong computational lower bounds via parameterized complexity.
\newblock {\em J. Comput. Syst. Sci.}, 72(8):1346--1367, 2006.

\bibitem{ChenL19}
Yijia Chen and Bingkai Lin.
\newblock The constant inapproximability of the parameterized dominating set
  problem.
\newblock {\em {SIAM} J. Comput.}, 48(2):513--533, 2019.

\bibitem{DBLP:conf/icalp/Cohen-AddadG0LL19}
Vincent Cohen{-}Addad, Anupam Gupta, Amit Kumar, Euiwoong Lee, and Jason Li.
\newblock {Tight {FPT} Approximations for $k$-Median and $k$-Means}.
\newblock In Christel Baier, Ioannis Chatzigiannakis, Paola Flocchini, and
  Stefano Leonardi, editors, {\em 46th International Colloquium on Automata,
  Languages, and Programming, {ICALP} 2019, July 9-12, 2019, Patras, Greece},
  volume 132 of {\em LIPIcs}, pages 42:1--42:14. Schloss Dagstuhl -
  Leibniz-Zentrum f{\"{u}}r Informatik, 2019.

\bibitem{cygan2015parameterized}
Marek Cygan, Fedor~V. Fomin, {\L}ukasz Kowalik, Daniel Lokshtanov, D{\'a}niel
  Marx, Marcin Pilipczuk, Micha{\l} Pilipczuk, and Saket Saurabh.
\newblock {\em Parameterized algorithms}, volume~5.
\newblock Springer, 2015.

\bibitem{DBLP:books/sp/CyganFKLMPPS15}
Marek Cygan, Fedor~V. Fomin, Lukasz Kowalik, Daniel Lokshtanov, D{\'{a}}niel
  Marx, Marcin Pilipczuk, Michal Pilipczuk, and Saket Saurabh.
\newblock {\em Parameterized Algorithms}.
\newblock Springer, 2015.

\bibitem{DBLP:conf/ifaamas/DAngeloD25}
Gianlorenzo D'Angelo and Esmaeil Delfaraz.
\newblock Approximation algorithms for connected maximum coverage.
\newblock In Sanmay Das, Ann Now{\'{e}}, and Yevgeniy Vorobeychik, editors,
  {\em Proceedings of the 24th International Conference on Autonomous Agents
  and Multiagent Systems, {AAMAS} 2025, Detroit, MI, USA, May 19-23, 2025},
  pages 538--546. International Foundation for Autonomous Agents and Multiagent
  Systems / {ACM}, 2025.

\bibitem{diestel2025graph}
Reinhard Diestel.
\newblock {\em Graph theory}, volume 173.
\newblock Springer Nature, 2025.

\bibitem{DBLP:conf/stoc/DinurS14}
Irit Dinur and David Steurer.
\newblock Analytical approach to parallel repetition.
\newblock In David~B. Shmoys, editor, {\em Symposium on Theory of Computing,
  {STOC} 2014, New York, NY, USA, May 31 - June 03, 2014}, pages 624--633.
  {ACM}, 2014.

\bibitem{DBLP:journals/siamcomp/DowneyF95}
Rodney~G. Downey and Michael~R. Fellows.
\newblock Fixed-parameter tractability and completeness {I:} basic results.
\newblock {\em {SIAM} J. Comput.}, 24(4):873--921, 1995.

\bibitem{DBLP:conf/stoc/FeigeM06}
Uriel Feige and Mohammad Mahdian.
\newblock Finding small balanced separators.
\newblock In Jon~M. Kleinberg, editor, {\em Proceedings of the 38th Annual
  {ACM} Symposium on Theory of Computing, Seattle, WA, USA, May 21-23, 2006},
  pages 375--384. {ACM}, 2006.

\bibitem{DBLP:journals/jacm/FominLPS16}
Fedor~V. Fomin, Daniel Lokshtanov, Fahad Panolan, and Saket Saurabh.
\newblock Efficient computation of representative families with applications in
  parameterized and exact algorithms.
\newblock {\em J. {ACM}}, 63(4):29:1--29:60, 2016.

\bibitem{DBLP:journals/jal/GandhiKS04}
Rajiv Gandhi, Samir Khuller, and Aravind Srinivasan.
\newblock Approximation algorithms for partial covering problems.
\newblock {\em J. Algorithms}, 53(1):55--84, 2004.

\bibitem{DBLP:conf/wg/GolovachV08}
Petr~A. Golovach and Yngve Villanger.
\newblock Parameterized complexity for domination problems on degenerate
  graphs.
\newblock In Hajo Broersma, Thomas Erlebach, Tom Friedetzky, and Dani{\"{e}}l
  Paulusma, editors, {\em Graph-Theoretic Concepts in Computer Science, 34th
  International Workshop, {WG} 2008, Durham, UK, June 30 - July 2, 2008.
  Revised Papers}, volume 5344 of {\em Lecture Notes in Computer Science},
  pages 195--205, 2008.

\bibitem{DBLP:journals/mst/GuoNW07}
Jiong Guo, Rolf Niedermeier, and Sebastian Wernicke.
\newblock Parameterized complexity of vertex cover variants.
\newblock {\em Theory Comput. Syst.}, 41(3):501--520, 2007.

\bibitem{DBLP:conf/soda/GuptaLLM019}
Anupam Gupta, Euiwoong Lee, Jason Li, Pasin Manurangsi, and Michal Wlodarczyk.
\newblock Losing treewidth by separating subsets.
\newblock In Timothy~M. Chan, editor, {\em Proceedings of the Thirtieth Annual
  {ACM-SIAM} Symposium on Discrete Algorithms, {SODA} 2019, San Diego,
  California, USA, January 6-9, 2019}, pages 1731--1749. {SIAM}, 2019.

\bibitem{DBLP:conf/stoc/GuruswamiLRS024}
Venkatesan Guruswami, Bingkai Lin, Xuandi Ren, Yican Sun, and Kewen Wu.
\newblock Parameterized inapproximability hypothesis under exponential time
  hypothesis.
\newblock In Bojan Mohar, Igor Shinkar, and Ryan O'Donnell, editors, {\em
  Proceedings of the 56th Annual {ACM} Symposium on Theory of Computing, {STOC}
  2024, Vancouver, BC, Canada, June 24-28, 2024}, pages 24--35. {ACM}, 2024.

\bibitem{DBLP:conf/stoc/GuruswamiLRS025}
Venkatesan Guruswami, Bingkai Lin, Xuandi Ren, Yican Sun, and Kewen Wu.
\newblock {Almost Optimal Time Lower Bound for Approximating Parameterized
  Clique, CSP, and More, under {ETH}}.
\newblock In Michal Kouck{\'{y}} and Nikhil Bansal, editors, {\em Proceedings
  of the 57th Annual {ACM} Symposium on Theory of Computing, {STOC} 2025,
  Prague, Czechia, June 23-27, 2025}, pages 2136--2144. {ACM}, 2025.

\bibitem{hochbaum1998analysis}
Dorit~S Hochbaum and Anu Pathria.
\newblock Analysis of the greedy approach in problems of maximum $k$-coverage.
\newblock {\em Naval Research Logistics (NRL)}, 45(6):615--627, 1998.

\bibitem{DBLP:journals/ipl/HochbaumR20}
Dorit~S. Hochbaum and Xu~Rao.
\newblock Approximation algorithms for connected maximum coverage problem for
  the discovery of mutated driver pathways in cancer.
\newblock {\em Inf. Process. Lett.}, 158:105940, 2020.

\bibitem{DBLP:conf/icalp/00020LS0U24}
Tanmay Inamdar, Pallavi Jain, Daniel Lokshtanov, Abhishek Sahu, Saket Saurabh,
  and Anannya Upasana.
\newblock Satisfiability to coverage in presence of fairness, matroid, and
  global constraints.
\newblock In Karl Bringmann, Martin Grohe, Gabriele Puppis, and Ola Svensson,
  editors, {\em 51st International Colloquium on Automata, Languages, and
  Programming, {ICALP} 2024, July 8-12, 2024, Tallinn, Estonia}, volume 297 of
  {\em LIPIcs}, pages 88:1--88:18. Schloss Dagstuhl - Leibniz-Zentrum f{\"{u}}r
  Informatik, 2024.

\bibitem{DBLP:conf/soda/0001KPSS0U23}
Pallavi Jain, Lawqueen Kanesh, Fahad Panolan, Souvik Saha, Abhishek Sahu, Saket
  Saurabh, and Anannya Upasana.
\newblock {Parameterized Approximation Schemes for Biclique-Free Max $k$-Weight
  SAT and Max Coverage}.
\newblock {\em ACM Trans. Algorithms}, August 2025.
\newblock Just Accepted.

\bibitem{DBLP:journals/jacm/SLM19}
{Karthik {C. S.}}, Bundit Laekhanukit, and Pasin Manurangsi.
\newblock On the parameterized complexity of approximating dominating set.
\newblock {\em J. {ACM}}, 66(5):33:1--33:38, 2019.

\bibitem{DBLP:journals/siamdm/KhullerPS20}
Samir Khuller, Manish Purohit, and Kanthi~K. Sarpatwar.
\newblock Analyzing the optimal neighborhood: Algorithms for partial and
  budgeted connected dominating set problems.
\newblock {\em {SIAM} J. Discret. Math.}, 34(1):251--270, 2020.

\bibitem{DBLP:journals/tcs/LamprouSZ21}
Ioannis Lamprou, Ioannis Sigalas, and Vassilis Zissimopoulos.
\newblock Improved budgeted connected domination and budgeted edge-vertex
  domination.
\newblock {\em Theor. Comput. Sci.}, 858:1--12, 2021.

\bibitem{Lin18}
Bingkai Lin.
\newblock The parameterized complexity of the {$k$}-biclique problem.
\newblock {\em J. {ACM}}, 65(5):34:1--34:23, 2018.

\bibitem{Lin19}
Bingkai Lin.
\newblock A simple gap-producing reduction for the parameterized set cover
  problem.
\newblock In Christel Baier, Ioannis Chatzigiannakis, Paola Flocchini, and
  Stefano Leonardi, editors, {\em 46th International Colloquium on Automata,
  Languages, and Programming, {ICALP} 2019, July 9-12, 2019, Patras, Greece},
  volume 132 of {\em LIPIcs}, pages 81:1--81:15. Schloss Dagstuhl -
  Leibniz-Zentrum f{\"{u}}r Informatik, 2019.

\bibitem{DBLP:conf/soda/LokshtanovS0S025}
Daniel Lokshtanov, Abhishek Sahu, Saket Saurabh, Vaishali Surianarayanan, and
  Jie Xue.
\newblock Parameterized approximation for capacitated \emph{d}-hitting set with
  hard capacities.
\newblock In Yossi Azar and Debmalya Panigrahi, editors, {\em Proceedings of
  the 2025 Annual {ACM-SIAM} Symposium on Discrete Algorithms, {SODA} 2025, New
  Orleans, LA, USA, January 12-15, 2025}, pages 1565--1592. {SIAM}, 2025.

\bibitem{DBLP:journals/siamcomp/LokshtanovSS24}
Daniel Lokshtanov, Saket Saurabh, and Vaishali Surianarayanan.
\newblock A parameterized approximation scheme for min {$k$}-cut.
\newblock {\em {SIAM} J. Comput.}, 53(6):S20--205, 2024.

\bibitem{DBLP:conf/soda/Manurangsi19}
Pasin Manurangsi.
\newblock {A Note on Max $k$-Vertex Cover: Faster FPT-AS, Smaller Approximate
  Kernel and Improved Approximation}.
\newblock In Jeremy~T. Fineman and Michael Mitzenmacher, editors, {\em 2nd
  Symposium on Simplicity in Algorithms, {SOSA} 2019, January 8-9, 2019, San
  Diego, CA, {USA}}, volume~69 of {\em OASIcs}, pages 15:1--15:21. Schloss
  Dagstuhl - Leibniz-Zentrum f{\"{u}}r Informatik, 2019.

\bibitem{DBLP:conf/soda/Manurangsi20}
Pasin Manurangsi.
\newblock Tight running time lower bounds for strong inapproximability of
  maximum \emph{k}-coverage, unique set cover and related problems (via
  \emph{t}-wise agreement testing theorem).
\newblock In Shuchi Chawla, editor, {\em Proceedings of the 2020 {ACM-SIAM}
  Symposium on Discrete Algorithms, {SODA} 2020, Salt Lake City, UT, USA,
  January 5-8, 2020}, pages 62--81. {SIAM}, 2020.

\bibitem{DBLP:journals/tcs/Manurangsi25}
Pasin Manurangsi.
\newblock Improved {FPT} approximation scheme and approximate kernel for
  biclique-free max $k$-weight {SAT:} greedy strikes back.
\newblock {\em Theor. Comput. Sci.}, 1028:115033, 2025.

\bibitem{DBLP:journals/cj/Marx08}
D{\'{a}}niel Marx.
\newblock Parameterized complexity and approximation algorithms.
\newblock {\em Comput. J.}, 51(1):60--78, 2008.

\bibitem{DBLP:journals/jco/MisraPRSS12}
Neeldhara Misra, Geevarghese Philip, Venkatesh Raman, Saket Saurabh, and
  Somnath Sikdar.
\newblock {FPT} algorithms for connected feedback vertex set.
\newblock {\em J. Comb. Optim.}, 24(2):131--146, 2012.

\bibitem{otter1948number}
Richard Otter.
\newblock The number of trees.
\newblock {\em Annals of Mathematics}, 49(3):583--599, 1948.

\bibitem{DBLP:conf/wg/PlehnV90}
J{\"{u}}rgen Plehn and Bernd Voigt.
\newblock Finding minimally weighted subgraphs.
\newblock In Rolf~H. M{\"{o}}hring, editor, {\em Graph-Theoretic Concepts in
  Computer Science, 16rd International Workshop, {WG} '90, Berlin, Germany,
  June 20-22, 1990, Proceedings}, volume 484 of {\em Lecture Notes in Computer
  Science}, pages 18--29. Springer, 1990.

\bibitem{DBLP:conf/esa/Sellier23}
Fran{\c{c}}ois Sellier.
\newblock Parameterized matroid-constrained maximum coverage.
\newblock In Inge~Li G{\o}rtz, Martin Farach{-}Colton, Simon~J. Puglisi, and
  Grzegorz Herman, editors, {\em 31st Annual European Symposium on Algorithms,
  {ESA} 2023, September 4-6, 2023, Amsterdam, The Netherlands}, volume 274 of
  {\em LIPIcs}, pages 94:1--94:16. Schloss Dagstuhl - Leibniz-Zentrum f{\"{u}}r
  Informatik, 2023.

\bibitem{DBLP:journals/jair/SkowronF17}
Piotr Skowron and Piotr Faliszewski.
\newblock Chamberlin-courant rule with approval ballots: Approximating the
  maxcover problem with bounded frequencies in {FPT} time.
\newblock {\em J. Artif. Intell. Res.}, 60:687--716, 2017.

\bibitem{DBLP:journals/siamcomp/SpielmanT11}
Daniel~A. Spielman and Shang{-}Hua Teng.
\newblock Spectral sparsification of graphs.
\newblock {\em {SIAM} J. Comput.}, 40(4):981--1025, 2011.

\bibitem{Wlodarczyk20}
Michal Wlodarczyk.
\newblock Parameterized inapproximability for steiner orientation by gap
  amplification.
\newblock In Artur Czumaj, Anuj Dawar, and Emanuela Merelli, editors, {\em 47th
  International Colloquium on Automata, Languages, and Programming, {ICALP}
  2020, July 8-11, 2020, Saarbr{\"{u}}cken, Germany (Virtual Conference)},
  volume 168 of {\em LIPIcs}, pages 104:1--104:19. Schloss Dagstuhl -
  Leibniz-Zentrum f{\"{u}}r Informatik, 2020.

\end{thebibliography}

\end{document}